\documentclass[11pt]{article}
\usepackage{eurosym}
\usepackage{diagbox}
\usepackage{bookmark}
\usepackage{authblk}
\usepackage{latexsym,amsmath,url,epsfig}
\usepackage{textcomp}
\usepackage{amsfonts,euscript}
\usepackage{psfig}
\usepackage{amsmath}
\usepackage{amssymb}
\usepackage{amsfonts}
\usepackage{amsthm}
\usepackage{mathrsfs}
\usepackage[all]{xy}
\usepackage{epstopdf}
\usepackage{subfigure}
\usepackage{rotating}
\usepackage[title]{appendix}
\usepackage{multirow,rotating}
\usepackage{url}
\usepackage{algorithm}
\usepackage{algorithmicx}
\usepackage{algpseudocode}
\usepackage{subfigure}
\usepackage{setspace}
\usepackage{listings}
\usepackage{color}
\usepackage{tikz}
\usepackage{graphics}
\usepackage[english]{babel}
\usepackage{datetime}
\usepackage{subfigure}
\usepackage{booktabs}
\usepackage{rotfloat}
\usepackage[authoryear]{natbib}
\usepackage{framed}
\usepackage{natbib}
\usepackage{graphicx}
\usepackage{todonotes}

\bookmarksetup{
  numbered,
  addtohook={    \ifnum\bookmarkget{level}>1       \bookmarksetup{numbered=false}    \fi
  },
}
\usetikzlibrary{chains,shapes,decorations,arrows,calc,arrows.meta,fit,positioning}
\tikzset{three sided/.style={
        draw=none,
        append after command={
            [shorten <= -0.5\pgflinewidth]
            ([shift={(-0.5\pgflinewidth,-0.5\pgflinewidth)}]\tikzlastnode.north east)
        edge([shift={( 0.5\pgflinewidth,-0.5\pgflinewidth)}]\tikzlastnode.north west) 
              ([shift={( 0.5\pgflinewidth,-0.5\pgflinewidth)}]\tikzlastnode.north east)
        edge([shift={( 0.5\pgflinewidth,-.95\pgflinewidth)}]\tikzlastnode.south east)  
        
            ([shift={( -0.5\pgflinewidth,-0.5\pgflinewidth)}]\tikzlastnode.south west)
        edge([shift={(0.5\pgflinewidth,-0.5\pgflinewidth)}]\tikzlastnode.south east)
        }
    }
}
\tikzset{three sided2/.style={
        draw=none,
        append after command={
            [shorten <= -0.5\pgflinewidth]
          ([shift={( -0.5\pgflinewidth,-0.5\pgflinewidth)}]\tikzlastnode.north west)
        edge([shift={( -0.5\pgflinewidth,-.95\pgflinewidth)}]\tikzlastnode.south west)  
        
              ([shift={( 0.5\pgflinewidth,-0.5\pgflinewidth)}]\tikzlastnode.north east)
        edge([shift={( 0.5\pgflinewidth,-.95\pgflinewidth)}]\tikzlastnode.south east)  
        
            ([shift={( -0.5\pgflinewidth,-0.5\pgflinewidth)}]\tikzlastnode.south west)
        edge([shift={(0.5\pgflinewidth,-0.5\pgflinewidth)}]\tikzlastnode.south east)
        }
    }
}
\tikzset{three sided3/.style={
        draw=none,
        append after command={
            [shorten <= -0.5\pgflinewidth]
  ([shift={(-0.5\pgflinewidth,-0.5\pgflinewidth)}]\tikzlastnode.north east)
        edge([shift={( 0.5\pgflinewidth,-0.5\pgflinewidth)}]\tikzlastnode.north west) 
          ([shift={( -0.5\pgflinewidth,-0.5\pgflinewidth)}]\tikzlastnode.north west)
        edge([shift={( -0.5\pgflinewidth,-.95\pgflinewidth)}]\tikzlastnode.south west)  
        
              ([shift={( 0.5\pgflinewidth,-0.5\pgflinewidth)}]\tikzlastnode.north east)
        edge([shift={( 0.5\pgflinewidth,-.95\pgflinewidth)}]\tikzlastnode.south east)

        }
    }
}
\graphicspath{{pic/}}
\newtheorem{proposition}{Proposition}
\newtheorem{assumption}{Assumption}
\theoremstyle{definition}

\newtheorem{remark}{Remark}

\algdef{SE}[DOWHILE]{Do}{doWhile}{\algorithmicdo}[1]{\algorithmicwhile\ #1}

\allowdisplaybreaks

\begin{document}

\title{Drift Control of High-Dimensional RBM: A Computational Method Based
on Neural Networks}
\author[1]{Baris Ata\thanks{baris.ata@chicagobooth.edu}}
\author[2]{J. Michael Harrison\thanks{mike.harrison@stanford.edu}}
\author[3]{Nian Si\thanks{niansi@ust.hk}}
\affil[1]{Booth School of Business, University of Chicago}
\affil[2]{Graduate School of Business, Stanford University}
\affil[3]{IEDA, Hong Kong University of Science and Technology}
\date{ }
\maketitle

\begin{abstract}
Motivated by applications in queueing theory, we consider a stochastic
control problem whose state space is the $d$-dimensional positive orthant.
The controlled process $Z$ evolves as a reflected Brownian motion whose
covariance matrix is exogenously specified, as are its directions of
reflection from the orthant's boundary surfaces. A system manager chooses a
drift vector $\theta(t)$ at each time $t$ based on the history of $Z$, and
the cost rate at time $t$ depends on both $Z(t)$ and $\theta(t)$. In our
initial problem formulation, the objective is to minimize expected
discounted cost over an infinite planning horizon, after which we treat the
corresponding ergodic control problem. Extending earlier work by Han et al.
(Proceedings of the National Academy of Sciences, 2018, 8505-8510), we
develop and illustrate a simulation-based computational method that relies
heavily on deep neural network technology. For test problems studied thus
far, our method is accurate to within a fraction of one percent, and is
computationally feasible in dimensions up to at least $d=30$.
\end{abstract}

\affil[1]{Booth School of Business, University of Chicago} %
\affil[2]{Graduate School of Business, Stanford University}

\affil[1]{Booth School of Business, University of Chicago} %
\affil[2]{Stanford Graduate School of Business}

\section{Introduction}

Beginning with the seminal work of %
\citet{iglehart1970multiple,iglehart1970multiple2}, there has developed over
the last 50+ years a large literature that justifies the use of reflected
Brownian motions as approximate models of queueing systems under
\textquotedblleft heavy traffic\textquotedblright\ conditions. In
particular, a limit theorem proved by \citet{reiman1984open} justifies the
use of $d$-dimensional reflected Brownian motion (RBM) as an approximate
model of a $d$-station queueing network. Reiman's theory is restricted to
networks of the generalized Jackson type, also called single-class networks,
or networks with homogeneous customer populations, but it has been extended
to more complex multi-class networks under certain restrictions, most
notably by \citet{peterson1991heavy} and \citet{williams1998diffusion}. The
survey papers by \citet{williams1996approximation} and by %
\citet{harrison1993brownian} provide an overview of heavy traffic limit
theory through its first 25 years.

Many authors have commented on the compactness and simplicity of RBM as a
mathematical model, at least in comparison with the conventional
discrete-flow models that it replaces. For example, in the preface to %
\citet{kushner2001heavy}'s book on heavy traffic analysis one finds the
following passage:

\begin{quote}
``These approximating [Brownian] models have the basic structure of the
original problem, but are significantly simpler. Much inessential detail is
eliminated ... They greatly simplify analysis, design, and optimization,
[yielding] good approximations to problems that would otherwise be
intractable ...''
\end{quote}

Of course, having adopted RBM as a system model, one still confronts the
question of how to do performance analysis, and in that regard there has
been an important recent advance: \citet{blanchet2021efficient} have
developed a simulation-based method to estimate steady-state performance
measures for RBM in dimensions up to 200, and those estimates come with
performance guarantees.

\textbf{Descriptive performance analysis versus optimal control.} Early work
on heavy traffic approximations, including the papers cited above, focused
on descriptive performance analysis under fixed operating policies. %
\citet{harrison1988brownian,harrison2000brownian} expanded the framework to
include consideration of dynamic control, using informal arguments to
justify Brownian approximations for queueing network models where a system
manager can make sequencing, routing and/or input control decisions. Early
papers in that vein by \citet{harrison1989scheduling,harrison1990scheduling}
and by \citet{wein1991brownian} dealt with Brownian models simple enough
that their associated control problems could be solved analytically. But for
larger systems and/or more complex decisions, the Brownian control problem
that approximates an original queueing control problem may only be solvable
numerically. Such stochastic control problems may be of several different
types, depending on context.

At one end of the spectrum are drift control problems, in which the
controlling agent can effect changes in system state only at bounded finite
rates. At the other end of the spectrum are impulse control problems, in
which the controlling agent can effect instantaneous jumps in system state,
usually with an associated fixed cost. In between are singular control
problems, in which the agent can effect instantaneous state changes of any
desired size, usually at a cost proportional to the size of the
displacement; see for example, \citet{karatzas1983}. In this paper we
develop a computational method for the first of those three problem classes,
and then illustrate its use on selected test problems. Our method is a
variant of the one developed by \citet{han2018solving} for solution of
semi-linear partial differential equations, and in its implementation we
have re-used substantial amounts of the code provided by %
\citet{han2018solving} and \citet{zhou2021actor}.

\textbf{Literature Review.} Two of the most relevant streams of literature
are \textit{i}) drift rate control problems, and \textit{ii}) solving PDEs
using deep learning. \citet{ata-harrison-shepp2005} considers a
one-dimensional drift rate control problem on a bounded interval under a
general cost of control but no state costs. The authors characterize the
optimal policy in closed form; and they discuss the application of their
model to a power control problem in wireless communication. %
\citet{ormeci-matoglu-vande-vate2011} consider a drift rate control problem
where a system controller incurs a fixed cost to change the drift rate. The
authors prove that a deterministic, non-overlapping control band policy is
optimal; also see \citet{vande-vate2021}. 
\citet{ghosh-weerasinghe2007,
ghosh-weerasinghe2010} extend \citet{ata-harrison-shepp2005} by
incorporating state costs, abandonments and optimally choosing the interval
where the process lives.

Drift control problems arise in a broad range of applications in practice. %
\citet{rubino-ata2009} studies a dynamic scheduling problem for a
make-to-order manufacturing system. The authors model order cancellations as
abandonments from their queueing system. This model feature gives rise to a
drift rate control problem in the heavy traffic limit. %
\citet{ata-lee-sonmez2019} uses a drift control model to study a dynamic
staffing problem in order to determine the number of volunteer gleaners, who
sign up to help but may not show up, for harvesting leftover crops donated
by farmers for the purpose of feeding food-insecure individuals. %
\citet{bar-ilan-marion-perry2007} use a drift control model to study
international reserves. 
   {All of the papers mentioned above
study one-dimensional drift-rate control problems.}

   {The recent working paper by \citet{ata-kasikaralar2023}
studies dynamic scheduling of a multiclass queue motivated by call center
industry.} Focusing on the Halfin-Whitt asymptotic regime, the authors
derive a (limiting) drift rate control problem whose state space is $\mathbb{%
R}^d$, where $d$ is the number of buffers in their queueing model. Like us, those authors build on earlier work by \citet{han2018solving} to solve their
(high-dimensional) drift rate control problem. However, our work differs
from theirs significantly, because their control problem has no state space
constraints.

As mentioned earlier, our work builds on the seminal paper by
\citep{han2018solving}. In the last five years, there have been many other papers
written on solving PDEs using deep neural networks; see the recent surveys by \citet{beck-hutzenthaler-jentzen-kuckuck2023} and \citet{e-han-jentzen2022}.

\textbf{The remainder of this paper.} Section \ref{sec:RBM} recapitulates
essential background knowledge from RBM theory, after which Section \ref%
{sec:problem} states in precise mathematical terms the discounted control
and ergodic control problems that are the object of our study. In each case,
the problem statement is expressed in probabilistic terms initially, and
then re-expressed analytically in the form of an equivalent
Hamilton-Jacobi-Bellman equation (hereafter abbreviated to HJB equation).
Section \ref{sec:sdes} derives key identities, that significantly contribute
to the subsequent development of our computational method. Section \ref%
{sec:han} describes our computational method in detail.

Section \ref{sec:example} specifies three families of drift control test
problems, each of which has members of dimensions $d = 1, 2, \ldots$ . The
first two families arise as heavy traffic limits of certain queueing network
control problems, and we explain that motivation in some detail. Drift
control problems in the third family have a separable structure that allows
them to be solved exactly by analytical means, which is of obvious value for
assessing the accuracy of our computational method. Section \ref%
{sec:numerical} presents numerical results obtained with our method for all
three families of test problems. In that admittedly limited context, our
computed solutions are accurate to within a fraction of one percent, and our
method remains computationally feasible up to at least dimension $d = 30$,
and in some cases up to dimension 100 or more. In Section \ref%
{sec:conclusion} we describe variations and generalizations of the problems
formulated in Section \ref{sec:problem} that are of interest for various
purposes, and which we expect to be addressed in future work. Finally, there
are a number of appendices that contain proofs or other technical
elaboration for arguments or procedures that have only been sketched in the
body of the paper.

\section{RBM preliminaries}

\label{sec:RBM} We consider here a reflected Brownian motion $Z=\left\{
Z(t),t\geq 0\right\} $ with state space $\mathbb{R}_{+}^{d},$ where $d\geq
1. $ The data of $Z$ are a (negative) drift vector $\mu \in \mathbb{R}^{d},$
a $d\times d$ positive-definite covariance matrix $A =(a_{ij}),$ and a $%
d\times d$ reflection matrix $R$ of the form 
\begin{equation}
R=I-Q,\text{ where }Q\text{ has non-negative entries and spectral radius }%
\rho (Q)<1.  \label{RBM:R}
\end{equation}%
The restriction to reflection matrices of the form (\ref{RBM:R}) is not
essential for our purposes, but it simplifies the technical development and
is consistent with usage in the related earlier paper by %
\citet{blanchet2021efficient}. Denoting by $W=\left\{ W(t),t\geq 0\right\} $
a $d$-dimensional Brownian motion with zero drift, covariance matrix $A$,
and $W(0)=0$, we then have the representation 
\begin{align}
Z(t)& =Z(0) + W(t)-\mu t+RY(t),\text{ }t\geq 0,\text{ where}  \label{RBM:Z}
\\
Y_{i}(\cdot )& \text{ is continuous and non-decreasing with }Y_{i}(0)=0\text{
}(i=1,2,\ldots ,d),\text{ and}  \label{RBM:Y1} \\
Y_{i}(\cdot )& \text{ only increases at those times }t\text{ when }Z_{i}(t)=0%
\text{ }(i=1,2,\ldots ,d).  \label{RBM:Y2}
\end{align}

\citet{harrison1981reflected} showed that the relationships\ (\ref{RBM:R})
to (\ref{RBM:Y2}) determine $Y$ and $Z$ as pathwise functionals of $W,$ and
that the mapping $W\rightarrow (Y,Z)$ is continuous in the topology of
uniform convergence. We interpret the $i^{\text{th}}$ column of $R$ as the
direction of reflection on the boundary surface $S_i=\{z\in \mathbb{R}%
_{+}^{d}:z_{i}=0\},$ and call $Y_{i}=\left\{ Y_{i}(t),t\geq 0\right\} $ the
\textquotedblleft pushing process" on that boundary surface.

In preparation for future developments, let $f$ $\ $be an arbitrary $C^{2}$
(that is, twice continuously differentiable) function $\mathbb{R}%
^{d}\rightarrow \mathbb{R}$, and let $\nabla f$ denote its gradient vector
as usual. Also, we define a second-order differential operator $\mathcal{L}$
via%
\begin{equation}
\mathcal{L}f=\frac{1}{2}\sum_{i=1}^{d}\sum_{j=1}^{d}a_{ij}\frac{\partial ^{2}%
}{\partial z_{i}\partial z_{j}}f,  \label{RBM:Lf}
\end{equation}%
and a first-order differential operator $\mathcal{D}=(\mathcal{D}_{1},\ldots
,\mathcal{D}_{d})^\top$ via 
\begin{equation}
\mathcal{D}f=R^{\top }\nabla f,  \label{RBM:D}
\end{equation}%
where $\top $ in (\ref{RBM:D}) denotes transpose. Thus $\mathcal{D}%
_{i}f(\cdot )$ is the directional derivative of $f$ in the direction of
reflection on the boundary surface $S_i = \left\{z \in \mathbb{R}_+^d: z_{i}=0\right\} .$ With these
definitions, an application of Ito's formula now gives the following
identify, cf. \citet{harrison1981reflected}, Section 3:

\begin{equation}
\mathrm{d}f(Z(t))=\nabla f(Z(t))\cdot \mathrm{d}W(t)+(\mathcal{L}f-\mu \cdot
\nabla f)(Z(t)) \,\mathrm{d}t+\mathcal{D}f(Z(t))\cdot \mathrm{d}Y(t),\ t\geq
0.  \label{RBM:df}
\end{equation}%
In the obvious way, the first inner product on the right side of (\ref%
{RBM:df}) is shorthand for a sum of $d$ Ito differentials, while the last
one is shorthand for a sum of $d$ Riemann-Stieltjes differentials.

\section{Problem statements and HJB equations}

\label{sec:problem} Let us now consider a stochastic control problem whose
state space is $\mathbb{R}_{+}^{d}$ $(d\geq 1).$\ The controlled process $Z$
has the form 
\begin{equation}
Z(t)= Z(0) + W(t)-\int_{0}^{t}\theta (s)\mathrm{d}s+RY(t),\ t\geq 0,
\label{RBM:controled}
\end{equation}%
where (i) $W=\{W(t),t\geq 0\}$ is a $d$-dimensional Brownian motion with
zero drift, covariance matrix $A$, and $W(0)=0$ as in Section 2, (ii) $%
\theta =\left\{ \theta (t),t\geq 0\right\} $ is a non-anticipating control,
or non-anticipating drift process, chosen by a system manager and taking
values in a bounded set $\Theta \subset \mathbb{R}^{d},$ and (iii) $%
Y=\left\{ Y(t),t\geq 0\right\} $ is a $d$-dimensional pushing process with
components $Y_{i}$ that satisfy (\ref{RBM:Y1}) and (\ref{RBM:Y2}). Note that
our sign convention on the drift in the basic system equation (\ref%
{RBM:controled}) is \textit{not} standard. That is, we denote by $\theta(t)$
the \textit{negative} drift vector at time $t$.

The control $\theta $ is chosen to optimize an economic objective (see
below), and attention will be restricted to \textit{stationary} Markov
controls, or stationary control policies, by which we mean that 
\begin{equation}
\theta (t)=u(Z(t)),\ t\geq 0\text{ for some measurable policy function }u:%
\mathbb{R}_{+}^{d}\rightarrow \Theta .  \label{def:policy:u}
\end{equation}%
Hereafter the set $\Theta $ of drift vectors available to the system manager
will be referred to as the \textit{action space }for our control problem, and a
function $u:\mathbb{R}_{+}^{d}\rightarrow \Theta $ will simply be called a 
\textit{policy}. We denote by $Z^{u}$ the controlled RBM defined via (\ref%
{RBM:controled}) and (\ref{def:policy:u}), and denote by $Y^{u}$ the
associated $d$-dimensional boundary pushing process. 

Before specifying the system manager's economic objective, we establish the following terminology: for $m, n \geq 1$, a function $g: D \subset \mathbb{R}^m
\rightarrow \mathbb{R}^n$ is said to have polynomial growth if there exist
constants $\alpha_1, \, \beta_1 > 0 $ such that 
\begin{align*}
|g(z)| \leq \alpha_1 \left( 1+ |z|^{\beta_1} \right), \ z \in D.
\end{align*}
Because the action space $\Theta$ is bounded,  for a function $g:\mathbb{R}_+^d \times \Theta \rightarrow \mathbb{R}$, the polynomial growth assumption reduces to the following: 
\begin{align}
|g(z,\theta)| \leq \alpha_2 \left( 1+ |z|^{\beta_2} \right) \text{ for all }
z \in \mathbb{R}_+^d \text{ and } \theta \in \Theta,
\label{eqn:growth-condition:cost}
\end{align}
where $\alpha_2, \, \beta_2$ are positive constants.

With regard to the
system manager's objective, we take as given a continuous cost function $c:%
\mathbb{R}_{+}^{d}\times \Theta \rightarrow \mathbb{R}$ with polynomial
growth and a vector $%
\kappa \in \mathbb{R}_+^d$ of \textit{penalty rates} associated with pushing
at the boundary. (As it happens, the boundary penalty rates are all zeros
for the numerical examples considered in this paper, but positive penalty rates 
will be needed in future applications.) The cumulative cost incurred over
the time interval $[0,t]$ under policy $u$ is 
\begin{equation}
C^{u}(t)\equiv \int_{0}^{t}c(Z^{u}(s),u(Z^{u}(s))) \,\mathrm{d}s+ \kappa
\cdot Y^u(t),\text{ }t\geq 0.  \label{def:Cost}
\end{equation}

Because our action space $\Theta$ is bounded by assumption, the controlled
RBM $Z^u$ has bounded drift under any policy $u$, from which one can prove
the following mild but useful property; see Appendix \ref%
{appendix:proof:prop:bounded} for its proof.

\begin{proposition}
Under any policy $u$ and for any integer $n=1,2,\ldots$ the function 
\begin{equation*}
g_n(z,t) = \mathbb{E}_z \left\{ |Z^{u}(t)|^n \right\}, \ t \geq 0,
\end{equation*}
has polynomial growth in $t$ for each fixed $z \in \mathbb{R}_+^d$. \label%
{prop:bounded}
\end{proposition}

\subsection{Discounted control}

\label{sec:disc:form}

In our first problem formulation, an interest rate $r>0$ is taken as given,
and we adopt the following discounted cost objective: choose a policy $u$ to
minimize 
\begin{equation}
V^{u}(z)\equiv \mathbb{E}_{z}\left[ \int_{0}^{\infty }e^{-rt}\mathrm{d}C^u(t)%
\right] =\mathbb{E}_{z}\left[ \int_{0}^{\infty }e^{-rt}\left[%
c(Z^{u}(t),u(Z^{u}(t)))\,\mathrm{d}t + \kappa \cdot \mathrm{d} Y^u(t)) %
\right]\right] ,  \label{problem:discounted}
\end{equation}%
where $\mathbb{E}_{z}\left( \cdot \right) $ denotes a conditional
expectation given that $Z(0)=z.$ Given the polynomial growth condition (\ref%
{eqn:growth-condition:cost}), it follows from Proposition \ref{prop:bounded}
that the moments of $Z(t)$    {have polynomial growth} as
functions of $t$ for each fixed initial state $z$. Also, because $\Theta$ is
bounded by assumption, one can easily derive an affine bound for $\mathbb{E}%
_z\{\kappa \cdot Y^u(t)\}$ viewed as a function of $t$. Given the assumed
positivity of the interest rate $r$, the expectation in (\ref%
{problem:discounted}) is therefore well defined and finite for each $z \in 
\mathbb{R}_+^d$.

Hereafter we refer to $V^{u}(\cdot )$ as the \textit{value function} under
policy $u$, and define the \textit{optimal value function}%
\begin{equation}
V(z)=\min_{u \in \,\mathcal{U}}V^{u}(z)\text{ for each }z\in \mathbb{R}%
_{+}^{d},  \label{problem;discount:min}
\end{equation}%
where $\mathcal{U}$ is the set of stationary Markov control policies.

    {Let $u$ be an arbitrary policy, and let $f: \mathbb{R}_+^d \rightarrow 
\mathbb{R}$ be a $C^2$ function with polynomial growth. Also, we define the
second-order differential operation associated with policy $u$, 
\begin{equation*}
\mathcal{A}^uf(z)=\mathcal{L} f(z)-u(z)\cdot \nabla f(z), \text{ }z\in 
\mathbb{R}_+^d, 
\end{equation*}
where $\mathcal{L}$ is defined via (\ref{RBM:Lf}). The identity (\ref{RBM:df}%
) continues to hold for the controlled process $Z^u$ if one replaces the
constant (reference) drift vector $\mu$ be the state-dependent drift
function $\theta(z)=u(z)$, and then standard arguments as in Section 6.3 of %
\citet{harrison2013brownian} leads to the following identity 
\begin{equation}
f(z)=\mathbb{E}_z\left \{ \int_0^\infty e^{-rt}\left[ -(\mathcal{A}^u-r)
f(Z^u(t)) \mathrm{d}t -\mathcal{D} f(Z^u(t)) \mathrm{d}Y^u(t) \right]\right
\}.  \label{eq:expection:f}
\end{equation}
Comparing (\ref{eq:expection:f}) with (\ref{problem:discounted}), one is led
to the following PDE for $V^u$: 
\begin{equation}
\mathcal{A}^u V^u(z)-rV^u(z) + c(z,u(z))=0, \text{ } z\in\mathbb{R}^d_+ \text{ } (i=1,\ldots,d).  \label{HJB:discount:u:L}
\end{equation}
with boundary conditions 
\begin{equation}
\mathcal{D}_i V^u(z)=-\kappa_i\text{ if } z_i = 0\text{ } (i=1,\ldots,d).
\label{HJB:discount:u:D}
\end{equation}
That is, if there is a $C^2$ solution of the PDE
(\ref{HJB:discount:u:L})-(\ref{HJB:discount:u:D}) that has polynomial
growth, then it is equal to the value function under policy $u$.}

The corresponding HJB equation, to be solved for the \textit{optimal} value
function $V(\cdot ),$ is 
\begin{equation}
\mathcal{L}V(z)-\max_{\theta \in \Theta }\left\{ \theta \cdot \nabla
V(z)-c(z,\theta )\right\} =rV(z),\ z\in \mathbb{R}_{+}^{d},
\label{HJB:discount:min:L}
\end{equation}%
with boundary conditions 
\begin{equation}
\mathcal{D}_{i}V(z)=-\kappa_i\text{ if }z_{i}=0\text{ }(i=1,2,\ldots ,d).
\label{HJB:discount:min:D}
\end{equation}%
Moreover, the policy 
\begin{equation}
u^{\ast }(z)=\arg \max_{\theta \in \Theta }\{\theta \cdot \nabla
V(z)-c(z,\theta )\}  \label{HJB:discount:optimal-policy}
\end{equation}%
is optimal, meaning that $V^{u^{\ast }}(z)=V(z)$ for $z\in \mathbb{R}%
_{+}^{d} $.

There will be no attempt here to prove existence of $C^{2}$ solutions, but
our computational method proceeds as if that were the case, striving to
compute a $C^{2}$ function $V$ that satisfies (\ref{HJB:discount:min:L})-(%
\ref{HJB:discount:min:D}) as closely as possible in a certain sense. 
   {As an aside, whenever we refer to a $C^2$ function on a
closed set, we mean that it is $C^2$ in an open neighborhood of the set}.

In Appendix \ref{appendix:verification:discounted} we use (\ref{RBM:df}) to
verify that a sufficiently regular solution of the PDE (\ref%
{HJB:discount:u:L})-(\ref{HJB:discount:u:D}) does in fact satisfy (\ref%
{problem:discounted}) as intended, and similarly, that a sufficiently
regular solution of (\ref{HJB:discount:min:L})-(\ref{HJB:discount:min:D})
does in fact satisfy (\ref{problem;discount:min}).


\subsection{Ergodic control}

\label{sec:ergodic:form}

For our second problem formulation, it is assumed that 
\begin{equation}
c(z,\theta )\geq 0\text{ for all }(z,\theta )\in \mathbb{R}_{+}^{d}\times
\Theta .  \label{assumption:non-negative}
\end{equation}%
Readers will see that our analysis can be extended to cost functions that
take on negative values in at least some states, but to do so one must deal
with certain irritating technicalities. To be specific, the issue is whether
the expected values involved in our formulation are well defined.

In preparation for future developments, let us recall that a square matrix $%
R $ of the form (\ref{RBM:R}), called a \textit{Minkowski} \textit{matrix}
in linear algebra (or just \textit{M-matrix} for brevity), is non-singular,
and its inverse is given by the Neumann expansion%
\begin{equation*}
R^{-1}=I+Q+Q^{2}+\ldots \geq 0.
\end{equation*}%
Hereafter, we assume that 
\begin{equation}
\text{there exists at least one }\theta \in \Theta \text{ such that }%
R^{-1}\theta > 0.  \label{assumption:non-singular}
\end{equation}%
It is known that an RBM with a non-singular covariance matrix, reflection
matrix $R$, and negative drift vector $\theta $ has a stationary
distribution if and only if the inequality in (\ref{assumption:non-singular}%
) holds, cf. Section 6 of \citet{harrison1987brownian}. Of course, our
statement of this ``stability condition" reflects the non-standard sign
convention used in this paper. That is, $\theta$ denotes the \textit{negative%
} drift vector of the RBM under discussion.

For our ergodic control problem, a policy function $u:\mathbb{R}%
_{+}^{d}\rightarrow \Theta $ is said to be \textit{admissible} if, first,
the corresponding controlled RBM $Z^{u}$ has a unique stationary
distribution $\pi ^{u}$, and if, moreover, 
\begin{equation}
\int_{\mathbb{R}_{+}^{d}}|f(z)|\,\pi ^{u}(dz)<\infty
\label{eqn:admissibility-cond-ergodic-case}
\end{equation}%
for any function $f:\mathbb{R}_{+}^{d}\rightarrow \mathbb{R}$ with
polynomial growth. For an admissible prolicy $u$, as done in \citet{harrison1987brownian} and %
\citet{dai1991steady}, one can define the boundary measures $\nu
_{1}^{u},\ldots ,\nu _{d}^{u}$ as follows: 
\begin{equation*}
\nu _{i}(A)=\mathbb{E}_{\pi ^{u}}\left[ \int_{0}^{1}1_{\left\{ Z^{u}(t) \,
\in A\right\} }dY_{i}^{u}(t)\right] ,\text{ }i=1,\ldots ,d,
\end{equation*}%
where $A$ is any Borel subset of the boundary surface $S_i= \{ z \in \mathbb{R}_+^d: z_i =0 \}$ and $%
\mathbb{E}_{\pi ^{u}}$ denotes a conditional expectation given that $%
Z(0)\sim \pi ^{u}$. Our assumption (\ref{assumption:non-singular}) ensures
the existence of at least one admissible policy $u$, as follows. Let $\theta
\in \Theta $ be a negative drift vector satisfying (\ref%
{assumption:non-singular}), and consider the constant policy $u(\cdot
)\equiv \theta $. The corresponding controlled process $Z^{u}$ is then an
RBM having a unique stationary distribution $\pi ^{u}$, as noted above. It
has been shown in \citet{budhiraja-lee2007} that the moment generating
function of $\pi ^{u}$ is finite in a neighborhood of the origin, from which
it follows that $\pi ^{u}$ has finite moments of all orders. Thus $\pi ^{u}$
satisfies (\ref{eqn:admissibility-cond-ergodic-case}) for any function $f$
with polynomial growth, so $u$ is admissible.

Because our cost function $c(z,\theta )$ has polynomial growth and our
action space $\Theta $ is bounded, the steady-state average cost 
\begin{equation}
\xi ^{u}\equiv \int_{\mathbb{R}_{+}^{d}}c(z,u(z))\,\pi ^{u}(dz)+ \sum_{i=1}^d \kappa_i \nu_{i}^{u}(S_i)
\label{eqn:relative-value-func-under-u}
\end{equation}%
is well defined and finite under any admissible policy $u$. The objective in
our ergodic control problem is to find an admissible policy $u$ for which $%
\xi ^{u}$ is minimal.

   {Let $u$ be an arbitrary admissible policy and consider the following PDE:
\begin{equation}
\mathcal{L}v^{u}(z)-u(z)\cdot \nabla v^{u}(z)+c(z,u(z) )=\xi \text{
for each }z\in \mathbb{R}_{+}^{d},  \label{HJB:ergodic:L}
\end{equation}with boundary conditions 
\begin{equation}
\mathcal{D}_{i}v^{u}(z)=-\kappa_i\text{ if }z_{i}=0\text{ }(i=1,2,\ldots ,d).
\label{HJB:ergodic:D}
\end{equation}If $(\xi, v^u)$ solve this PDE and $v^u$ is $C^2$ with polynomial growth, then one can show that $\xi = \xi^u$, the steady state average cost under policy $u$, and $v^u$ is the relative value function corresponding to policy $u$.
}

The HJB equation for ergodic control is again of a standard form, involving
a constant $\xi $ (interpreted as the minimum achievable steady-state
average cost) and a relative value function $v:\mathbb{R}_{+}^{d}\rightarrow 
\mathbb{R}$. To be specific, the HJB equation is 
\begin{equation}
\mathcal{L}v(z)-\max_{\theta \in \Theta }\left\{ \theta \cdot \nabla
v(z)-c(z,\theta )\right\} =\xi \text{ for each }z\in \mathbb{R}_{+}^{d},
\label{HJB:ergodic:min:L}
\end{equation}%
with boundary conditions%
\begin{equation}
\mathcal{D}_{i}v(z)=-\kappa _{i}\text{ if }z_{i}=0\text{ }(i=1,2,\ldots ,d).
\label{HJB:ergodic:min:D}
\end{equation}%
Paralleling the previous development for discounted control, we show the
following in Appendix \ref{appendix:verification:ergodic}: if a $C^{2}$
function $v$ and a constant $\xi $ jointly satisfy (\ref{HJB:ergodic:min:L}%
)-(\ref{HJB:ergodic:min:D}), then 
\begin{equation}
\xi =\inf_{u\in \mathcal{U}}\xi ^{u},  \label{eqn:ergodic:optimal-xi}
\end{equation}%
where $\mathcal{U}$ denotes the set of admissible controls for the ergodic
cost formulation. Moreover, the policy 
\begin{equation}
u^{\ast }(z)=\arg \max_{\theta \in \Theta }\left\{ \theta \cdot \nabla
v(z)-c(z,\theta )\right\} ,\ z\in \mathbb{R}_{+}^{d},
\label{eqn:defn:ergodic:optimal:policy}
\end{equation}%
is optimal, meaning that $\xi ^{u^{\ast }}=\xi .$ Again paralleling the
previous development for discounted control, we make no attempt to prove
that such a solution for (\ref{HJB:ergodic:min:L})-(\ref{HJB:ergodic:min:D})
exists. In Appendix \ref{appendix:verification:ergodic} we use (\ref{RBM:df}%
) to verify that a sufficiently regular solution of the PDE (\ref%
{HJB:ergodic:L})-(\ref{HJB:ergodic:D}) does in fact satisfy (\ref%
{eqn:relative-value-func-under-u}) as intended, and similarly, that a
sufficiently regular solution of (\ref{HJB:ergodic:min:L})-(\ref%
{HJB:ergodic:min:D}) does in fact satisfy (\ref{eqn:ergodic:optimal-xi}).

\section{Equivalent SDEs}

\label{sec:sdes}

In this section we prove two key identities, Equations (\ref{eqn:loss:dis})
and (\ref{eqn:identity:erg}) below, that are closely patterned after results
used by \citet{han2018solving} to justify their ``deep BSDE method" for
solution of certain non-linear PDEs. That earlier work provided both
inspiration and detailed guidance for our study, but we include these
derivations to make the current account as nearly self-contained as
possible. Sections \ref{sec:infinite-horizon:SDE} and \ref{sec:ergodic:SDE}
treat the discounted and ergodic cases, respectively.

Our method begins by specifying what we call a \textit{reference policy}.
This is a nominal or default policy, specified at the outset but possibly
revised in light of computational experience, that we use to generate sample
paths of the controlled RBM $Z$. Roughly speaking, one wants to choose the
reference policy so that paths    {of the state process} tend
to occupy parts of the state space thought to be most frequently visited by
an optimal policy.

\subsection{Discounted control}

\label{sec:infinite-horizon:SDE}

Our reference policy for the discounted case chooses a constant action $u(z)
= \tilde{\theta} > 0$ in every state $z \in \mathbb{R}_+^d$. (Again we
stress that, given the non-standard sign convention embodied in (\ref%
{RBM:controled}) and (\ref{assumption:non-negative}), this means that $%
\tilde{Z}$ has a constant drift vector $- \tilde{\theta}$, with all
components negative.) Thus the corresponding \textit{reference process} $%
\tilde{Z}$ is a $d$-dimensional RBM which, in combination with its $d$%
-dimensional pushing process $\tilde{Y}$ and the $d$-dimensional Brownian
motion $W$ defined in Section \ref{sec:RBM}, satisfies 
\begin{equation}
\tilde{Z}(t) = \tilde{Z}(0) + W(t) -\tilde{\theta}\,t +R \, \tilde{Y}(t)
,~t\geq 0,  \label{RBM:behaviour}
\end{equation}%
plus the obvious analogs of Equations (\ref{RBM:Y1}) and (\ref{RBM:Y2}). For
the key identity (\ref{eqn:loss:dis}) below, let 
\begin{equation}
F(z,x) = \tilde{\theta} \cdot x - \max_{ \theta \in \Theta} \left\{ \theta
\cdot x - c(z, \theta) \right\} \text{ for } z \in \mathbb{R}_+^d \text{ and 
} x \in \mathbb{R}^d.  \label{eq:Ffunction}
\end{equation}
{   \textbf{Remark:} In our numerical examples, the cost functions $%
c(z, \theta)$ allow for a closed-form expression of $F(z,x)$, simplifying
our method. However, when the cost function is complex, computing $F(z,x)$
and its partial derivatives with respect to $x$, i.e., $\partial F(z,x) /
\partial x_i$ ($i=1,\ldots,d$), can be challenging. In such cases, several
approaches are possible. First, one can generate a dataset consisting of
inputs ($z^l, x^l$) and outputs $F(z^l,x^l)$ for $l=1,\ldots,L$. With a
sufficiently large $L$, a neural network can be trained offline to
approximate $F$. The partial derivatives of $F$ can then be approximated via
autodifferentiation, as demonstrated in \citet{ata-zhou2024}. Alternatively,
an optimization solver can be used to compute $F(z,x)$, and the envelope
theorem can be invoked to determine its partial derivatives $\partial F(z,x)
/ \partial x_i$ ($i=1,\ldots,d$). However, this approach is computationally
expensive. Lastly, the actor-critic method developed in \citet{zhou2021actor}
can be employed. }

\begin{proposition}
\label{prop:discount:loss} If $V\left( \cdot \right) $ 
   {is a
$C^2$ function with polynomial growth that} satisfies the HJB\ equation (\ref%
{HJB:discount:min:L}) - (\ref{HJB:discount:min:D}), then it also satisfies
the following identity almost surely for any $T>0$: 
\begin{eqnarray}
e^{-rT}V(\tilde{Z}(T))-V(\tilde{Z}(0)) & = &\int_{0}^{T}e^{-rt}\nabla V(%
\tilde{Z}(t)) \cdot \mathrm{d}W(t)  \notag \\
& - &\int_0^T e^{-rt} \kappa \cdot \mathrm{d}\tilde{Y}(t)
-\int_{0}^{T}e^{-rt}F(\tilde{Z}(t),\nabla V(\tilde{Z}(t)))\mathrm{d}t.
\label{eqn:loss:dis}
\end{eqnarray}
\end{proposition}

\begin{proof} Applying Ito's formula to $e^{-rt}V(\tilde{Z}(t))$ and using Equation (\ref{RBM:df}) yield
\begin{eqnarray}
e^{-rT}V(\tilde{Z}(T))-V(\tilde{Z}(0))  
&=&\int_{0}^{T}e^{-rt}\nabla V(\tilde{Z}(t)) \cdot \mathrm{d}W(t)+%
\int_{0}^{T}e^{-rt} \mathcal{D} V(\tilde{Z}(t)) \cdot \mathrm{d}\tilde{Y}(t)  \label{dis:ito}  \\
&&+\int_{0}^{T}e^{-rt}\left( \mathcal{L}V(\tilde{Z}(t))-\tilde{\theta}\cdot
\nabla V(\tilde{Z}(t))-rV(\tilde{Z}(t))\right) \mathrm{d}t. \notag
\end{eqnarray}%
Using the boundary condition (\ref{HJB:discount:min:D}), plus the complementarity condition (\ref{RBM:Y2}) for $\tilde{Y}$ and $\tilde{Z}$, one has
\begin{equation}
\int_{0}^{T}e^{-rt} \mathcal{D} V(\tilde{Z}(t)) \cdot \mathrm{d}\tilde{Y}(t)=-\int_0^T e^{-rt} \kappa \cdot \mathrm{d}\tilde{Y}(t). \label{eqn:aux-bound-cond}
\end{equation}%
Furthermore, substituting $z=\tilde{Z}(t)$ in the HJB equation (\ref{HJB:discount:min:L}), multiplying both sides by $e^{-rt}$, rearranging the terms, and integrating over 
$[0,T]$ yields 
\begin{equation}
\int_{0}^{T}e^{-rt}\left( \mathcal{L}V(\tilde{Z}(t))-rV(\tilde{Z}(t))\right)\mathrm{d}t=\int_{0}^{T}e^{-rt}\max_{\theta \in \Theta
}\left( \theta \cdot \nabla V(\tilde{Z}(t))-c(\tilde{Z}(t),\theta )
\right) \mathrm{d}t .
\label{HJB:dis:int}
\end{equation}%
Substituting Equations (\ref{eqn:aux-bound-cond}) and (\ref{HJB:dis:int}) into Equation (\ref{dis:ito}) gives Equation (\ref{eqn:loss:dis}). 
\end{proof}
Proposition \ref{prop:discount:loss} provides the motivation for the loss
function that we strive to minimize in our computational method (see Section %
\ref{sec:han}). Before developing that approach, we prove the following,
which can be viewed as a converse of Proposition \ref{prop:discount:loss}.

\begin{proposition}
\label{thm:doublepara:dis} Suppose that $V: \mathbb{R}_+^d \rightarrow 
\mathbb{R}$ is a $C^2$ function, $G: \mathbb{R}_+^d \rightarrow \mathbb{R}^d$
is continuous, and $V, \, \nabla V$, and $G$ all have polynomial growth.
Also assume that the following identity holds almost surely for some fixed $%
T>0$ and every $\tilde{Z}(0)= z \in \mathbb{R}_+^d$: 
\begin{equation}
e^{-rT}V(\tilde{Z}(T))-V(\tilde{Z}(0))=\int_{0}^{T}e^{-rt}G(\tilde{Z}%
(t))\cdot \mathrm{d}W(t) - \int_0^T e^{-rt} \kappa \cdot \mathrm{d}\tilde{Y}%
(t) -\int_{0}^{T}e^{-rt}F(\tilde{Z}(t),G(\tilde{Z}(t))) \, \mathrm{d}t.
\label{eqn:thm:VG}
\end{equation}%
Then $G(\cdot )=\nabla V(\cdot ) $ and $V$ satisfies the HJB equation (\ref%
{HJB:discount:min:L}) - (\ref{HJB:discount:min:D}).
\end{proposition}

\begin{remark}
The surprising conclusion that (\ref{eqn:thm:VG}) implies $\nabla V(\cdot
)=G(\cdot )$, without any \textit{a priori} relationship between $G$ and $%
\nabla V$ being assumed, motivates the ``double parametrization'' method in
Section \ref{sec:han}.
\end{remark}

\begin{proof}[Proof Sketch for Proposition \ref{thm:doublepara:dis}] Because $\tilde{Z}$ is a time-homogeneous Markov process, we can express (\ref{eqn:thm:VG}) equivalently as follows for any $k=0,1,\ldots$ :
\begin{align}
    e^{-rT} V(\tilde{Z}((k+1)T)) - V(\tilde{Z}(kT)) & =   \int_{kT}^{(k+1)T} e^{-r (t-kT)} G(\tilde{Z}(t)) \cdot dW(t) - \int_{kT}^{(k+1)T} e^{-r (t-kT)}\kappa \cdot \mathrm{d}\tilde{Y}(t)
    \notag \\
    & -\int_{kT}^{(k+1)T} e^{-r (t-kT)} F(\tilde{Z}(t),G(\tilde{Z}(t)))\, dt. \label{eqn:proof:conv:identiy:disc:iterative:aux}
\end{align}
Now multiply both sides of (\ref{eqn:proof:conv:identiy:disc:iterative:aux}) by $e^{-rkT}$, then add the resulting relationships for $k=0,1,\ldots,n-1$ to arrive at the following:
\begin{align}
    & e^{-rnT} V(\tilde{Z}(nT)) \notag \\
     =& V(\tilde{Z}(0)) 
    + \int_{0}^{nT} e^{-rt} G(\tilde{Z}(t)) \cdot dW(t) -\int_{0}^{nT} e^{-rt} \kappa \cdot \mathrm{d}\tilde{Y}(t) - \int_{0}^{nT} e^{-rt} F(\tilde{Z}(t),G(\tilde{Z}(t))) dt. \label{eqn:aux2:proof:conv:disc:identity}
\end{align}
Because $G$ has polynomial growth, one can show that 
\begin{align*}
    \mathbb{E}_z \left[ \int_0^{nT} e^{-2rt} G(\tilde{Z}(t))^2 dz \right]  < \infty
\end{align*}
for all $n \geq 1$. Thus, when we take $\mathbb{E}_z$ of both sides of (\ref{eqn:aux2:proof:conv:disc:identity}), the stochastic integral (that is, the second term) on the right side vanishes, and then rearranging terms gives the following:
\begin{equation*}
V(z)=e^{-rnT}\mathbb{E}_{z}\left[ V(\tilde{Z}(nT))\right] +\mathbb{E}_{z}%
\left[ \int_{0}^{nT}e^{-rt}F(\tilde{Z}(t),G(\tilde{Z}(t)))\mathrm{d}t+\int_{0}^{nT} e^{-rt} \kappa \cdot \mathrm{d}\tilde{Y}(t)\right]
,
\end{equation*}%
for arbitrary positive integer $n.$ 

By Proposition \ref{prop:bounded} and the polynomial growth condition of $V$, we have 
$ e^{-rnT}\mathbb{E}_{z}[ V(\tilde{Z}(nT))]  
\rightarrow 0$ as $n\rightarrow \infty$. Therefore, 
\begin{equation*}
V(z)=\lim_{n\rightarrow \infty} \mathbb{E}_{z}\left[ \int_{0}^{nT }e^{-rt}F\left( \tilde{Z}(t),G(%
\tilde{Z}(t))\right) \mathrm{d}t+\int_{0}^{nT} e^{-rt} \kappa \cdot \mathrm{d}\tilde{Y}(t)\right] \text{ for }z\in \mathbb{R}_{+}^{d}.
\end{equation*}%
Similarly, since $F$ and $G$ have polynomial growth, we conclude that
\begin{align*}
	 &\mathbb{E}_{z}\left[ \int_{0}^{\infty }e^{-rt}\left|F\left( \tilde{Z}(t),G(%
	\tilde{Z}(t))\right)\right| \mathrm{d}t\right] <+\infty  \text{ for }z\in \mathbb{R}_{+}^{d}, \text{ and } \\
	&
	\int_{0}^{nT }e^{-rt}F\left( \tilde{Z}(t),G(%
	\tilde{Z}(t))\right) \mathrm{d}t \leq  \int_{0}^{\infty }e^{-rt}\left|F\left( \tilde{Z}(t),G(%
	\tilde{Z}(t))\right)\right| \mathrm{d}t <+\infty  \text{ for }z\in \mathbb{R}_{+}^{d}.
\end{align*}%
Thus, by dominated convergence and monotone convergence, we have
\begin{equation}
V(z)=\mathbb{E}_{z}\left[ \int_{0}^{\infty }e^{-rt}F\left( \tilde{Z}(t),G(%
\tilde{Z}(t))\right) \mathrm{d}t+\int_{0}^{\infty} e^{-rt} \kappa \cdot \mathrm{d}\tilde{Y}(t)\right] \text{ for }z\in \mathbb{R}_{+}^{d}.
\label{eqn:V(z):double}
\end{equation}%
In other words, $V(z)$ can be viewed as the expected discounted cost
associated with the RBM under the reference policy starting in state $\tilde{Z%
}(0)=z,$ where $F\left( \cdot ,G(\cdot )\right) $ is the state-cost function.    {Now, consider the following PDE:
\begin{equation*}
\mathcal{L}f(z)-\tilde{\theta}\cdot \nabla f(z)+F\left( z,G(z)\right)
=rf(z), \ z\in \mathbb{R}_{+}^{d}, 
\end{equation*}
with boundary conditions
\begin{eqnarray*}
    \mathcal{D}_i f(z) =  -\kappa_i \text{ if } z_i = 0\ (i=1,2, \ldots,d). 
\end{eqnarray*}
As assumed for the other PDEs considered in this paper, we assume this PDE has a $C^2$ solution with polynomial growth. Using Ito's lemma, one can then show that $f(z) = V(z)$. Thus, $V$ satisfies the following PDE:}
\begin{equation}
\mathcal{L}V(z)-\tilde{\theta}\cdot \nabla V(z)+F\left( z,G(z)\right)
=rV(z), \ z\in \mathbb{R}_{+}^{d},  \label{HJB:VG}
\end{equation}%
   {with boundary conditions
\begin{eqnarray}
    \mathcal{D}_i V(z) =  -\kappa_i \text{ if } z_i = 0 \ (i=1,2, \ldots,d). \label{eqn:boundary-condition-prop3-proof}
\end{eqnarray}
}
Suppose that $G(\cdot )=\nabla V(\cdot )$ (which we will prove later).
Substituting this into Equation (\ref{HJB:VG}) and using the definition of $%
F $, it follows that 
\begin{equation}
\mathcal{L}V(z)-\max_{\theta \in \Theta }\left\{ \theta \cdot \nabla
V(z)-c(z,\theta )\right\} =rV(z),\ z\in \mathbb{R}_{+}^{d},
\label{eq:dis:HJB:prop3}
\end{equation}%
which along with the boundary condition (\ref{HJB:discount:u:D}) gives the
desired result.

To complete the proof, it remains to show that $G(\cdot )=\nabla V(\cdot ).$
By applying Ito's formula to $e^{-rt}V(\tilde{Z}(t))$ and using Equations (\ref{RBM:Y1})-(\ref{RBM:Y2}) and (\ref{HJB:discount:min:D}), we conclude that 
\begin{align*}
e^{-rT}V(\tilde{Z}(T))-V(\tilde{Z}(0)) =  &\int_{0}^{T}e^{-rt}\left( \mathcal{L}V(\tilde{Z}(t))-\tilde{\theta}\cdot
\nabla V(\tilde{Z}(t))-rV(\tilde{Z}(t))\right) \mathrm{d}t \\
& +
\int_{0}^{T}e^{-rt}\nabla V(\tilde{Z}(t)) \cdot \mathrm{d}W(t)+\int_0^T e^{-rt}\mathcal{D} V(\tilde{Z}(t))\cdot \mathrm{d}Y(t). 
\end{align*}%
Then, using Equation (\ref{HJB:VG}) and (\ref{eqn:boundary-condition-prop3-proof}), we rewrite the preceding equation as
follows: 
\begin{equation*}
e^{-rT}V(\tilde{Z}(T))-V(\tilde{Z}(0))=\int_{0}^{T}e^{-rt}\nabla V(\tilde{Z}%
(t)) \cdot \mathrm{d}W(t)- \int_0^T e^{-rt} \kappa \cdot \mathrm{d}\tilde{Y}(t) -\int_{0}^{T}e^{-rt}F\left( \tilde{Z}(t),G(\tilde{Z}%
(t))\right) \mathrm{d}t.
\end{equation*}
Comparing this with Equation (\ref{eqn:thm:VG}) yields 
\begin{equation*}
\int_{0}^{T}e^{-rt}\left( G(\tilde{Z}(t))-\nabla V(\tilde{Z}(t))\right)  \cdot 
\mathrm{d}W(t)=0,
\end{equation*}
which yields the following:
\begin{equation}
\mathbb{E}_z\left[\left(\int_{0}^{T}e^{-rt}\left( G(\tilde{Z}(t))-\nabla V(\tilde{Z}(t))\right)  \cdot 
\mathrm{d}W(t)\right)^2\right]=0.
\label{eq:=0}
\end{equation}
Thus, provided that $e^{-rt}\left( G(\tilde{Z}(t))-\nabla V(\tilde{Z}(t))\right)$ is square integrable, Ito's isometry \citep[Lemma D.1]{zhang2020wasserstein} yields the following: 
\begin{align*}
\mathbb{E}_z\left[\left(\int_{0}^{T}e^{-rt}\left( G(\tilde{Z}(t))-\nabla V(\tilde{Z}(t))\right)  \cdot 
\mathrm{d}W(t)\right)^2\right] 
=\mathbb{E}_{z}\left[ \int_{0}^{T}\left\| e^{-rt}\left( G(\tilde{Z}(t))-\nabla
V(\tilde{Z}(t)) \right)\right\|_A ^{2} \mathrm{d}t\right] = 0,
\label{eq:ito_isometry}
\end{align*}
where $\|x\|_A:=x^\top A x$. The square integrability of $e^{-rt} \left(G(\tilde{Z}(t))- \nabla V(\tilde{Z}(t)) \right)$ follows because $G$ and $\nabla V$ have polynomial growth and the action space $\Theta$ is bounded. Because of $A$ is a  positive definite matrix, we then have $\nabla V(\tilde{Z}(t))=G(\tilde{Z}(t))$ almost surely. By the
continuity of $\nabla V\left( \cdot \right) $ and $G(\cdot ),$ we conclude
that $\nabla V(\cdot )=G(\cdot )$.
\end{proof}

\subsection{Ergodic control}

\label{sec:ergodic:SDE}

Again we use a reference policy with constant (negative) drift vector $%
\tilde{\theta}$, and now we assume that $R^{-1} \tilde{\theta} > 0$, which
ensures that the reference policy is admissible for our ergodic control
formulation.

\begin{proposition}
\label{prop:ergodic:loss} If $v\left( \cdot \right) $ and $\xi $ solve the
HJB\ equation (\ref{HJB:ergodic:min:L}) - (\ref{HJB:ergodic:min:D}) 
   {and that $v(\cdot)$ is a $C^2$ function with polynomial
growth}, then 
   {they also satisfy the following identity
almost surely for any $T>0$:} 
\begin{equation}  \label{eqn:identity:erg}
v(\tilde{Z}(T))-v(\tilde{Z}(0))=\int_{0}^{T}\nabla v(\tilde{Z}(t)) \cdot 
\mathrm{d}W(t)+T\xi- \int_0^T \kappa \cdot \mathrm{d}\tilde{Y}(t)
-\int_{0}^{T}F(\tilde{Z}(t),\nabla v(\tilde{Z}(t))) \, \mathrm{d}t.
\end{equation}
\end{proposition}

\begin{proof}
Applying Ito's formula to $v(z)$ yields
\begin{align}
v(\tilde{Z}(T))&-v(\tilde{Z}(0))  \label{erg:ito} \\
&=\int_{0}^{T}\nabla v(\tilde{Z}(t)) \cdot \mathrm{d}W(t)+\int_{0}^{T} \mathcal{D} v(%
\tilde{Z}(t)) \cdot \mathrm{d}\tilde{Y}(t)+\int_{0}^{T}\left( \mathcal{L}v(\tilde{Z%
}(t))-\tilde{\theta}\cdot \nabla v(\tilde{Z}(t))\right) \mathrm{d}t.  \notag
\end{align}%
Recall the boundary condition of the HJB equation is $\mathcal{D}_{j}v(z)=-\kappa_j$ if $z_{j}=0$. Thus Equations (\ref{RBM:Y1})-(\ref{RBM:Y2}) jointly imply
\begin{equation*}
\int_{0}^{T}\mathcal{D} v(\tilde{Z}(t))\cdot \mathrm{d}\tilde{Y}(t)= -\int_0^T \kappa \cdot \mathrm{d}\tilde{Y}(t) .
\end{equation*}%
Then, substituting the HJB equation (\ref{HJB:ergodic:min:L}) into
Equation (\ref{erg:ito}) gives (\ref{eqn:identity:erg}).
\end{proof}

\begin{proposition}
\label{thm:doublepara:ergodic} Suppose that $v: \mathbb{R}_+^d \rightarrow 
\mathbb{R}$ is a $C^2$ function, $g: \mathbb{R}_+^d \rightarrow \mathbb{R}^d$
is continuous, and $v, \, \nabla v, \, g$ all have polynomial growth. Also
assume that the following identity holds almost surely for some fixed $T>0$,
a scalar $\xi$ and every $Z(0)=z \in \mathbb{R}_+^d$: 
\begin{equation}
v(\tilde{Z}(T))-v(\tilde{Z}(0))=\int_{0}^{T}g(\tilde{Z}(t)) \cdot \mathrm{d}%
W(t)+T\xi - \int_0^T \kappa \cdot \mathrm{d}\tilde{Y}(t) -\int_{0}^{T}F(%
\tilde{Z}(t),g(\tilde{Z}(t))) \, \mathrm{d}t.  \label{eqn:thm:VG_erg}
\end{equation}%
Then, $g(\cdot ) = \nabla v(\cdot )$ and $(v, \xi)$ satisfies the HJB
equation (\ref{HJB:discount:min:L}) - (\ref{HJB:discount:min:D}).
\end{proposition}

\begin{proof}[Proof Sketch]
Let $\tilde{\pi}$ be the stationary distribution of the RBM $\tilde{Z}$ under the reference policy and $\tilde{Z%
}(\infty )$ be a random variable with the distribution $\tilde{\pi}$. Then, assuming the initial distribution of the RBM under the reference policy is $\tilde{\pi}$, i.e. $\tilde{Z}(0)\sim \tilde{\pi}$, its marginal distribution at time t is also $\tilde{\pi}$, i.e. $\tilde{Z}(t)\sim \tilde{\pi%
}$ for every $t\geq 0$.

Because $g$ has polynomial growth, one can show that the expectation of the stochastic integral (that is, the first term) on the right side of (\ref{eqn:thm:VG_erg}) vanishes. Then, by taking the expectation over $\tilde{Z}(0)\sim \tilde{\pi},$ Equation (\ref%
{eqn:thm:VG_erg}) implies 
\begin{equation}
\mathbb{E}_{\tilde{\pi}}\left[ v(\tilde{Z}(0))\right] =\mathbb{E}_{\tilde{\pi%
}}\left[ v(\tilde{Z}(T))\right] +\mathbb{E}_{\tilde{\pi}}\left[
\int_{0}^{T}F(\tilde{Z}(t),g(\tilde{Z}(t)))\mathrm{d}t\right] + \mathbb{E}_{\tilde{\pi}}\left[\int_0^T  \kappa \cdot \mathrm{d}\tilde{Y}(t)\right] -T\xi .
\label{eqn:VT2:ergodic}
\end{equation}%
By observing that $\mathbb{E}_{\tilde{\pi}}[ v(\tilde{Z}(0))]  = \mathbb{E}_{\tilde{%
\pi}}[ v(\tilde{Z}(T))]$ and
\begin{eqnarray*}
\mathbb{E}_{\tilde{\pi}}\left[ F(\tilde{Z}(t),g(\tilde{Z}(t)))\right]  &=&%
\mathbb{E}\left[ F(\tilde{Z}(\infty ),g(\tilde{Z}(\infty )))\right] \text{
for }t\geq 0, \\
\mathbb{E}_{\tilde{\pi}}\left[\int_0^T  \kappa \cdot \mathrm{d}\tilde{Y}(t)\right] &=&T\mathbb{E}_{\tilde{\pi}}\left[\int_0^1  \kappa \cdot \mathrm{d}\tilde{Y}(t)\right],
\end{eqnarray*}%
we conclude that 
\begin{align*}
\xi& =\mathbb{E}[ F(\tilde{Z}(\infty ),g(\tilde{Z}(\infty )))]+\mathbb{E}_{\tilde{\pi}}\left[\int_0^1  \kappa \cdot \mathrm{d}\tilde{Y}(t)\right] \\
&=\mathbb{E}[ F(\tilde{Z}(\infty ),g(\tilde{Z}(\infty )))]+\sum_{i=1}^{d} \kappa_i \tilde{\nu}_i(S_i),
\end{align*} 
where $\tilde{\nu}_i (\cdot)$ is the boundary measure on the boundary surface $S_i = \{z \in \mathbb{R}_+^d : z_i =0 \}$ associated with the RBM under the reference policy. 
In other words, 
$\xi $ can be viewed as the expected steady-state cost
associated with the RBM under the reference policy, where $F\left( \cdot
,g(\cdot )\right) $ is the state-cost function.    {Now, consider the following PDE:}
\begin{equation}
\mathcal{L}\tilde{v}(z)-\tilde{\theta}\cdot \nabla \tilde{v}(z)+F\left(
z,g(z)\right) =\xi, \ \ z\in \mathbb{R}_{+}^{d},
\label{HJB:VG:ergo}
\end{equation}%
with boundary conditions $\mathcal{D}_i \tilde{v}(z) = -\kappa_i \text { if } z_i=0 \ \ (i=1,\ldots,d)$. 

   {We assume (\ref{HJB:VG:ergo}) has a $C^2$ solution with polynomial growth. Using Ito's lemma, one can then show that $\tilde{v}$ is the relative value function corresponding to the reference process $\tilde{Z}$, i.e. for policy $u(z) = \tilde{\theta}$ for $z \in \mathbb{R}_+^d$, under the state cost function $F(z,g(z))$.
}
Furthermore, applying Ito's formula to $\tilde{v}(\tilde{Z}(t))$    {on the interval $[0, nT]$ for $n \geq 1$} yields 
\begin{align}
\tilde{v}(\tilde{Z}(nT))- &\tilde{v}(\tilde{Z}(0))  \label{erg:ito2} \\
&=\int_{0}^{nT}\nabla \tilde{v}(\tilde{Z}(t)) \cdot \mathrm{d}W(t)+\int_{0}^{nT}%
\mathcal{D} \tilde{v}(\tilde{Z}(t)) \cdot \mathrm{d}\tilde{Y}(t)+\int_{0}^{nT}\left( 
\mathcal{L}\tilde{v}(\tilde{Z}(t))-\tilde{\theta}\cdot \nabla \tilde{v}(%
\tilde{Z}(t))\right) \mathrm{d}t.  \notag
\end{align}%
Since $\tilde{v}(z)$ also satisfies the boundary conditions $\mathcal{D}_i \tilde{v}(z) = - \kappa_i \text { if } z_i=0 \ \ (i=1,\ldots,d)$, it follows from Equations (\ref{RBM:Y1})-(\ref{RBM:Y2}) that 
\begin{equation*}
\int_{0}^{nT}\mathcal{D} \tilde{v}(\tilde{Z}(t)) \cdot \mathrm{d}\tilde{Y}(t)=-\int_{0}^{nT}\kappa \cdot \mathrm{d}\tilde{Y}(t).
\end{equation*}%
Then, substituting Equation (\ref{HJB:VG:ergo}) into Equation (\ref{erg:ito2}%
) gives 
\begin{equation}
\tilde{v}(\tilde{Z}(nT))-\tilde{v}(\tilde{Z}(0))=\int_{0}^{nT}\nabla \tilde{v}(%
\tilde{Z}(t))\cdot \mathrm{d}W(t)+n T\xi-\int_0^{nT}  \kappa \cdot \mathrm{d}\tilde{Y}(t)  -\int_{0}^{nT}F(\tilde{Z}(t),g(\tilde{Z}(t))) \,
\mathrm{d}t, \ n \geq 1.  \label{eqn:tilde_v}
\end{equation}%

In the proof of Proposition \ref{thm:doublepara:dis}, we first showed that, because $\tilde{Z}$ is a time-homogeneous Markov process, the assumed stochastic relationship (\ref{eqn:thm:VG}) can be extended to the more general form (\ref{eqn:aux2:proof:conv:disc:identity}) with $n$ an arbitrary positive integer. In the current context one can argue in exactly the same way to establish the following. First, the assumed stochastic relationship (\ref{eqn:thm:VG_erg}) actually holds in the more general form where $T$ is replaced by $nT$, with $n$ an arbitrary positive integer. And then, after taking expectations on both sides of the generalized version of (\ref{eqn:thm:VG_erg}) and (\ref{eqn:tilde_v}) we arrive at the following:
\begin{eqnarray}
v(z) &=&\mathbb{E}_{z}\left[ v(\tilde{Z}(nT))\right] -nT\xi + \mathbb{E}_z\left[\int_0^{nT}  \kappa \cdot \mathrm{d}\tilde{Y}(t) \right]+\mathbb{E}_{z}\left[
\int_{0}^{nT}F(\tilde{Z}(t),g(\tilde{Z}(t)))\mathrm{d}t\right] ,\text{ and}
\label{eqn:Tvn} \\
\tilde{v}(z) &=&\mathbb{E}_{z}\left[ \tilde{v}(\tilde{Z}(nT))\right] -nT\xi+ \mathbb{E}_z\left[\int_0^{nT}  \kappa \cdot \mathrm{d}\tilde{Y}(t) \right] +\mathbb{%
E}_{z}\left[ \int_{0}^{nT}F(\tilde{Z}(t),g(\tilde{Z}(t)))\mathrm{d}t\right] ,
\label{eqn:Tvntilde}
\end{eqnarray}%
for  $z\in \mathbb{R}_{+}^{d}$ and an arbitrary positive integer $n.$ Note that the expectation of the stochastic integral vanishes because $\nabla \tilde{v}$ has polynomial growth. Subtracting (\ref{eqn:Tvntilde}) from (\ref{eqn:Tvn}) further yields
\begin{equation*}
v(z)-\tilde{v}(z)=\mathbb{E}_{z}\left[ v(Z(nT))\right] -\mathbb{E}_{z}\left[ 
\tilde{v}(Z(nT))\right] .
\end{equation*}

   {Because the solution of (\ref{HJB:VG:ergo}) is determined only up to an additive constant}, we assume $\mathbb{E}\left[ \tilde{v}(\tilde{Z}%
(\infty ))\right] =\mathbb{E}\left[ v(\tilde{Z}(\infty ))\right] $. Since $v(\cdot)$ and $\tilde{v}(\cdot)$ have polynomial growth,    {we conclude the following from Theorem 4.12 of \citet{budhiraja-lee2007}:} 
\begin{eqnarray*}
\lim_{n\rightarrow +\infty }\mathbb{E}_{z}\left[ v(Z(nT))\right] &=&\mathbb{E%
}\left[ v(\tilde{Z}(\infty ))\right] \text{ and} \\
\lim_{n\rightarrow +\infty }\mathbb{E}_{z}\left[ \tilde{v}(Z(nT))\right] &=&\mathbb{E%
}\left[ \tilde{v}(\tilde{Z}(\infty ))\right] .
\end{eqnarray*}%
Therefore, we have 
\begin{equation*}
v(z)-\tilde{v}(z)=\lim_{n\rightarrow +\infty }\left( \mathbb{E}_{z}\left[
v(Z(nT))\right] -\mathbb{E}_{z}\left[ \tilde{v}(Z(nT))\right] \right) = 0 %
\text{ for }z\in \mathbb{R}_{+}^{d},
\end{equation*}
which means $v(\cdot )$ also satisfies the  PDE (\ref{HJB:VG:ergo}) and the associated boundary conditions. That is,  
\begin{align}
&\mathcal{L}v(z)-\tilde{\theta}\cdot \nabla v(z)+F\left( z,g(z)\right) =\xi ,%
\text{ for }z\in \mathbb{R}_{+}^{d}.  \label{HJB:VG:ergo2} \\
&\mathcal{D}_i v(z) = -\kappa_i \text { if } z_i=0 \ \ (i=1,\ldots,d). \label{HJB:aux:ergodic:BC}
\end{align}

Suppose that $g(\cdot )=\nabla v(\cdot )$ (which we will prove later).
Substituting this into Equation (\ref{HJB:VG:ergo2}) and using the
definition of $F$, it follows that 
\begin{equation*}
\mathcal{L}v(z)-\max_{\theta \in \Theta }\left\{ \theta \cdot \nabla
v(z)-c(z,\theta )\right\} =\xi, \ \ z\in \mathbb{R}_{+}^{d},
\end{equation*}%
which along with the boundary condition (\ref{HJB:aux:ergodic:BC}) gives the
desired result.

To complete the proof, it remains to show that $g(\cdot )=\nabla v(\cdot ).$
By applying Ito's formula to $v(\tilde{Z}(t))$ and using Equations (\ref{RBM:Y1})-(\ref{RBM:Y2}) and (\ref{HJB:aux:ergodic:BC}), we conclude that 
\begin{eqnarray*}
v(\tilde{Z}(T))-v(\tilde{Z}(0)) =\int_{0}^{T}\left( \mathcal{L}v(\tilde{Z}%
(t))-\tilde{\theta}\cdot \nabla v(\tilde{Z}(t))\right) \mathrm{d}%
t-  \int_0^T  \kappa \cdot \mathrm{d}\tilde{Y}(t)+\int_{0}^{T}\nabla v(\tilde{Z}(t)) \cdot \mathrm{d}W(t).
\end{eqnarray*}%
Then, using Equation (\ref{HJB:VG:ergo2}), we rewrite the preceding equation
as follows: 
\begin{equation*}
v(\tilde{Z}(T))-v(\tilde{Z}(0))=T\xi -\int_{0}^{T}F\left( \tilde{Z}(t),g(%
\tilde{Z}(t))\right) \mathrm{d}t-  \int_0^T  \kappa \cdot \mathrm{d}\tilde{Y}(t) +\int_{0}^{T}\nabla v(\tilde{Z}(t)) \cdot \mathrm{d}%
W(t).
\end{equation*}%
Comparing this with Equation (\ref{eqn:thm:VG_erg}) yields 
\begin{equation*}
\int_{0}^{T}\left( g(\tilde{Z}(t))-\nabla v(\tilde{Z}(t))\right) \cdot \mathrm{d}%
W(t)=0,
\end{equation*}
which yields the following:
\begin{equation}
\mathbb{E}_z\left[\left( \int_{0}^{T}\left( g(\tilde{Z}(t))-\nabla v(\tilde{Z}(t))\right) \cdot \mathrm{d}%
W(t)\right)^2 \right]=0.
\label{eq:=0_ergodic}
\end{equation}
Thus, provided $g(\tilde{Z}(t))-\nabla v(\tilde{Z}(t))$ is square integrable, Ito's isometry  \citep[Lemma D.1]{zhang2020wasserstein} gives the following:
\begin{align*}
    \mathbb{E}_z\left[\left( \int_{0}^{T}\left( g(\tilde{Z}(t))-\nabla v(\tilde{Z}(t))\right) \cdot \mathrm{d}%
W(t)\right)^2 \right] 
=\mathbb{E}_{z}\left[ \int_{0}^{T}\left \| g(\tilde{Z}(t))-\nabla v(\tilde{Z}%
(t))\right \|_A ^{2}\mathrm{d}t\right] \label{eq:itoiosmetry_ergodic}=0.
\end{align*}
The square integrability of $g(\tilde{Z}(t))-\nabla v(\tilde{Z}(t))$ follows because $g$ and $\nabla v$ have polynomial growth, and $\mathbb{E}_{z}\left( | \tilde{Z}(nT)| ^{k}\right) $  is finite for all $k$ because our action space $\Theta$ is bounded. 
Then, since $A$ is positive definite, $\nabla v(\tilde{Z}(t))=g(\tilde{Z}(t))$ almost surely. By the
continuity of $\nabla v\left( \cdot \right) $ and $g(\cdot ),$ we conclude
that $\nabla v(\cdot )=g(\cdot )$. \end{proof}

\section{Computational method}

\label{sec:han}

We follow in the footsteps of \citet{han2018solving}, who developed a
computational method to solve semilinear parabolic partial differential
equations (PDEs). Those authors focused on a backward stochastic
differential equation (BSDE) associated with their PDE, and in similar
fashion, we focus on the stochastic differential equations (\ref{eqn:thm:VG}%
) and (\ref{eqn:thm:VG_erg}) that are associated with our two stochastic
control formulations (see Section \ref{sec:sdes}). Our method differs from
that of \citet{han2018solving}, because they consider PDEs on a finite-time
interval with an unbounded state space and a specified terminal condition,
whereas our stochastic control problem has an infinite time horizon and
state space constraints. As such, it leads to a PDE on a polyhedral domain
with oblique derivative boundary conditions. We modify the approach of %
\citet{han2018solving} to incorporate those additional features, treating
the discounted and ergodic formulations in Sections \ref%
{sec:infinite-horizon:algo} and \ref{sec:ergodic:algo}, respectively.

\subsection{Discounted control}

\label{sec:infinite-horizon:algo} We approximate the value function $V(\cdot
)$ and its gradient $\nabla V(\cdot )$ by deep neural networks $%
V_{w_{1}}(\cdot )$ and $G_{w_{2}}(\cdot )$, respectively, with associated
parameter vectors $w_{1}$ and $w_{2}$. Seeking an approximate solution of
the stochastic equation (\ref{eqn:thm:VG}), we define the loss function 
\begin{eqnarray}
\ell (w_{1},w_{2}) &=&\mathbb{E}\left[ \left( e^{-rT}\,V_{w_{1}}(\tilde{Z}%
(T))-V_{w_{1}}(\tilde{Z}(0))+\int_{0}^{T}e^{-rt}\kappa \cdot \mathrm{d}%
\tilde{Y}(t)\right. \right.  \label{dis:loss:l} \\
&&\left. \left. -\int_{0}^{T}e^{-rt}G_{w_{2}}(\tilde{Z}(t))\cdot \mathrm{d}%
W(t)+\int_{0}^{T}e^{-rt}F(\tilde{Z}(t),G_{w_{2}}(\tilde{Z}(t)))\,\mathrm{d}%
t\right) ^{2}\right] .  \notag
\end{eqnarray}%
{  The initial state $\tilde{Z}(0)$ can be randomly chosen, and
the expectation in the previous equation is calculated with respect to the
sample path distribution of the reference process $\tilde{Z}(\cdot )$ (see
Algorithm \ref{algo:dis} for details). As mentioned in Step 6 of Algorithm %
\ref{algo:dis}, the process $\tilde{Z}(\cdot )$ is run continuously, meaning
that the terminal state of one iteration becomes the initial state of the
next iteration. This approach can be seen as an approximation to starting
the reference process with its steady-state distribution. }

{  \textbf{Remark:} The loss function $\ell (w_{1},w_{2})$ defined
in Equation (\ref{dis:loss:l}) corresponds to the $L_1(V)$ loss defined in
Equation (2.31) of \citet{zhou2021actor}. The authors consider this a
"variance reduced" loss function due to the additional term $%
\int_{0}^{T}e^{-rt}G_{w_{2}}(\tilde{Z}(t)) \cdot \mathrm{d}W(t)$. This term
can be interpreted as an approximating martingale process or a control
variate, a common approach in the simulation literature (see, for example, %
\citet{andradottir1993variance}, \citet{henderson2002approximating}, and %
\citet{dai2022queueing}). }

Our definition (\ref{dis:loss:l}) of the loss function does not explicitly
enforce the consistency requirement $\nabla V_{w_{1}}(\cdot
)=G_{w_{2}}(\cdot )$, but Proposition \ref{thm:doublepara:dis} provides the
justification for this separate parametrization. This type of double
parametrization has also been implemented by \citet{zhou2021actor}.

\makeatletter\renewcommand{\ALG@name}{Subroutine} \makeatother

\begin{algorithm}[!ht]
\caption{Euler discretization scheme}
\label{algo:euler}
\begin{algorithmic}[1]
\Require{The drift vector $-\tilde{\theta}$, the covariance matrix $A$, the reflection matrix $R$, the time horizon $T$, a step-size $h$ (for simplicity, we assume $N\triangleq T/h$ is an integer), and a starting point $\tilde{Z}(0) =z$.}
\Ensure{A discretized reflected Brownian motion with the boundary pushing process increments and the Brownian increments at times $h,2h,\ldots,Nh$. }
\Function{Discretize}{$T,h,z$}
\State For time interval $[0,T]$ and $N=T/h$, construct the partition $0=t_0<t_1<\ldots<t_N = T$, where $\Delta t_n = t_{n+1}-t_n = h$ for $n=0,1,\ldots,N-1$.
\State Generate $N$ i.i.d. $d$-dimensional Gaussian random variables with mean zero and covariance matrix $hA $, denoted by $\delta_0,\ldots, \delta_{N-1}$. 
\For{$k \gets 0$ to  $N-1$ }
\State $x \gets \tilde{Z}(k h) + \delta_k -\tilde{\theta} h$
\State { $\tilde{Z}((k + 1) h),  \Delta \tilde{Y}(kh)  \gets$ \Call{Skorokhod}{$x$}}
\EndFor
\State \Return  $\tilde{Z}(h),\tilde{Z}(2h),\ldots, \tilde{Z}(Nh)$; $\Delta \tilde{Y}(0)$,$\Delta \tilde{Y}(h),\ldots,\Delta \tilde{Y}((N-1)h)$; and $\delta_0$,\ldots,$\delta_{N-1}$.
\EndFunction
\end{algorithmic}
\end{algorithm}

\begin{algorithm}
\caption{Solve the Skorokhod problem (linear complementarity problem)}
\label{algo:Skorokhod}
\begin{algorithmic}[1]
\Require{A vector $x\in \mathbb{R}^d$ and the reflection matrix $R$.}
\Ensure{A solution to the Skorokhod problem $y\in \mathbb{R}^d_+$}
\State Set $\epsilon = 10^{-8}$;
\Function{Skorokhod}{$x$}
\State ${y}={x}$;
\State {$u= 0$;}
			\While{Exists ${y}_i<-\epsilon$}
			\State Compute the set $B = \{i: {y}_i<\epsilon$\};
			\State Compute ${L}_B=-R_{B,B}^{-1} {x}_B$;
\State Compute ${y} = {x} + R_{:,B} \times {L}_B$;
\EndWhile
\State { $u_B=L_B$;
\State\Return $y$, $u$.}
\EndFunction
		\end{algorithmic}
\end{algorithm}\makeatletter\renewcommand{\ALG@name}{Algorithm} 
\makeatother
Our computational method seeks a neural network parameter combination $%
(w_{1},w_{2})$ that minimizes an approximation of the loss defined via (\ref%
{dis:loss:l}). Specifically, we first simulate multiple discretized paths of
the reference RBM $\tilde{Z}$ with the boundary pushing process $\tilde{Y}$,
restricted to a fixed and finite time domain $[0,T]$. To do that, we sample
discretized paths of the underlying Brownian motion $W$, and then solve a
discretized Skorohod problem for each path of $W$ (this is the purpose of
Subroutine \ref{algo:Skorokhod}) to obtain the corresponding path of $\{ 
\tilde{Z},\tilde{Y}\} $. Thereafter, our method computes a discretized
version of the loss (\ref{dis:loss:l}), summing over sampled paths to
approximate the expectation and over discrete time steps to approximate the
integral over $[0,T]$, and minimizes it using stochastic gradient descent;
see Algorithm \ref{algo:dis}. {  The discretization inevitably
introduces bias into our method. While we do not provide a formal
convergence proof as the partition becomes finer, we direct readers to %
\citet{han2020convergence} for a rigorous analysis in a similar context.
Furthermore, our numerical examples indicate that the impact of
discretization on the learned control is small.}

In Subroutine \ref{algo:Skorokhod}, given the index set $B$, $R_{B, B}$ is
the submatrix derived by deleting the rows and columns of $R$ with indices
in $\{1,\ldots,d \}$\textbackslash$B$. Similarly, $R_{:,B}$ is the matrix
that one arrives at by deleting the columns of $R$ whose indices are in the
set $\{1,\ldots,d \}$\textbackslash$B$. {  One has considerable
latitude in choosing $\tilde{\theta}$, i.e., the reference policy, provided
it can explore the state space sufficiently. For our numerical examples, we
set $\tilde{\theta} = 1$.}

\begin{algorithm}[!ht]
\caption{Method for the discounted control case}
\label{algo:dis}
\begin{algorithmic}[1]
\Require{The number of iteration steps $M$, a batch size $B$, a learning rate $\alpha$, a time horizon $T$, a discretization step-size $h$ (for simplicity, we assume $N\triangleq T/h$ is an integer),   a starting point $z$, and an optimization solver (SGD, ADAM, RMSProp, etc).}
\Ensure{A neural network approximation of the value function $ {V}_{w_1}$ and the gradient function $G_{w_2}$.}
\State Initialize the neural networks $ {V}_{w_1}$ and  $G_{w_2}$; set $z_0^{(i)}=z$ for $i=1,2,...,B$.
\For{$k \gets 0$ to  $M-1$ }
\State Simulate $B$ discretized RBM paths and the Brownian increments $\{\tilde{Z}^{(i)},\Delta \tilde{Y}^{(i)},\delta^{(i)}\}$ with a time horizon $T$ and a discretization step-size $h$ starting from $\tilde{Z}^{(i)}(0)=z_k^{(i)}$ by invoking \Call{Discretize}{$T,h,z_k^{(i)}$} for $i=1,2,...,B$.  
\State Compute the empirical loss 
\begin{eqnarray}
	\hat{\ell}(w_1,w_2) &=&\frac{1}{B}\sum_{i=1}^{B}\left( e^{-rT}V_{w_1}(%
	\tilde{Z}^{(i)}(T))-V_{w_1}(\tilde{Z}^{(i)}(0))  +\sum_{j=0}^{N-1}e^{-rhj}\kappa \cdot \Delta\tilde{Y}^{(i)}(hj)      \right.  \\
	&&\left.  -\sum_{j=0}^{N-1}e^{-rhj}G_{w_2}(\tilde{Z}^{(i)}(hj)) \cdot 
	\delta_j^{(i)}+\sum_{j=0}^{N-1}e^{-rhj}F(\tilde{Z}^{(i)}(hj),G_{w_2}(\tilde{Z}^{(i)}(hj)))%
	h\right) ^{2} .  \notag
\end{eqnarray}
\State Compute the gradient $\partial{\hat{\ell}(w_1,w_2)}/\partial w_1,\partial{\hat{\ell}(w_1,w_2)}/\partial w_2 $ and update $w_1,w_2$ using the chosen optimization solver.
\State Update $z_{k+1}^{(i)}$ as the end point of the path $\tilde{Z}^{(i)}$: $z_{k+1}^{(i)}\gets \tilde{Z}^{(i)}(T)$.
\EndFor
\State \textbf{return}  Functions $ V_{w_1} (\cdot)$ and  $G_{w_2} (\cdot)$.

\end{algorithmic}
\end{algorithm}

After the parameter values $w_1$ and $w_2$ have been determined, our
proposed policy is as follows: 
\begin{equation}
\theta _{w_{2}}(z)=\arg \max_{\theta \in \Theta }\left\{ \theta \cdot
G_{w_{2}}(z)-c(z,\theta )\right\}, \ z\in \mathbb{R}_+^d .
\label{eqn:optimal_policy_G}
\end{equation}

\begin{remark}
One can also consider the policy using $\nabla V_{w_{1}}(\cdot )$ instead of 
$G_{w_{2}}(\cdot ).$ That is, 
\begin{equation}
\arg \max_{\theta \in \Theta }\left\{ \theta \cdot \nabla
V_{w_{1}}(z)-c(z,\theta )\right\} , \ z\in \mathbb{R}_+^d .
\label{eqn:optimal_policy_V}
\end{equation}%
However, our numerical experiments suggest that this policy is inferior to (%
\ref{eqn:optimal_policy_G}).
\end{remark}

\subsection{Ergodic control}

\label{sec:ergodic:algo} We parametrize $v(\cdot )$ and $\nabla v(\cdot )$
using deep neural networks $v_{w_{1}}(\cdot )$ and $g_{w_{2}}(\cdot )$ with
parameters $w_{1}$ and $w_{2}$, respectively, and then use Equation (\ref%
{eqn:thm:VG_erg}) to define an auxiliary loss function 
\begin{eqnarray}
\tilde{\ell}(w_{1},w_{2},\xi ) &=&\mathbb{E}\left[ \left( v_{w_{1}}(\tilde{Z}%
(T))-v_{w_{1}}(\tilde{Z}(0))+\int_{0}^{T}\kappa \cdot \mathrm{d}\tilde{Y}%
(t)\right. \right. \\
&&\left. \left. -\int_{0}^{T}g_{w_{2}}(\tilde{Z}(t))\mathrm{d}W(t)-T\xi
+\int_{0}^{T}F\left( \tilde{Z}(t),g_{w_{2}}\left( \tilde{Z}(t)\right)
\right) \mathrm{d}t\right) ^{2}\right] .  \notag
\end{eqnarray}

Then,    {we define} the loss function $\ell
(w_{1},w_{2})=\min_{\xi }\ell (w_{1},w_{2},\xi )$. 
   {Letting
$X$ denote a random variable with a finite second moment, we note that} 
\begin{align*}
\mathrm{Var}\left( X\right) = \min_{\xi }\mathbb{E[}\left( X-\xi \right) ^{2}%
\mathbb{]}.
\end{align*}
   {Thus,} we arrive at the following expression for the loss
function 
\begin{align}
&\ell (w_{1},w_{2})  \label{loss:ergodic} \\
=&\mathrm{Var}\left( v_{w_{1}}(\tilde{Z}(T))-v_{w_{1}}(\tilde{Z}%
(0))+\int_{0}^{T}\kappa \cdot \mathrm{d}\tilde{Y}(t)-\int_{0}^{T}g_{w_{2}}(%
\tilde{Z}(t))\mathrm{d}W(t)+\int_{0}^{T}F\left( \tilde{Z}(t),g_{w_{2}}\left( 
\tilde{Z}(t)\right) \right) \mathrm{d}t)\right) .  \notag
\end{align}
We present our method for the ergodic control case formally in Algorithm \ref%
{algo:ergo}. 
\begin{algorithm}[ht]
	\caption{Method for the ergodic control case}
	\label{algo:ergo}
	\begin{algorithmic}[1]
		\Require{The number of iteration steps $M$, a batch size $B$, a learning rate $\alpha$, a time horizon $T$, a discretization step-size $h$ (for simplicity, we assume $N\triangleq T/h$ is an integer),   a starting point $z$, and an optimization solver (SGD, ADAM, RMSProp, etc).}
		\Ensure{A neural network approximation of the value function $ {v}_{w_1}$ and the gradient function $g_{w_2}$.}
		\State Initialize the neural networks $ {v}_{w_1}$ and  $g_{w_2}$; set $z_0^{(i)}=z$ for $i=1,2,...,B$.
		\For{$k \gets 0$ to  $M-1$ }
		\State Simulate $B$ discretized RBM paths and the Brownian increments $\{\tilde{Z}^{(i)},\Delta \tilde{Y}^{(i)},\delta^{(i)}\}$ with a time horizon $T$ and a discretization step-size $h$ starting from $\tilde{Z}^{(i)}(0)=z_k^{(i)}$ by invoking \Call{Discretize}{$T,h,z_k^{(i)}$}, for $i=1,2,...,B$.   
		\State Compute the empirical loss 
	\begin{eqnarray}
		\hat{\ell}(w_1,w_2) &=&\widehat{\mathrm{Var}}\left( v_{w_1}(\tilde{Z}%
		^{(i)}(T))-v_{w_1}(\tilde{Z}^{(i)}(0))-\sum_{j=0}^{N-1}g_{w_2}(\tilde{Z}%
		^{(i)}(hj)) \cdot \delta^{(i)}_j \right.  \\
		&&\left. +\sum_{j=0}^{N-1}\kappa \cdot \Delta\tilde{Y}^{(i)}(hj)  +\sum_{j=0}^{N-1} f\left( \tilde{Z}^{(i)}(hj),g_{w_2}\left( \tilde{Z}%
		^{(i)}(hj)\right) \right) h\right) .  \notag
	\end{eqnarray}%
		\State Compute the gradient $\partial{\hat{\ell}(w_1,w_2)}/\partial w_1,\partial{\hat{\ell}(w_1,w_2)}/\partial w_2 $ and update $w_1,w_2$ using the chosen optimization solver.
		\State Update $z_{k+1}^{(i)}$ as the end point of the path $\tilde{Z}^{(i)}$: $z_{k+1}^{(i)}\gets \tilde{Z}^{(i)}(T)$.
		\EndFor
		\State \textbf{return}  Functions $ v_{w_1} (\cdot)$ and  $g_{w_2} (\cdot)$.
		
	\end{algorithmic}
\end{algorithm}

After the parameters values $w_1$ and $w_2$ have been determined, our
proposed policy is the following: 
\begin{equation*}
\bar{\theta}_{w_{2}}(z)=\arg \max_{\theta \in {\Theta }}\left( \theta \cdot
g_{w_{2}}(z)-c(z,\theta )\right) , \ z\in \mathbb{R}_+^d.
\end{equation*}

\section{Three families of test problems}

\label{sec:example}

Here we specify three families of test problems for which numerical results
will be presented later (see Section \ref{sec:numerical}). Each family
consists of RBM drift control problems indexed by $d = 1, 2, \ldots,$ where $%
d$ is the dimension of the orthant that serves as the problem's state space.
The first of the three problem families, specified in Section \ref%
{subsec:linearcost:feedforwardnetwork}, is characterized by a feed-forward
network structure and linear cost of control. Recapitulating earlier work by %
\citet{ata2006dynamic}, Section \ref{subsec:formal-HT-limit} explains the
interpretation of such problems as ``heavy traffic'' limits of input control
problems for certain feed-forward queueing networks.

Our second family of test problems is identical to the first one except that
now the cost of control is quadratic rather than linear. The exact meaning
of that phrase will be spelled out in Section \ref{sec:pricing-example},
where we also explain the interpretation of such problems as heavy traffic
limits of dynamic pricing problems for queueing networks. In Section \ref%
{sec:benchmark}, we describe two parametric families of policies with
special structure that will be used later for comparison purposes in our
numerical study. Finally, Section \ref{sec:deco_example} specifies our third
family of test problems, which have a separable structure that allows them
to be solved exactly by analytical means. Such problems are of obvious value
for evaluating the accuracy of our computational method. %
   {For all our test problems, the penalty rates associated
with pushing at the boundary are set to zero. That is, $\kappa = 0$.}

\subsection{Main example with linear cost of control}

\label{subsec:linearcost:feedforwardnetwork} We consider a family of test
problems with parameters $K=0,1,\ldots$, attaching to each such problem the
index $d$ (mnemonic for \textit{dimension}) $=K+1$. Problem $d$ has state
space $\mathbb{R}_+^d$ and the $d \times d$ reflection matrix 
\begin{equation}
R=\left[ 
\begin{array}{cccc}
1 &  &  &  \\ 
-p_{1} & 1 &  &  \\ 
\vdots &  & \ddots &  \\ 
-p_{K} &  &  & 1%
\end{array}%
\right] ,  \label{eqn:defn:reflection-matrix}
\end{equation}%
where $p_1, \ldots, p_K > 0$ and $p_1 + \cdots + p_K =1$. Also, the set of
drift vectors available in each state is 
\begin{align}
\Theta =\prod\limits_{k=0}^{K}\left[ \underline{\theta }_{k},\overline{%
\theta }_{k}\right] .  \label{eqn:defn:action:space}
\end{align}
where the lower limit $\underline{\theta}_k$ and upper limit $\overline{%
\theta}_k$ are as specified in Section \ref{subsec:formal-HT-limit} below.
Similarly, the $d \times d$ covariance matrix $A$ for problem $d$ is as
specified in Section \ref{subsec:formal-HT-limit}. Finally, the cost
function for problem $d$ has the linear form 
\begin{equation}
c(z,\theta )=h^{\top }z+c^{\top }\theta \ \text{ where } \ h, \, c \in 
\mathbb{R}_+^d.  \label{eqn:defn:linear:cost}
\end{equation}
That is, the cost rate $c(Z(t),u(Z(t)))$ that the system manager incurs
under policy $u$ at time $t$ is linear in both the state vector $Z(t)$ and
the chosen drift rate $u(Z(t))$.

In either the discounted control setting or the ergodic control setting,
inspection of the HJB equation displayed earlier in Section \ref{sec:problem}
shows that, given this linear cost structure, there exists an optimal policy 
$u^*(\cdot)$ such that 
\begin{align}
\text{either } \ u_k^*(z) = \underline{\theta}_k \ \text{ or } \ u_k^*(z) = 
\overline{\theta}_k  \label{eqn:bang-bang-policy}
\end{align}
for each state $z \ \in \mathbb{R}_+^{K+1}$ and each component $%
k=0,1,\ldots,K$.

In the next section we explain how drift control problems of the form
specified here arise as heavy traffic limits in queueing theory. Strictly
speaking, however, that interpretation of the test problems is inessential
to the main subject of this paper: the computational results presented in
Section \ref{sec:numerical} can be read without reference to the queueing
theoretic interpretations of our test problems.

\subsection{Interpretation as heavy traffic limits of queueing network
control problems}

\label{subsec:formal-HT-limit}

Let us consider the feed-forward queueing network model of a make-to-order
production system portrayed in Figure \ref{fig:exp1}. There are $d=K+1$
buffers, represented by the open-ended rectangles, indexed by $k=0,1,\ldots
,K.$ Each buffer has a dedicated server, represented by the circles in
Figure \ref{fig:exp1}. Arriving jobs wait in their designated buffer if the
server is busy. There are two types of jobs arriving to the system: regular
versus thin streams. Thin stream jobs have the same service time
distributions as the regular jobs, but they differ from the regular jobs in
two important ways: First, thin stream jobs can be turned away upon arrival.
That is, a system manager can exercise admission control in this manner, but
in contrast, she must admit all regular jobs arriving to the system. Second,
the volume of thin stream jobs is smaller than that of the regular jobs; see
Assumption \ref{assumption:heavy-traffic}. 
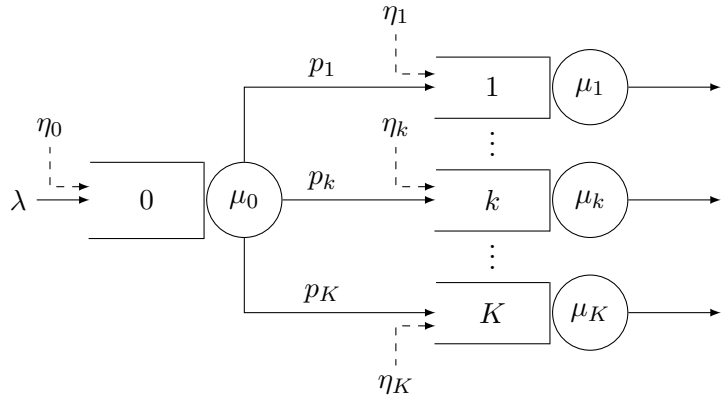
\begin{figure}[th]
\centering
\begin{tikzpicture}[start chain=going right,>=latex,node distance=1pt]
		\node[three sided,minimum width=1.5cm,minimum height = 1cm,on chain] (wa) {$0$};

		\node[draw,circle,on chain,minimum size=1cm] (se) {$\mu_0$};
		
		%
		\node[three sided,minimum width=1.5cm,minimum height =0.8cm,on chain]  (wa2)[right =of 
		se,xshift=2cm,yshift=1.5cm]  
		{$1$};
		
		\node[draw,circle,on chain,minimum size=1cm] (se2) {$\mu_1$};

		\node[three sided,minimum width=1.5cm,minimum height = 0.8cm,on chain]  (wamiddle)[right =of 
		se,xshift=2cm]  {$k$};
		
		\node[draw,circle,on chain,minimum size=1cm] (semiddle) {$\mu_k$};

		\node[three sided,minimum width=1.5cm,minimum height =0.8cm,on chain]  (wa3)[right =of 
		se,xshift=2cm,yshift=-1.5cm]  
		{$K$};

		\node[draw,circle,on chain,minimum size=1cm] (se3) {$\mu_K$};
		
		\draw[->] (se) |-   node[above , at end, xshift = -1.5cm] {$p_{1}$} (wa2.west);
		\draw[->](se) |-  node[above, at end, xshift = -1.5cm] {$p_{K}$} (wa3.west);
		\draw[->](se)  edge node[above , at end, xshift = -1.5cm] {$p_{k}$} (wamiddle.west);
		
		\draw[<-,dashed] (wa2.west) +(0pt,5pt) -|  +(-15pt,20pt)  node[above , at end = -0.95cm] {$\eta_1$};
		\draw[<-,dashed] (wamiddle.west) +(0pt,5pt) -|  +(-15pt,20pt)  node[above , at end = -0.95cm] {$\eta_k$};
		\draw[<-,dashed] (wa3.west) +(0pt,-5pt) -|  +(-15pt,-20pt)  node[below , at end = 0.95cm] {$\eta_K$};

		\draw[->] (se2.east) -- node[above] {}  +(35pt,0);
		\draw[->] (semiddle.east) -- node[above] {}  +(35pt,0);
		\draw[->] (se3.east) --node[above] {}   +(35pt,0);
		
		\draw[<-,dashed] (wa.west) +(0pt,5pt) -|  +(-15pt,20pt)  node[above , at end = -0.95cm] {$\eta_0$};
		
		\draw[<-] (wa.west) -- +(-20pt,0) node[left] {$\lambda$};
		\path (wa2.south) -- node[auto=false,yshift = 0.1cm]{\vdots} (wamiddle.north);
		\path (wamiddle.south) -- node[auto=false,yshift = 0.1cm]{\vdots} (wa3.north);
	\end{tikzpicture}
\caption{A feedforward queueing network with thin arrival streams.}
\label{fig:exp1}
\end{figure}

Regular jobs enter the systems only through buffer zero, as shown by the
solid arrow pointing to buffer zero in Figure \ref{fig:exp1}. A renewal
process $E=\{E(t):t\geq 0\}$ models the cumulative number of regular jobs
arriving to the system over time. We let $\lambda $ denote the arrival rate
and $a^{2}$ denote the squared coefficient of variation of the interarrival
times for the regular jobs. The thin stream jobs arrive to buffer $k$ (as
shown by the dashed arrows in Figure \ref{fig:exp1}) according to the
renewal process $A_{k}=\{A_{k}(t):t\geq 0\}$ for $k=0,1,\ldots ,K$. We let $%
\eta _{k}$ denote the arrival rate and $b_{k}^{2}$ denote the squared
coefficient of variation of the interarrival times for renewal process $%
A_{k} $.

Jobs in buffer $k$ have i.i.d. general service time distributions with mean $%
m_{k}$ and squared coefficient of variation $s_{k}^{2}\geq 0, \ k=0,1,\ldots
,K$; $\mu _{k}=1/m_{k}$ is the corresponding service rate. We let $S_k =
\{S_k(t): t \geq 0 \} $ denote the renewal process associated with the
service completions by server $k$ for $k=1, \ldots, K$. To be specific, $%
S_k(t)$ denotes the number of jobs server $k$ processes by time $t$ if it
incurs no idleness during $[0,t]$. The jobs in each buffer are served on a
first-come-first-served (FCFS) basis, and servers work continuously unless
their buffer is empty. After receiving service, jobs in buffer zero join
buffer $k$ with probability $p_{k}, \ k=1,2,\ldots ,K$, independently of
other events. This probabilistic routing structure is captured by a
vector-valued process $\Phi (\cdot )$ where $\Phi_k (\ell )$ denotes the
total number of jobs routed to buffer $k$ among the first $\ell $ jobs
served by server zero for $k=1,\ldots ,K$ and $\ell \geq 1$. We let $%
p=(p_{k})$ denote the $K$-dimensional vector of routing probabilities. Jobs
in buffers $1,\ldots ,K$ leave the system upon receiving service.

As stated earlier, the system manager makes admission control decisions for
thin stream jobs. Turning away a thin stream job arriving to buffer $k$
(externally) results in a penalty of $c_k$. For mathematical convenience, we
model admission control decisions as if the system manager can simply ``turn
off'' each of the thin stream arrival processes as desired. In particular,
we let $\Delta_k(t)$ denote the cumulative amount of time that the
(external) thin stream input to buffer $k$ is turned off during the interval 
$[0,t]$. Thus, the vector-valued process $\Delta=(\Delta_k)$ represents the
admission control policy. Similarly, we let $T_k(t)$ denote the cumulative
amount of time server $k$ is busy during the time interval $[0,t]$, and $%
I_k(t)=t-T_k(t)$ denotes the cumulative amount of idleness that server $K$
incurs during $[0,t]$.

Letting $Q_{k}(t)$ denote the number of jobs in buffer $k$ at time $t$, the
vector-valued process $Q=(Q_{k})$ will be called the queue-length process.
Given a control $\Delta =(\Delta _{k})$, assuming $Q(0)=0$, it follows that%
\begin{eqnarray}
Q_{0}(t) &=&E(t)+A_{0}(t-\Delta _{0}(t))-S_{0}(T_{0}(t))\geq 0,\text{ } \
t\geq 0,  \label{eqn:newtork:Q0} \\
Q_{k}(t) &=&A_{k}(t-\Delta _{k}(t))+\Phi _{k}\left( S_{0}(T_{0}(t))\right)
-S_{k}(T_{k}(t))\geq 0,\text{ } \ t\geq 0,\text{ } \ k=1,\ldots ,K.
\label{eqn:newtork:Qk}
\end{eqnarray}%
Moreover, the following must hold: 
\begin{eqnarray}
&&I(\cdot )\text{ is continuous and nondecreasing with }I(0)=0,
\label{eqn:network:I_cont} \\
&&I_{k}(\cdot )\text{ only increases at those times $t$ when }Q_{k}(t)=0, \
\ k=0,1,\ldots ,K,  \label{eqn:network:I_k} \\
&&\Delta _{k}(t)-\Delta _{k}(s)\leq t-s, \ \ 0\leq s\leq t<\infty, \ \
k=0,1,\ldots ,K.  \label{eqn:network:Delta} \\
&&I,\Delta \text{ are non-anticipating.}  \label{eqn:network:non-anti}
\end{eqnarray}%
The system manager also incurs a holding cost at rate $h_{k}$ per job in
buffer $k$ per unit of time. We use the processes $\xi =\left\{ \xi
(t),t\geq 0\right\} $ as a proxy for the cumulative cost under a given
admission control policy $\Delta \left( \cdot \right) ,$ where 
\begin{equation*}
\xi (t)=\sum_{k=0}^{K}c_{k}\eta _{k}\Delta
_{k}(t)+\sum_{k=0}^{K}\int_{0}^{t}h_{k}Q_{k}(s) \, \mathrm{d}s,\ \ t\geq 0.
\end{equation*}%
This is an approximation of the realized cost because the first term on the
right-hand side replaces the admission control penalties actually incurred
with their means.

In order to derive the approximating Brownian control problem, we consider a
sequence of systems indexed by a system parameter $n=1,2,\ldots ;$ we attach
a superscript of $n$ to various quantities of interest. Following the
approach used by \citet{ata2006dynamic}, we assume that the sequence of
systems satisfies the following heavy traffic assumption.

\begin{assumption}
\label{assumption:heavy-traffic} For $n\geq 1,$ we have that 
\begin{equation*}
\lambda ^{n}=n\lambda ,\eta _{k}^{n}=\eta _{k}\sqrt{n}\text{ and }\mu
_{k}^{n}=n\mu _{k}+\sqrt{n}\beta _{k}, \ \ k=0,1,\ldots ,K,
\end{equation*}%
where $\lambda ,$ $\mu _{k},\eta _{k}$ and $\beta _{k}$ are nonnegative
constants. Moreover, we assume that 
\begin{equation*}
\lambda =\mu _{0}=\frac{\mu _{k}}{p_{k}} \ \text{ for } \ k=1,\ldots ,K.
\end{equation*}
\end{assumption}

One starts the approximation procedure by defining suitably centered and
scaled processes. For $n\geq 1,$ we define%
\begin{alignat*}{3}
\hat{E}^{n}(t) &=\frac{E^{n}(t)-\lambda ^{n}t}{\sqrt{n}} \quad & \text{ and }%
&\quad\hat{\Phi}^{n}(q)=\frac{\Phi \left( \lbrack nq]\right) -p([nq])}{\sqrt{%
n}},\text{ } \ t\geq 0,\text{ } \ q\geq 0, &  \\
\hat{A}_{k}^{n}(t) &=\frac{A_k^{n}(t)-\eta_{k}^nt}{\sqrt{n}} & \text{ and }%
&\quad\hat{S}_{k}^{n}(t)=\frac{\text{ }S_{k}^{n}(t)-\mu _{k}^{n}t}{\sqrt{n}}%
,\ t\geq 0,\ k=0,1,\ldots ,K, &  \\
\hat{Q}^{n}(t) &=\frac{Q^{n}(t)}{\sqrt{n}} & \text{ and }&\quad\hat{\xi}%
^{n}(t)=\frac{\xi ^{n}(t)}{\sqrt{n}}, \ t\geq 0. & 
\end{alignat*}%
In what follows, we assume 
\begin{equation}
T_{k}^{n}(t)=t-\frac{1}{\sqrt{n}}I_{k}(t)+o\left( \frac{1}{\sqrt{n}}\right)
, \ t\geq 0,\ k=0,1,\ldots ,K,  \label{eqn:network:T}
\end{equation}%
where $I_{k}(\cdot )$ is the limiting idleness process for server $k$; see %
\citet{harrison1988brownian} for an intuitive justification of (\ref%
{eqn:network:T}).

Then, defining 
\begin{eqnarray*}
\chi _{0}^{n}(t) &=&\hat{E}^{n}(t)+\hat{A}_{0}^{n}(t-\Delta _{0}(t))-\hat{S}%
_{0}^{n}(T_{0}^{n}(t)),\text{ }t\geq 0, \\
\chi _{k}^{n}(t) &=&\hat{A}_{k}^{n}(t-\Delta _{k}(t))+\hat{\Phi}^n
_{k}\left( \frac{1}{n}\hat{S}_{0}^{n}(T_{0}^{n}(t))\right) +p_{k}\hat{S}%
_{0}^{n}(T_{0}^{n}(t)) \\
&& \quad -\hat{S}_{k}^{n}(T_{k}^{n}(t)),\text{ }t\geq 0,\text{ }k=1,2,\ldots
,K,
\end{eqnarray*}%
and using Equations (\ref{eqn:newtork:Q0}) - (\ref{eqn:newtork:Qk}) and (\ref%
{eqn:network:T}), it is straightforward to derive the following for $t\geq 0$
and $k=1,\ldots ,K:$%
\begin{eqnarray}
\hat{Q}_{0}^{n}(t) &=&\chi _{0}^{n}(t)+\left( \eta _{0}-\beta _{0}\right)
t-\eta _{0}\Delta _{0}(t)+\mu _{0}I_{0}(t)+o(1),
\label{eqn:network:Q_0_tilde} \\
\hat{Q}_{k}^{n}(t) &=&\chi _{k}^{n}(t)+\left( \eta _{k}+p_{k}\beta
_{0}-\beta _{k}\right) t-\eta _{k}\Delta _{k}(t)+\mu _{k}I_{k}(t)-p_{k}\mu
_{0}I_{0}(t)+o(1).  \label{eqn:network:Q_k_tilde}
\end{eqnarray}%
Moreover, it follows from Equation (\ref{eqn:network:Delta}) that $\Delta
_{k}(t)$ is absolutely continuous. We denote its density by $\delta
_{k}(\cdot ),i.e.,$%
\begin{equation*}
\Delta _{k}(t)=\int_{0}^{t}\delta _{k}(s)\mathrm{d}s, \ t\geq 0,\text{ }%
k=0,1,\ldots ,K,
\end{equation*}%
where $\delta _{k}(t)\in \lbrack 0,1].$\ Using this, we write 
\begin{equation}
\hat{\xi}^{n}(t)=\sum_{k=0}^{K}\int_{0}^{t}c_{k}\eta _{k}\delta _{k}(s)%
\mathrm{d}s+\sum_{k=0}^{K}\int_{0}^{t}h_{k}\hat{Q}_{k}^{n}(s)\mathrm{d}%
s,t\geq 0.  \label{eqn:network:xi}
\end{equation}%
Then passing to the limit formally as $n\rightarrow \infty ,$ and denoting
the weak limit of $( \hat{Q}^{n},\hat{X}^{n},\hat{\xi}^{n}) $ by $\left(
Z,\chi ,\xi \right) ,$ where $\chi $ is a $(K+1)$-dimensional driftless
Brownian motion with covariance matrix (see Appendix \ref{appendix:var} for
its derivation) 
\begin{equation*}
A=\mu _{0}\left[ 
\begin{array}{ccccc}
s_{0}^{2}+a^{2} & -p_{1}s_{0}^{2} & \cdots & \cdots & -p_{K}s_{0}^{2} \\ 
-p_{1}s_{0}^{2} & p_{1}(1-p_{1})+p_{1}^{2} s_{0}^{2}+p_{1} s_{1}^{2} & 
p_{1}p_{2}\left( s_{0}^{2}-1\right) & \cdots & p_{1}p_{K}\left(
s_{0}^{2}-1\right) \\ 
\vdots & p_{1}p_{2}\left( s_{0}^{2}-1\right) & \ddots &  & \vdots \\ 
\vdots & \vdots &  & \ddots & p_{K-1}p_{K}\left( s_{0}^{2}-1\right) \\ 
-p_{K}s_{0}^{2} & p_{1}p_{K}\left( s_{0}^{2}-1\right) & \cdots & \cdots & 
p_{K}(1-p_{K})+p_{K}^{2}s_{0}^{2}+p_Ks_{K}^2%
\end{array}%
\right] ,
\end{equation*}%
we deduce from (\ref{eqn:network:Q_0_tilde}) - (\ref{eqn:network:Q_k_tilde})
and (\ref{eqn:network:xi}) that%
\begin{eqnarray*}
Z_{0}(t) &=&\chi _{0}(t)+(\eta _{0}-\beta _{0})t-\int_{0}^{t}\eta _{0}\delta
_{0}(s)\mathrm{d}s+\mu _{0}I_{0}(t), \\
Z_{k}(t) &=&\chi _{k}(t)+\left( \eta _{k}+p_{k}\beta _{0}-\beta _{k}\right)
t-\int_{0}^{t}\eta _{k}\delta _{k}(s)\mathrm{d}s+\mu _{k}I_{k}(t) - p_k
\mu_0 I_0(t),\ k=1,\ldots ,K, \\
\xi (t) &=&\sum_{k=0}^{K}\int_{0}^{t}c_{k}\eta _{k}\delta _{k}(s)\mathrm{d}%
s+\sum_{k=0}^{K}\int_{0}^{t}h_{k}Z_{k}(s)\mathrm{d}s.
\end{eqnarray*}%
In order to streamline the notation, we make the following change of
variables: 
\begin{eqnarray*}
Y_{k}(t) &=&\mu _{k}I_{k}(t),\text{ }k=0,1, \ldots ,K, \\
\theta _{0}(t) &=&\eta _{0}\delta _{0}(t)-(\eta _{0}-\beta _{0}),\text{ }%
t\geq 0, \\
\theta _{k}(t) &=&\eta _{k}\delta _{k}(t)-\left( \eta _{k}+p_{k}\beta
_{0}-\beta _{k}\right)
\end{eqnarray*}%
and let 
\begin{eqnarray*}
\underline{\theta }_{0} &=&\beta _{0} \ \text{ and } \ \overline{\theta }%
_{0}=\beta _{0}-\eta _{0}, \\
\underline{\theta }_{k} &=&\beta _{k}-p_{k}\beta _{0} \ \text{ and } \ 
\overline{\theta }_{k}=\beta _{k}-\eta _{k}-p_{k}\beta _{0},\ k=1,\ldots ,K.
\end{eqnarray*}%
Lastly, we define the set of negative drift vectors available to the system
manager as in Equation (\ref{eqn:defn:action:space}). As a result, we arrive
at the following Brownian system model: 
\begin{eqnarray}
Z_{0}(t) &=&\chi _{0}(t)-\int_{0}^{t}\theta _{0}(s)\mathrm{d}s+Y_{0}(t),\ \
t\geq 0 ,  \label{eqn:network:Z_0} \\
Z_{k}(t) &=&\chi _{k}(t)-\int_{0}^{t}\theta _{k}(s)\mathrm{d}s + Y_{k}(t) -
p_{k}Y_{0}(t), \ \ k=1,\ldots ,K,  \label{eqn:network:Z_k}
\end{eqnarray}%
which can be written as in Equation (\ref{RBM:Z}) with $d=K+1,$ where the
reflection matrix $R$ is given by Equation (\ref{eqn:defn:reflection-matrix}%
). Moreover, the processes $Y,Z$ inherit properties in Equation (\ref%
{eqn:newtork:Q0}) - (\ref{eqn:network:Z_0}) from their pre-limit
counterparts in the queueing model, cf. Equations (\ref{eqn:network:I_cont})
- (\ref{eqn:network:I_k}).

To minimize technical complexity, we restrict attention to stationary Markov
control policies as done in Section \ref{sec:problem}. That is, $\theta
(t)=u(Z(t))$ for $t\geq 0$ for some policy function $u:\mathbb{R}%
_{+}^{d}\rightarrow \Theta .$ Then, defining $c=\left( c_{0},c_{1},\ldots
,c_{K}\right) ^{\top },~h=\left( h_{0},h_{1},\ldots ,h_{K}\right) ^{\top }$
and%
\begin{equation*}
c(z,\theta )=h^{\top }z+c^{\top }\theta ,
\end{equation*}%
as in Equation (\ref{eqn:defn:linear:cost}), the cumulative cost incurred
over the time interval $[0,t]$ under policy $u$ can be written as in
Equation (\ref{def:Cost}). Note that $C^{u}(t)$ and $\xi (t)$ differ only by
a term that is independent of the control. Given $C^{u}(t),$ one can
formulate the discounted control problem as done in Section \ref%
{sec:disc:form}. Similarly, the ergodic control problem can be formulated as
done in Section \ref{sec:ergodic:form}.

\textbf{Interpreting the solution of the drift control problem in the
context of the queueing network formulation.} Because the instantaneous cost
rate $c(z,\theta )$ is linear in the control, inspection of the HJB equation
reveals that the optimal control is of bang-bang nature. That is, $\theta
_{k}(t)\in \{\underline{\theta }_{k},\overline{\theta }_{k}\}$ for all $k,t$
as stated in Equation (\ref{eqn:bang-bang-policy}). This can be interpreted
in the context of the queueing network displayed in Figure \ref{fig:exp1} as
follows: For $k=0,1,\ldots ,K$, whenever $\theta _{k}(t)=\overline{\theta }%
_{k}$, the system manager turns away the thin stream jobs arriving to buffer 
$k$ externally, i.e., she shuts off the renewal process $A_{k}(\cdot )$ at
time $t$. Otherwise, she admits them to the system. Of course, the optimal
policy is determined by the gradient $\nabla V(z)$ of the value function
through the HJB equation, which we solve for using the method described in
Section \ref{sec:han}.

\subsection{Related example with quadratic cost of control}

\label{sec:pricing-example}

\citet{ccelik2008dynamic} and \citet{ata2023approximate} advance
formulations where a system manager controls the arrival rate of customers
to a queueing system by exercising dynamic pricing. One can follow a similar
approach for the feed-forward queueing networks displayed in Figure \ref%
{fig:exp1} with suitable modifications, e.g., the dashed arrows also
correspond to arrivals of regular jobs. This ultimately results in a problem
of drift control for RBM with the cost of control 
\begin{equation}
c(\theta ,z)=\sum_{k=0}^{K}\alpha _{k}(\theta _{k}-\underline{\theta }%
_{k})^{2}+\sum_{k=0}^{K}h_{k}z_{k},  \label{eqn:quadratic:control-cost}
\end{equation}%
where $\underline{\theta }$ is the drift rate vector corresponding to a
nominal price vector.

\subsection{Two parametric families of benchmark policies}

\label{sec:benchmark} Recall the optimal policy can be characterized as 
\begin{equation}
u^{\ast }(z)=\arg \max_{\theta \in \Theta }\left\{ \theta \cdot \nabla V(z)
- c(z,\theta )\right\} , \ z\in \mathbb{R}_{+}^{d}.
\label{eqn:optimal-policy:linear-cost}
\end{equation}

\textbf{The benchmark policy for the main test problem.} In our main test
problem (see Section \ref{subsec:linearcost:feedforwardnetwork}), we have $%
c(z,\theta )=h^{\top }z+c^{\top }\theta$. Therefore, it follows from (\ref%
{eqn:optimal-policy:linear-cost}) that for $k=0,1, \dots,K$, 
\begin{equation*}
u_{k}^{\ast }(z)=\left\{ 
\begin{array}{c}
\overline{\theta }_{k} \\ 
\underline{\theta }_{k}%
\end{array}%
\right. 
\begin{array}{l}
\text{ if }\left( \nabla V(z)\right) _{k}\geq c_{k}, \\ 
\text{ otherwise.}%
\end{array}%
\end{equation*}%
Namely, the optimal policy is of bang-bang type. Therefore, we consider the
following linear-boundary policies as our benchmark polices: For $%
k=0,1,\ldots,K$, 
\begin{equation*}
u_{k}^{\mathrm{lbp}}(z)=\left\{ 
\begin{array}{c}
\overline{\theta }_{k} \\ 
\underline{\theta }_{k}%
\end{array}%
\right. 
\begin{array}{l}
\text{ if }\beta _{k}^{\top }z\geq c_{k}, \\ 
\text{ otherwise,}%
\end{array}%
\end{equation*}%
where $\beta _{0},\beta _{1}\ldots ,\beta _{K} \in \mathbb{R}^{K+1}$ are
vectors of policy parameters to be tuned.

In our numerical study, we primarily focus attention on the symmetric case
where 
\begin{align*}
& h_0 > h_1 = \cdots = h_K, \\
& c_0 = c_1 = \cdots = c_K, \\
& p_1 = \cdots =p_K = \frac{1}{K}, \\
& \underline{\theta}_1 = \cdots = \underline{\theta}_K, \\
& \overline{\theta}_1 = \cdots = \overline{\theta}_K.
\end{align*}
   {The symmetry allows us to limit the number of parameters
needed for the benchmark policy. To be more specific, due} to this symmetry,
the downstream buffers look identical. As such, we restrict attention to
parameter vectors of the following form: 
\begin{eqnarray*}
\beta _{0} &=&\left( \phi _{1},\phi _{2},\ldots ,\phi _{2}\right) ,\text{ and%
} \\
\beta _{i} &=&(\phi _{3},\phi _{4},\ldots \phi _{4},\phi _{5},\phi
_{4},\ldots ,\phi _{4})\text{ where }\phi _{5}\text{ is the }i+1^\text{st} 
\text{ element of } \beta_i \text{ for }i=1,\ldots ,K.
\end{eqnarray*}%
The parameter vector $\beta_0$, which is used to determine the benchmark
policy for buffer zero, has two distinct parameters: $\phi_1$ and $\phi_2$.
In considering the policy for buffer zero, $\phi_1$ captures the effect of
its own queue length, whereas $\phi_2$ captures the effects of the
downstream buffers $1,\ldots,K$. We use a common parameter for the
downstream buffers because they look identical from the perspective of
buffer zero. Similarly, the parameter vector $\beta_i \ (i=1,\ldots,K)$ has
three distinct parameters: $\phi_3, \, \phi_4$ and $\phi_5$, where $\phi_3$
is used as the multiplier for buffer zero (the upstream buffer), $\phi_5$ is
used to capture the effect of buffer $i$ itself and $\phi_4$ is used for all
other downstream buffers. Note that all $\beta_i$ use the same three
parameters $\phi_3, \, \phi_4$ and $\phi_5$ for $i=1,\ldots,K$. They only
differ with respect to the position of $\phi_5$, i.e., it is in the $i+1^%
\text{st}$ position for $\beta_i$.

In summary, the benchmark policy uses five distinct parameters in the
symmetric case. This allows us to do a brute-force search via simulation on
a five-dimensional grid regardless of the number of buffers.

\textbf{The benchmark policy for the test problem with the quadratic cost of
control.} In this case, substituting Equation (\ref%
{eqn:quadratic:control-cost}) into Equation (\ref%
{eqn:optimal-policy:linear-cost}) gives the following characterization of
the optimal policy: 
\begin{equation}
u_{k}^{\ast }(z)= \underline{\theta }_{k} + \frac{(\nabla V(z))_k}{2\alpha_k}%
, \ k=0,1,\ldots, K.  \label{eqn:optimal:policy:quadratic:cost}
\end{equation}%
Namely, the optimal policy is affine in the gradient. Therefore, we consider
the following affine-rate policies as our benchmark polices: For $%
k=0,1,\ldots,K$, 
\begin{equation*}
u_{k}^{\mathrm{arp}}(z)= \underline{\theta }_{k} + \beta _{k}^{\top }z,
\end{equation*}%
where $\beta _{0},\beta _{1}\ldots ,\beta _{K} \in \mathbb{R}^{K+1}$ are
vectors of policy parameters to be tuned. We truncate this at the upper
bound $\overline{\theta}_k$ if needed.

We focus attention on the symmetric case for this problem formulation too.
To be specific, we assume 
\begin{align*}
& h_0 > h_1 = \ldots = h_K, \\
& \alpha_0 = \alpha_1 = \ldots = \alpha_K, \\
& p_1 = \ldots =p_K = \frac{1}{K}, \\
& \underline{\theta}_1 = \ldots = \underline{\theta}_K, \\
& \overline{\theta}_1 = \ldots = \overline{\theta}_K.
\end{align*}
Due to this symmetry, the downstream buffers look identical. As such, we
restrict attention to parameter vectors of the following form: 
\begin{eqnarray*}
\beta _{0} &=&\left( \phi _{1},\phi _{2},\ldots ,\phi _{2}\right) ,\text{ and%
} \\
\beta _{i} &=&(\phi _{3},\phi _{4},\ldots \phi _{4},\phi _{5},\phi
_{4},\ldots ,\phi _{4})\text{ where }\phi _{5}\text{ is the }i+1^\text{st} 
\text{element for } i=1,\ldots ,K.
\end{eqnarray*}%
As done for the first benchmark policy above, this particular form of the
parameter vectors can be justified using the symmetry as well.

\subsection{Parallel-server test problems}

\label{sec:deco_example}

In this section, we consider a problem whose solution can be derived
analytically by considering a one-dimensional problem. To be specific, we
consider the parallel-server network that consists of $K$ identical
single-server queues as displayed in Figure \ref{fig:parallel}. Clearly,
this network can be decomposed into $K$ separate single-server queues,
leading to K separate one-dimensional problem formulations, which can be
solved analytically, see Appendix \ref{appendix:decomposable:test-problem}
for details. For this example we have that $R=I_{d\times d}$ and $A
=I_{d\times d}$. In addition, we assume that the action space $\Theta$ and
the cost function $c(z,\theta)$ are the same as above. 
\begin{figure}[th]
\centering	
\begin{tikzpicture}[start chain=going right,>=latex,node distance=1pt]
		\node[three sided,minimum width=1.2cm,minimum height =0.8cm,on chain]  (wa2)  {$1$};
		
		\node[draw,circle,on chain,minimum size=1cm] (se2) {$\mu$};
		
		\draw[<-] (wa2.west) -- +(-20pt,0) node[left] {$\lambda$}; 
		\draw[->] (se2.east) -- node[above] {}  +(25pt,0);
		
		\node[three sided,minimum width=1.2cm,minimum height =0.8cm,on chain]  (wa22) [below = of wa2,yshift = -1cm] {$K$};
		
		\node[draw,circle,on chain,minimum size=1cm] (se22) {$\mu$};
		
		\draw[<-] (wa22.west) -- +(-20pt,0) node[left] {$\lambda$}; 
		\draw[->] (se22.east) -- node[above] {}  +(25pt,0);
		\path (wa2.south) -- node[auto=false,yshift = 0.1cm]{\vdots} (wa22.north);
	\end{tikzpicture}
\caption{A decomposable parallel-server queueing network.}
\label{fig:parallel}
\end{figure}
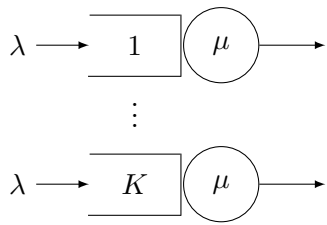

\section{Computational results}

\label{sec:numerical}

For the test problems introduced in Section \ref{sec:example}, we now
compare the performance of policies derived using our method (see Section %
\ref{sec:han}) with the best benchmark we could find. The results show that
our method performs well, and it remains computationally feasible up to at
least dimension $d = 30$. We implement our method using three-layer or
four-layer neural networks with the elu activation function %
\citep{rasamoelina2020review} in Tensorflow 2 \cite{abadi2016tensorflow},
and using code adapted from that of \citet{han2018solving} and %
\citet{zhou2021actorcode}; see Appendix \ref{appendix:implementation:details}
for further details of our implementation.\footnote{%
Our code is available in \url{https://github.com/nian-si/RBMSolver}.} {%
  In our numerical experiments, we observed that the performance
of our method is robust to most hyperparameters. For tuning our algorithm,
the activation function proved to be the most critical element. The 'elu'
activation function resulted in superior performance %
\citet{rasamoelina2020review}. Additionally, we found that decaying the
learning rate to 0.0001 helped achieve good performance.}

For our main test problem with linear cost of control (introduced previously
in Section \ref{subsec:linearcost:feedforwardnetwork}), and also for its
variant with quadratic cost of control (Section \ref{sec:pricing-example}),
the following parameter values are assumed: $h_0=2$, $h_k =1.9$ for $%
k=1,\ldots,K$, $c_k=1$ for $k=0,\ldots,K$, and $p_k =1/K$ for $k=1,\ldots,K$%
. Also, the reflection matrix $R$ and the covariance matrix $A$ for those
families of problems are as follows: 
\begin{equation*}
R=\left[ 
\begin{array}{cccc}
1 &  &  &  \\ 
-1/K & 1 &  &  \\ 
\vdots &  & \ddots &  \\ 
-1/K &  &  & 1%
\end{array}%
\right] \text{ and } A =\left[ 
\begin{array}{ccccc}
1 & 0 & \cdots & \cdots & 0 \\ 
0 & 1 & -\frac{1}{K^{2}} & \cdots & -\frac{1}{K^{2}} \\ 
\vdots & -\frac{1}{K^{2}} & \ddots &  & \vdots \\ 
\vdots & \vdots &  & \ddots & -\frac{1}{K^{2}} \\ 
0 & -\frac{1}{K^{2}} & \cdots & \cdots & 1%
\end{array}%
\right] .
\end{equation*}%
   {We also consider a variation of our main test problem that
has asymmetric routing probabilites.}

As stated previously in Section \ref{sec:deco_example}, the reflection
matrix and covariance matrix for our parallel-server test problems are $R =
I_{d \times d}$ and $A = I_{d \times d}$. Problems in that third class have $%
K=d$ buffers indexed by $k=1,\ldots,K$, and we set $h_1 = 2$ and $h_k = 1.9$
for $k=2,\ldots,K$.

\subsection{Main test problem with linear cost of control}

\label{section:main:test:problem:results}

For our main test problem with linear cost of control (Section \ref%
{subsec:linearcost:feedforwardnetwork}), we take 
\begin{align*}
\underline{\theta}_k = 0 \text{ and } \overline{\theta}_k=b \in \{2, 10 \} 
\text{ for all } k.
\end{align*}
Also, interest rates $r = 0.1$ and $r=0.01$ will be considered in the
discounted case.

To begin, let us consider the simple case where $K=0$ (that is, there are no
downstream buffers in the queueing network interpretation of the problem)
and hence $d=1$. In this case, one can solve the HJB equation analytically;
see Appendix \ref{appendix:decomposable:test-problem} for details. For the
discounted formulation with $r=0.1$, Figure \ref{fig:d=1}~compares the
derivative of the value function computed using the analytical solution,
shown in blue, with the neural network approximation for it that we computed
using our method, shown in red. (The comparisons for $r = 0.01$ and the
ergodic control case are similar.) 
\begin{figure}[!ht]
\centering
\subfigure[$b=2$]{
		\includegraphics[width=3.in]{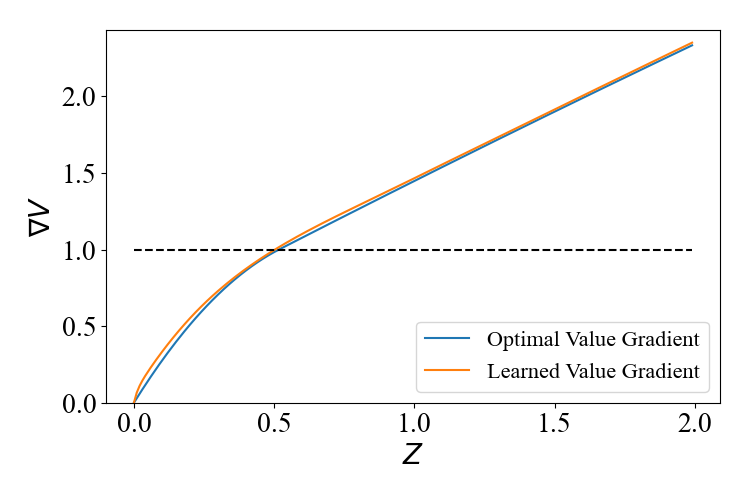}} 
\subfigure[$b=10$]{
		\includegraphics[width=3.in]{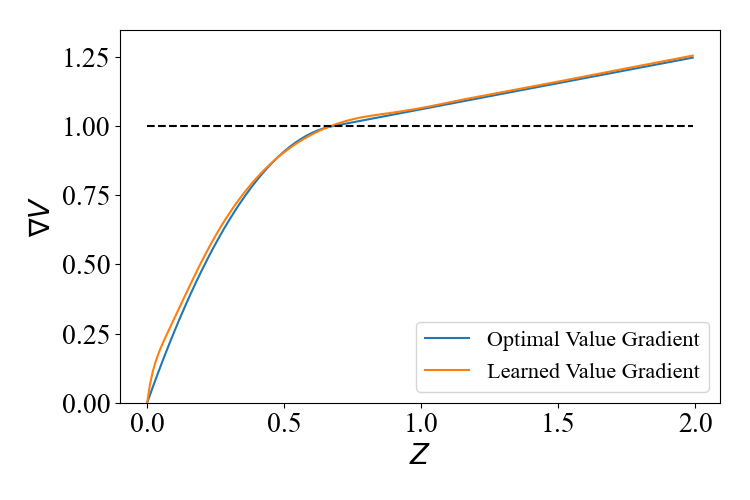}}
\caption{Comparison between the derivative $G_w(\cdot)$ learned from neural
networks and the derivative of the optimal value function for the case of $%
d=1$ and $r=0.1$. The dotted line indicates the cost $c_0=1$. When the value
function gradient is above this dotted line, the optimal control is $\protect%
\theta = b$, and otherwise it is $\protect\theta = 0$.}
\label{fig:d=1}
\end{figure}

Combining Figure \ref{fig:d=1} with Equation (\ref%
{eqn:optimal-policy:linear-cost}), one sees that the policy derived using
our method is close to the optimal policy. Table \ref{tab:d=1} reports the
simulated performance with standard errors of these two policies based on
four million sample paths and using the same discretization of time as in
our computational method. Specifically, we report the long-run average cost
under each policy in the ergodic control case, and report the simulated
value $V(0)$ in the discounted case. To repeat, the benchmark policy in this
case is the optimal policy determined analytically, but not accounting for
the discretization of the time scale. Of course, all the performance figures
reported in the table are subject to simulation errors. Finally, it is worth
noting that our method took less than one hour to compute its policy
recommendations using a 10-CPU core computer. 
\begin{table}[htbp]
\caption{Performance comparison of our proposed policy with the benchmark
policy in the one-dimensional case ($K=0$).}
\label{tab:d=1}\vskip 0.3cm \centering
\begin{tabular}{llccc}
\toprule &  & Ergodic & $r = 0.01$ & $r = 0.1$ \\ 
\midrule \multirow{2}[0]{*}{$b = 2$} & Our policy & 1.455 $\pm$ 0.0006 & 
145.3 $\pm$ 0.05 & 14.29 $\pm$ 0.004 \\ 
& Benchmark & 1.456 $\pm$ 0.0006 & 145.3 $\pm$ 0.05 & 14.29 $\pm$ 0.004 \\ 
\midrule \multirow{2}[0]{*}{$b = 10$} & Our policy & 1.375 $\pm$ 0.0007 & 
137.2 $\pm$ 0.06 & 13.56 $\pm$ 0.005 \\ 
& Benchmark & 1.374 $\pm$ 0.0007 & 137.2 $\pm$ 0.06 & 13.56 $\pm$ 0.005 \\ 
\bottomrule &  &  &  & 
\end{tabular}%
\end{table}

\begin{table}[!htb]
\caption{Performance comparison of our proposed policy with the benchmark
policy in the two-dimensional case ($K=1$).}
\label{tab:d=2}\vskip 0.3cm \centering
\begin{tabular}{llccc}
\toprule &  & Ergodic & $r = 0.01$ & $r = 0.1$ \\ 
\midrule \multirow{2}[0]{*}{$b = 2$} & Our policy & 2.471 $\pm$ 0.0008 & 
246.6 $\pm$ 0.08 & 24.28 $\pm$ 0.006 \\ 
& Benchmark & 2.473 $\pm$ 0.0008 & 246.8 $\pm$ 0.08 & 24.29 $\pm$ 0.006 \\ 
\midrule \multirow{2}[0]{*}{$b = 10$} & Our policy & 2.338 $\pm$ 0.0009 & 
233.3 $\pm$ 0.09 & 23.10 $\pm$ 0.006 \\ 
& Benchmark & 2.338 $\pm$ 0.0009 & 233.6 $\pm$ 0.09 & 23.10 $\pm$ 0.006 \\ 
\bottomrule &  &  &  & 
\end{tabular}%
\end{table}

Let us consider now the two-dimensional case ($K=1$), where the optimal
policy is unknown. Therefore, we compare our method with the best benchmark
we could find: the linear boundary policy described in Section \ref%
{sec:benchmark}. In the two-dimensional case, the linear boundary policy
reduces to the following: 
\begin{equation*}
\theta _{0}(z)=b\mathbb{I}\left\{ \beta _{0}^{\top }z\geq 1\right\} \text{
and } \, \theta _{1}(z)=b\mathbb{I}\left\{ \beta _{1}^{\top }z\geq 1\right\}
.
\end{equation*}%
Through simulation, we perform a brute-force search to identify the best
values of $\beta _{0}$ and $\beta _{1}$. The policies for $b=2$ and $b=10$
are shown in Figures \ref{fig:d=2b=2} and \ref{fig:d=2b=10}, respectively,
for the discounted control case with $r=0.1$. Our proposed policy sets the
drift to $b$ in the red region and to zero in the blue region, whereas the
best-linear boundary policy is represented by the white-dotted line. That
is, the benchmark policy sets the drift to $b$ in the region above and to
the right of the dotted line, and sets it to zero below and left of the
line. Table \ref{tab:d=2} presents the costs with standard errors of the
benchmark policy and our proposed policy obtained in a simulation study. The
two policies have similar performance. Our method takes about one hour to
compute policy recommendations using a 10-CPU core computer. 
\begin{figure}[!ht]
\centering
\subfigure[Server 0]{
		\includegraphics[width=3.in]{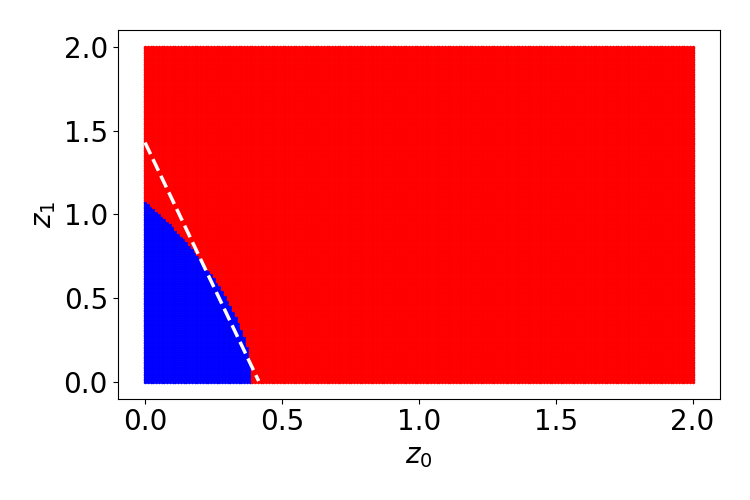}}
\subfigure[Server 1]{
		\includegraphics[width=3.in]{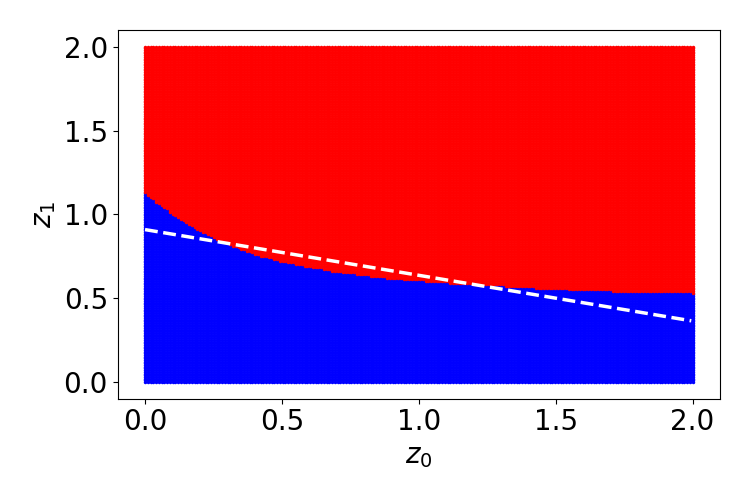}}
\caption{Graphical representation of the policy learned from neural networks
and the benchmark policy for the case $b=2,d=2$ and $r=0.1$}
\label{fig:d=2b=2}
\end{figure}

\begin{figure}[]
\centering
\subfigure[Server 0]{
		\includegraphics[width=3.in]{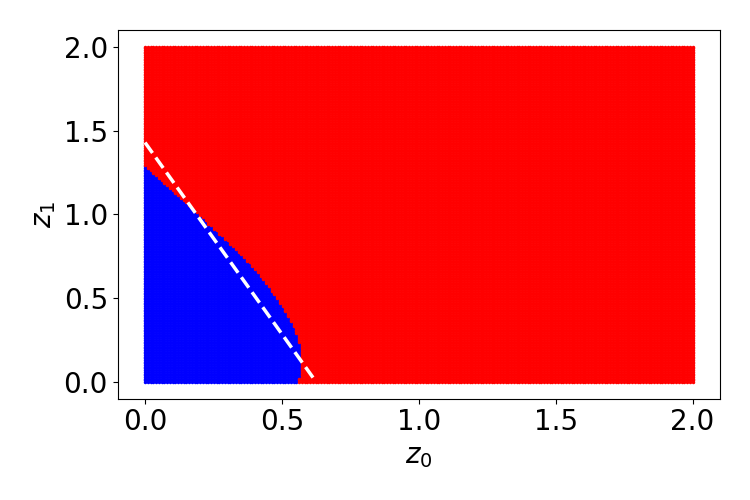}}
\subfigure[Server 1]{
		\includegraphics[width=3.in]{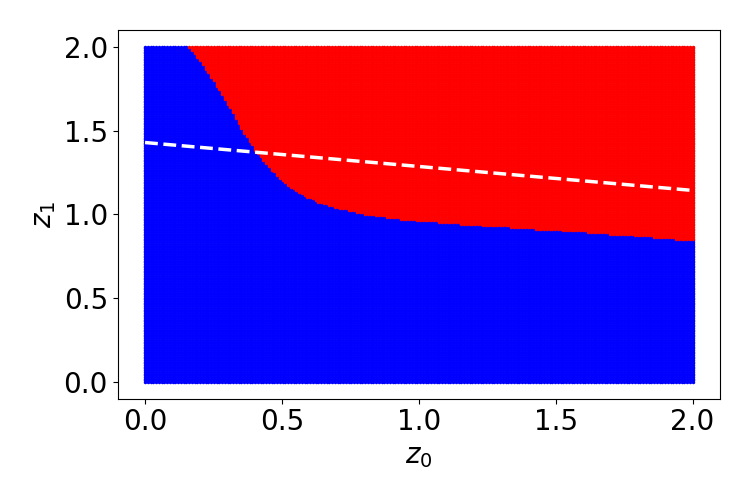}}
\caption{Graphical representation of the policy learned from neural networks
and the benchmark policy for the case $b=10,d=2$ and $r=0.1$}
\label{fig:d=2b=10}
\end{figure}

We then consider the six-dimensional case ($K=5$), where the linear boundary
policy reduces to 
\begin{equation*}
\theta _{i}(z)=b\mathbb{I}\left\{ \beta _{i}^{\top }z\geq 1\right\} \text{
for }i=0,1,2,\ldots ,5.
\end{equation*}%
Although there appear to be 36 parameters to be tuned, recall that we
reduced the number of parameters to five in Section \ref{sec:benchmark} by
exploiting symmetry. This makes the brute-force search computationally
feasible. Table \ref{tab:d=6} compares the performance with standard errors
of our proposed policies with the benchmark policies. They have similar
performance. In this case, the running time for our method is several hours
using a 10-CPU computer. 
\begin{table}[]
\caption{Performance comparison of our proposed policy with the benchmark
policy in the six-dimensional case $d=6$ ($K=5$).}
\label{tab:d=6}\centering
\begin{tabular}{llccc}
\toprule &  & Ergodic & $r = 0.01$ & $r = 0.1$ \\ 
\midrule \multirow{2}[0]{*}{$b = 2$} & Our policy & 7.927 $\pm$ 0.001 & 
791.0 $\pm$ 0.1 & 77.83 $\pm$ 0.01 \\ 
& Benchmark & 7.927 $\pm$ 0.001 & 791.3 $\pm$ 0.1 & 77.83 $\pm$ 0.01 \\ 
\midrule \multirow{2}[0]{*}{$b = 10$} & Our policy & 7.565 $\pm$ 0.0016 & 
754.8 $\pm$ 0.15 & 74.61 $\pm$ 0.01 \\ 
& Benchmark & 7.525 $\pm$ 0.0016 & 751.7 $\pm$ 0.15 & 74.32 $\pm$ 0.01 \\ 
\bottomrule &  &  &  & 
\end{tabular}%
\end{table}
{  To illustrate the scalability of our approach, we next consider
the twenty-one-dimensional case ($K=20$), where $p_1 = \cdots = p_{20} =
0.05 $. Due to symmetry, we can perform a brute-force search for the best
linear boundary policy as done earlier. The results, presented in Table \ref%
{tab:d=21}, demonstrate that our method's performance is comparable to that
of the best benchmark. The run-time for our method is about one day in this
case using a 20-CPU core computer.} 
\begin{table}[!h]
\caption{Performance comparison of our proposed policy with the benchmark
policy in the 21-dimensional case $d=21$ ($K=20$).}
\label{tab:d=21}\centering
\vskip 0.3cm 
\begin{tabular}{llccc}
\toprule &  & Ergodic & $r = 0.01$ & $r = 0.1$ \\ 
\midrule \multirow{2}[0]{*}{$b = 2$} & Our policy & 29.12 $\pm$ 0.0027 & 
2907 $\pm$ 0.26 & 285.8 $\pm$ 0.02 \\ 
& Benchmark & 29.12 $\pm$ 0.0027 & 2907 $\pm$ 0.26 & 285.8 $\pm$ 0.02 \\ 
\midrule \multirow{2}[0]{*}{$b = 10$} & Our policy & 27.78 $\pm$ 0.0031 & 
2773 $\pm$ 0.30 & 273.9 $\pm$ 0.02 \\ 
& Benchmark & 27.60 $\pm$ 0.0029 & 2756 $\pm$ 0.28 & 272.1 $\pm$ 0.02 \\ 
\bottomrule &  &  &  & 
\end{tabular}%
\end{table}

{To further demonstrate the effectiveness of our approach, we consider a six-dimensional test problem with asymmetric routing probabilities. Specifically, for the example shown in Figure \ref{fig:exp1}, we set
\begin{eqnarray*}
p_1=p_2 = 0.3, \ p_3 = 0.2, \ p_4 =p_5 =0.1.
\end{eqnarray*}
All other problem parameters remain the same as in the earlier six-dimensional symmetric test problem; see Appendix \ref{app:heuristic} for its reflection matrix $R$ and the covariance matrix A. However, for the asymmetric problem, tuning its 36 parameters for the linear boundary policy becomes computationally prohibitive. Therefore, we propose an alternative approach in Appendix \ref{app:heuristic} to identify an effective boundary policy. Table \ref{tab:d=6asy} presents the performance, along with standard errors, of our proposed policies compared to benchmark policies. The two policies exhibit similar performance, and the run-time of our method is comparable to that of the earlier symmetric six-dimensional example.
} 
\begin{table}[!h]
\caption{Performance comparison of our proposed policy with the benchmark
policy in the six-dimensional case $d=6$ ($K=5$) with asymmetric routing
probabilities.}
\label{tab:d=6asy}\centering
\begin{tabular}{llccc}
\toprule &  & Ergodic & $r = 0.01$ & $r = 0.1$ \\ 
\midrule \multirow{2}[0]{*}{$b = 2$} & Our policy & 7.938 $\pm$ 0.0013 & 
792.5 $\pm$ 0.13 & 77.98 $\pm$ 0.01 \\ 
& Benchmark & 7.948 $\pm$ 0.0013 & 793.6 $\pm$ 0.13 & 78.11 $\pm$ 0.01 \\ 
\midrule \multirow{2}[0]{*}{$b = 10$} & Our policy & 7.590 $\pm$ 0.0015 & 
757.4 $\pm$ 0.15 & 74.80 $\pm$ 0.01 \\ 
& Benchmark & 7.547 $\pm$ 0.0015 & 753.6 $\pm$ 0.15 & 74.53 $\pm$ 0.01 \\ 
\bottomrule &  &  &  & 
\end{tabular}%
\end{table}

\subsection{Test problems with quadratic cost of control}

\label{section:results:quadratic:cost}

In this section, we consider the test problem introduced in Section \ref%
{sec:pricing-example}, for which we set $\alpha_k =1$ and $\underline{\theta}%
_k =1$ for all $k$. As in the previous treatment of our main test example,
we report results for the cases of $d=1,2,6 $ in Tables \ref{tab:d=1:dp}, %
\ref{tab:d=2:dp}, \ref{tab:d=6:dp}, respectively, where the benchmark
policies are the affine-rate policies discussed in Section \ref%
{sec:benchmark}, with policy parameters optimized via simulation %
   {through a} brute-force search. We observe that our proposed
policies outperform the best affine-rate policies by very small margins in
all cases. 
\begin{table}[!htb]
\caption{Performance comparison of our proposed policy with the benchmark
policy in the case of quadratic cost of control and $d=1$.}
\label{tab:d=1:dp}\centering
\vskip 0.3cm 
\begin{tabular}{llccc}
\toprule & Ergodic & $r = 0.01$ & $r = 0.1$ &  \\ 
\midrule Our policy & 0.757 $\pm$ 0.0004 & 75.53 $\pm$ 0.03 & 7.415 $\pm$
0.003 &  \\ 
Benchmark & 0.758 $\pm$ 0.0004 & 75.67 $\pm$ 0.03 & 7.427 $\pm$ 0.003 &  \\ 
\bottomrule &  &  &  & 
\end{tabular}%
\end{table}
\begin{table}[!htb]
\caption{Performance comparison of our proposed policy with the benchmark
policy in the case of quadratic cost of control and $d=2$ $(K=1)$}
\label{tab:d=2:dp}\vskip 0.3cm \centering
\begin{tabular}{lccc}
\toprule & Ergodic & $r = 0.01$ & $r = 0.1$ \\ 
\midrule Our policy & 1.216 $\pm$ 0.0005 & 121.3 $\pm$ 0.04 & 11.94 $\pm$
0.003 \\ 
Benchmark & 1.219 $\pm$ 0.0005 & 121.7 $\pm$ 0.05 & 11.96 $\pm$ 0.003 \\ 
\bottomrule &  &  & 
\end{tabular}%
\end{table}
\begin{table}[!htb]
\caption{Performance comparison of our proposed policy with the benchmark
policy in the case of quadratic cost of control and $d=6$ $(K=5)$}
\label{tab:d=6:dp}\vskip 0.3cm \centering
\begin{tabular}{lccc}
\toprule & Ergodic & $r = 0.01$ & $r = 0.1$ \\ 
\midrule Our policy & 3.863 $\pm$ 0.0008 & 385.7 $\pm$ 0.08 & 37.92 $\pm$
0.006 \\ 
Benchmark & 3.874 $\pm$ 0.0008 & 386.9 $\pm$ 0.08 & 38.04 $\pm$ 0.006 \\ 
\bottomrule &  &  & 
\end{tabular}%
\end{table}
In the one-dimensional ergodic control case ($K=0$), we obtain analytical
solutions to the RBM control problem in closed form by solving the HJB
equation directly, which reduces to a first-order ordinary differential
equation in this case; see Appendix \ref{appendix:decomposable:test-problem}
for details. Figure \ref{fig:d=1dp} compares the derivative of the optimal
value function (derived in closed form) with its approximation via neural
networks in the ergodic case. Combining Figure \ref{fig:d=1dp} with Equation
(\ref{eqn:optimal:policy:quadratic:cost}) reveals that our proposed policy
is close to the optimal policy. 
\begin{figure}[!ht]
\centering   %
\includegraphics[width=3.1in]{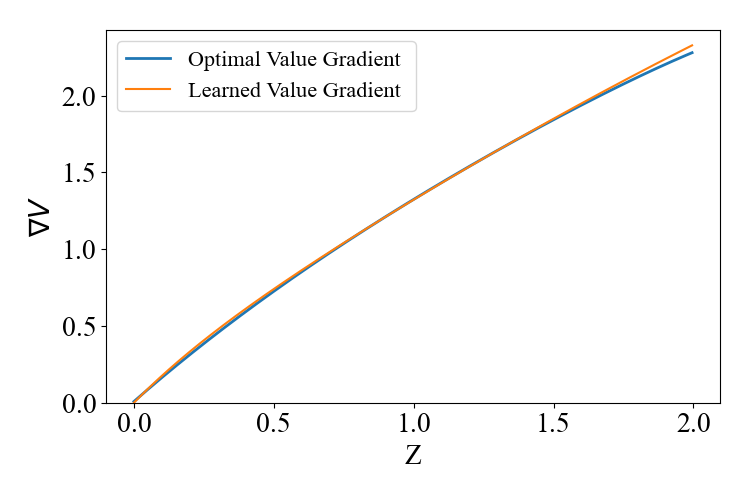}
\caption{Comparison of the gradient approximation $G_w(\cdot)$ learned from
neural networks with the derivative of the optimal value function for the
ergodic control case with quadratic cost of control in the one-dimensional
case ($d=1$).}
\label{fig:d=1dp}
\end{figure}

In the two-dimensional case, our proposed policy is shown in Figure \ref%
{fig:d=2dp} for the ergodic case, with contour lines showing the state
vectors ($z_0, z_1$) for which the policy chooses successively higher drift
rates. The white dotted lines similarly show the states ($z_0,z_1$) for
which our benchmark policy (that is, the best affine-rate policy) chooses
the drift rate $\theta_k =1.5$ (for $k=0$ in the left panel and $k=1$ in the
right panel). 
\begin{figure}[!ht]
\centering
\subfigure[Server 0]{
		\includegraphics[width=3.in]{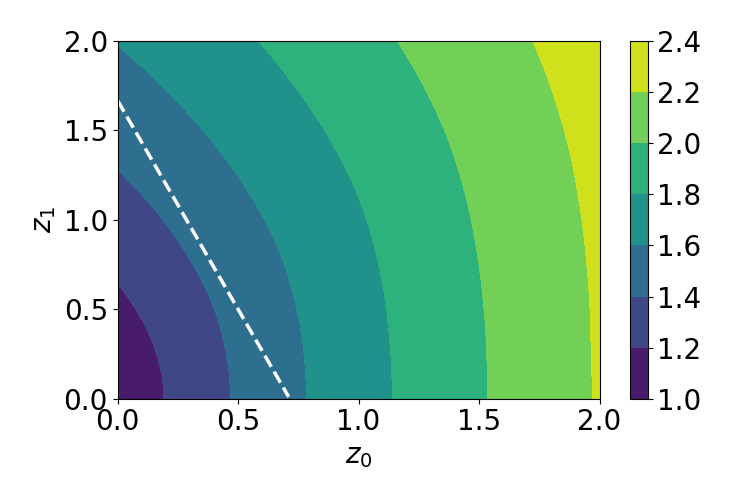}}
\subfigure[Server 1]{
		\includegraphics[width=3.in]{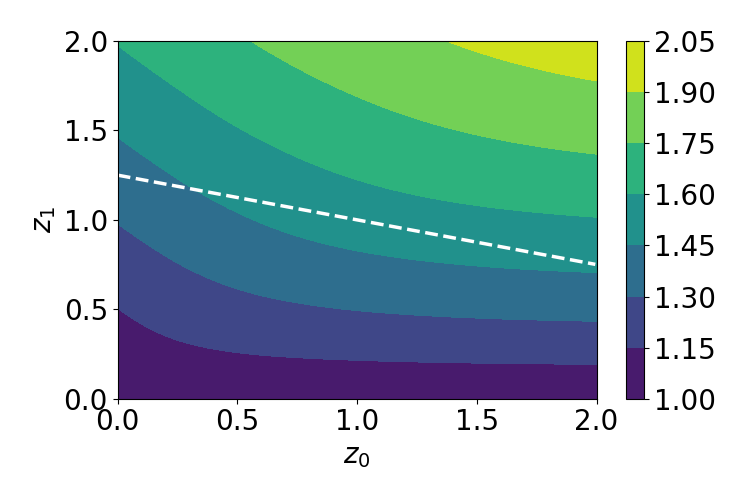}}
\caption{Graphical representation of the policy learned from neural networks
and the benchmark policy for the ergodic case with $d=2$.}
\label{fig:d=2dp}
\end{figure}

\subsection{Parallel-server test problems}

\label{section:results:main:decomposable-test-problem}\label%
{section:results:quadratic:decomposable}

This section focuses on parallel-server test problems (see Section \ref%
{sec:deco_example}) to demonstrate our method's scalability. As illustrated
in Figure \ref{fig:parallel}, the parallel-server networks are essentially $%
K $ independent copies of the one-dimensional case. We present the results
in Table \ref{tab:d=30} for $d=30$ and linear cost of control. When $b=2$,
our policies perform almost equally well as the optimal policy, while for $%
b=10$, our policies perform within 1\% of the optimal policy. The run-time
for our method is about one day in this case using a 20-CPU core computer. 
\begin{table}[tbh]
\caption{Performance comparison between our proposed policy and the
benchmark policy for $30$-dimensional parallel-server test problems with
linear cost of control.}
\label{tab:d=30}\centering
\vskip 0.3cm 
\begin{tabular}{llccc}
\toprule &  & {Ergodic} & {r = 0.01} & {r = 0.1} \\ 
\midrule \multirow{2}[0]{*}{$b = 2$} & Our policy & 42.56 $\pm$ 0.003 & 4247 
$\pm$ 0.3 & 417.3 $\pm$ 0.02 \\ 
& Benchmark & 42.52 $\pm$ 0.003 & 4244 $\pm$ 0.3 & 417.2 $\pm$ 0.02 \\ 
\midrule \multirow{2}[0]{*}{$b = 10$} & Our policy & 40.53 $\pm$ 0.004 & 
4054 $\pm$ 0.4 & 399.4 $\pm$ 0.026 \\ 
& Benchmark & 40.23 $\pm$ 0.004 & 4018 $\pm$ 0.4 & 396.7 $\pm$ 0.024 \\ 
\bottomrule &  &  &  & 
\end{tabular}%
\end{table}

For quadratic cost of control, we are able to solve the test problems up to
at least 100 dimensions. The results for $d=100$ are given in Table \ref%
{tab:d=100:dp}, where the benchmark policies are the best affine-rate
policies (see Section \ref{sec:benchmark}). The performance of our policy is
within 1\% of the benchmark performance. The run-time for our method is
several days in this case using a 30-CPU core computer. 
\begin{table}[ht]
\caption{Performance comparison between our proposed policy and the
benchmark policy for 100-dimensional parallel-server test problems with
quadratic cost of control.}
\label{tab:d=100:dp}\vskip 0.3cm \centering
\begin{tabular}{lccc}
\toprule & Ergodic & $r = 0.01$ & $r = 0.1$ \\ 
\midrule Our policy & 72.74 $\pm$ 0.003 & 7258.3 $\pm$ 0.3 & 712.4 $\pm$ 0.02
\\ 
Benchmark & 72.53 $\pm$ 0.003 & 7237.3 $\pm$ 0.3 & 710.2 $\pm$ 0.02 \\ 
\bottomrule &  &  & 
\end{tabular}%
\end{table}

\section{Concluding remarks}

\label{sec:conclusion} Consider the general drift control problem formulated
in Section \ref{sec:problem}, assuming specifically that the instantaneous
cost rate $c(z,\theta)$ is linear in $\theta$, and further assuming that the
set of available drift vectors is a rectangle $\Theta = [0, b_1] \times
\cdots \times [0,b_d]$. If one relaxes such a problem by letting $b_i
\uparrow \infty$ for one or more $i$, then one obtains what is called a 
\textit{singular} control problem, cf. \citet{kushner-martins1991}. Optimal
policies for such problems typically involve the imposition of \textit{%
endogenous} reflecting barriers, that is, reflecting barriers imposed by the
system controller in order to minimize cost, in addition to exogenous
reflecting barriers that may be imposed to represent physical constraints in
the motivating application.

There are many examples of queueing network control problems whose natural
heavy traffic approximations involve singular control; see, for example, %
\citet{krichagina-taksar1992}, \citet{martins-kushner1990}, and %
\citet{martins-shreve-soner1996}. 
   {In our follow-up paper,
\citet{ata-harrison-si-singular-control2024}, we extend the method developed
in this paper for drift control in a natural way to treat singular control,
and illustrate that extension by means of queueing network applications.}

Separately, the following are three desirable generalizations of the problem
formulations propounded in Section \ref{sec:problem} of this paper. Each of
them is straightforward in principle, and we expect to see these extensions
implemented in future work, perhaps in combination with mild additional
restrictions on problem data. (a) Instead of requiring that the reflection
matrix $R$ have the Minkowski form (\ref{RBM:R}), require only that $R$ be a
completely-$\mathcal{S}$ matrix, which \citet{taylor-williams1993} showed is
a necessary and sufficient condition for an RBM to be well defined. (b)
Allow a more general state space for the controlled process $Z$, such as the
convex polyhedrons characterized by \citet{dai-williams1996}. (c) Remove the
requirement that the action space $\Theta$ be bounded.

Lastly, we have considered in this paper the PDEs that arise in performance analysis and optimal control of RBMs, assuming that those PDEs admit  $C^2$ solutions. \citet{borkar-budhiraja2005} have established the existence and uniqueness of viscosity solutions for such PDEs, extending earlier work by  \citet{dupuis-ishii1991}. To the best of our knowledge, there is no theory currently available concerning existence and uniqueness of classical  $C^2$ solutions, either exact or approximate, and we leave that exploration as a topic for future research.

\newpage 
\begin{appendices}
\section{Proof of Proposition \ref{prop:bounded}}
\label{appendix:proof:prop:bounded}
\begin{proof}
Let $f: \mathbb{R}_+ \rightarrow \mathbb{R}^d$ be right continuous with left limits (rcll). Following \citet{williams1998invariance}, we define the oscillation of $f$ over an interval $[t_1,t_2]$ as follows:
\begin{align}
    Osc(f,[t_1,t_2]) = \sup \left\{ |f(t) - f(s)|: t_1 \leq s < t\leq t_2 \right\},
\end{align}
where $|a| = \max_{i=1.\ldots,d}|a_i|$ for any $a \in \mathbb{R}^d$. Then for two rcll functions $f, g$, the following holds:
\begin{align}
    Osc(f+g) \leq Osc(f) + Osc(g). \label{eqn:aux:oscillation} 
\end{align}
Also recall that the controlled RBM $Z$ satisfies $Z(t) = X(t) + RY(t)$, where
\begin{align}
    X(t) = W(t) - \int_0^t \theta(s)ds, \ t \geq 0.
\end{align}
Then it follows from Theorem 5.1 of \citet{williams1998invariance} that
\begin{align*}
     Osc(Z,[0,t])  & \leq C \, Osc(X,[0,t]) \\
                   & \leq C \, Osc(W,[0,t]) + C \bar{\theta} t
\end{align*}
for some $C>0$, where $\bar{\theta}=\sum_{l=1}^{d}\left( \bar{\theta}_{l}-\underline{\theta }%
_{l}\right) $ and $\underline{\theta }_{l},\bar{\theta}_{l}$ are the minimal
and maximal values on each dimension, and the second inequality follows from (\ref{eqn:aux:oscillation}).

Let $\mathcal{O}=Ocs(W,[0,t])$ and recall that we are interested in bounding the expectation $\mathbb{E}\left[ |Z(t)|^{n}\right]$. To that end, note that 
\begin{eqnarray}
\left\vert Z(t)-Z(0)\right\vert ^{n} &\leq & C^n \left( \mathcal{O}+\bar{\theta}%
t\right) ^{n}  \notag \\
&=& C^{n} \, \sum_{k=0}^{n}\binom{n}{k}\mathcal{O}^{k}\bar{\theta}^{n-k}t^{n-k}.
\label{eqn:prop1:binary}
\end{eqnarray}%
To bound $\mathbb{E}\left[ \mathcal{O}^{k}\right] ,$ note that 
\begin{eqnarray*}
\mathcal{O} &\mathcal{=}&\sup \left\{ |W(t_{2})-W(t_{1})|:0\leq t_{1}<
t_{2}\leq t\right\}  \\
&\leq &\sup \left\{ W(s):0\leq s\leq t\right\} -\inf \left\{ W(s):0\leq
s\leq t\right\}  \\
&\leq &2\sup \left\{ |W(s)|:0\leq s\leq t\right\}  \\
&\leq &2\sup \left\{ \sum_{l=1}^{d}|W_{l}(s)|:0\leq s\leq t\right\}  \\
&\leq &2\sum_{l=1}^{d}\sup \left\{ |W_{l}(s)|:0\leq s\leq t\right\} .
\end{eqnarray*}%
So, by the union bound, we write 
\begin{eqnarray*}
\mathbb{P}\left( \mathcal{O}>x\right)  &\leq &\sum_{l=1}^{d}\mathbb{P}\left(
\sup_{0\leq s\leq t}W_{l}(s)>\frac{x}{2d}\right) +\sum_{l=1}^{d}\mathbb{P}%
\left( \inf_{0\leq s\leq t}W_{l}(s)<-\frac{x}{2d}\right)  \\
&\leq &4\sum_{l=1}^{d}\mathbb{P}\left( W_{l}(t)>\frac{x}{2d}\right) ,
\end{eqnarray*}%
where the last inequality follows from the reflection principle.

Thus, 
\begin{equation*}
\mathbb{E}[\mathcal{O}^{k}]=\int_{0}^{\infty }x^{k-1}\mathbb{P}\left( 
\mathcal{O}>x\right) \mathrm{d}x\leq 4\sum_{l=1}^{d} \int_{0}^{\infty}   x^{k-1}\mathbb{P}\left(
W_{l}(t)>\frac{x}{2d}\right) \mathrm{d}x.
\end{equation*}%
By change of variable $y=x/d,$ we write 
\begin{eqnarray*}
\mathbb{E}[\mathcal{O}^{k}] &\leq &4\sum_{l=1}^{d}\left( 2d\right)
^{k}\int_{0}^{\infty }y^{k-1}\mathbb{P}\left( W_{l}(t)>y\right) \mathrm{d}y
\\
&=&4\sum_{l=1}^{d}\left( 2d\right) ^{k}\mathbb{E}[|W_{l}(t)|^{k}] \\
&=&\frac{4\left( 2d\right) ^{k}2^{k/2}t^{k/2}\Gamma \left( \frac{k+1}{2}%
\right) }{\sqrt{\pi }}\sum_{l=1}^{d}\sigma _{ll}^{k},
\end{eqnarray*}%
where $\Gamma $ is the Gamma function, and the last equality is a well-known
result; see, for example, Equation (12)\ in \citet{winkelbauer2012moments}.
Substituting this into (\ref{eqn:prop1:binary}) gives the following: 
\begin{eqnarray}
\mathbb{E}[|Z(t)-Z(0)|^{n}]  &\leq & C^{n} \sum_{k=0}^{n}\frac{4\left( 2d\right) ^{k}\binom{n}{k}%
2^{k/2}t^{k/2}\Gamma \left( \frac{k+1}{2}\right) }{\sqrt{\pi }}\bar{%
\theta}^{n-k}t^{n-k}\left( \sum_{l=1}^{d}\sigma _{ll}^{k}\right)   \notag \\
&\leq & \tilde{C}_{n} (t^{n}+1).  \label{eqn:prop1:bound2}
\end{eqnarray}
Letting $z=Z(0).$ We write 
\begin{eqnarray*}
|Z(t)|^{n} &=&\left\vert Z(t)-z+z\right\vert ^{n}\leq \left(
|Z(t)-z|+|z|\right) ^{n} \\
&\leq &\sum_{k=0}^{n}\binom{n}{k}\left|Z(t)-z\right| ^{k}|z|^{n-k}.
\end{eqnarray*}%
Using (\ref{eqn:prop1:bound2}), we can therefore write 
\begin{eqnarray*}
\mathbb{E} \left[  \left\vert Z(t)\right\vert ^{n} \right]  &\leq &\sum_{k=0}^{n}\binom{n}{k}%
\tilde{C}_{k}|z|^{n-k}\left( t^{k}+1\right)  \\
&\leq &\hat{C}_{n}(1+t^{n}).
\end{eqnarray*}
\end{proof}

\section{Validity of HJB Equations}
\subsection{Discounted Control}

\label{appendix:verification:discounted}

\begin{proposition}
\label{thm:verification-lemma:discounted-policy} Let $u \in \mathcal{U}$ be
an admissible policy and $V^u$ a $C^2$ solution of the associated PDE (\ref%
{HJB:discount:u:L})-(\ref{HJB:discount:u:D}). If both $V^u$ and its gradient have polynomial growth, then $V^u$ satisfies (\ref%
{problem:discounted}).
\end{proposition}

\begin{proof}
Applying Ito's formula to $e^{-rt}\,V^u(Z^u(t))$ and using Equation (\ref{RBM:df}), we write
\begin{align*}
    e^{-rT}\, V^u(Z^u(T)) - V^u(z) = & \int_0^T e^{-rt} \, \left( \mathcal{L}V^u(Z^u(t)) - u(Z^u(t)) \cdot V^u(Z^u(t)) - rV^u(Z^u(t)) \right) dt  \\
    & + \int_0^T e^{-rt} \mathcal{D}\,V^u (Z^u(t)) \cdot dY^u(t) + \int_0^T e^{-rt} \, \nabla V^u(Z^u(t)) \cdot dW(t).
\end{align*}
Then using (\ref{RBM:Y1})-(\ref{RBM:Y2}) and (\ref{HJB:discount:u:L})-(\ref{HJB:discount:u:D}), we arrive at the following:
\begin{align}
    e^{-rT} \, V^u(Z^{u}(T)) - V^u(z) = & - \int_0^T e^{-rt} \, c(Z^u(t), u(Z^u(t))) \, dt \nonumber \\
    &- \int_0^T e^{-rt} \kappa \cdot dY^u(t) + \int_0^T e^{-rt} \, \nabla V^u(Z^u(t)) \cdot dW(t).
    \label{eqn:auc:verification:lemma:discounted:policy}
\end{align}
Because $\nabla V^u$ has polynomial growth and the action space $\Theta$ is bounded, we have that
\begin{eqnarray*}
    \mathbb{E}_z \left[ \int_0^T e^{-rt} \nabla \, V^u(Z^u(t)) \cdot dW(t) \right] = 0;
\end{eqnarray*}
see, for example, Theorem 3.2.1 of \citet{oksendal2013stochastic}. Thus, taking the expectation of both sides of (\ref{eqn:auc:verification:lemma:discounted:policy}) yields
\begin{align*}
    V^u(z)= \mathbb{E}_z \left[ \int_0^T e^{-rt} \, c(Z^u(t),u(Z^u(t))) \, dt \right] 
    + \mathbb{E}_z \left[ \int_0^T e^{-rt} \, \kappa \cdot dY^u(t) \, dt \right]
    + e^{-rT} \, \mathbb{E}_z \left[ V^u(Z^u(T)) \right].
\end{align*}
Because $V^u$ has polynomial growth and $\Theta$ is bounded, the last term on the right-hand side vanishes by as $T \rightarrow \infty$. As mentioned earlier, because $\Theta$ is bounded by assumption, one can easily derive an affine bound for $\mathbb{E}_z \left[ \kappa \cdot Y^u(T)  \right]$ viewed as a function of $T$. Then, because $c$ has polynomial growth and $\Theta$ is bounded, passing to the limit as $T \rightarrow \infty$ completes the proof. 
\end{proof}

\begin{proposition}
\label{thm:verification-lemma:discounted-optimal} If $V$ is a $C^2$ solution of the HJB
equation (\ref{HJB:discount:min:L})-(\ref{HJB:discount:min:D}), and if both $V$ and its gradient have polynomial growth, then $V$ satisfies (\ref{problem;discount:min}).
\end{proposition}

\begin{proof}
First, consider an arbitrary admissible policy $u$ and let $V^u$ denote the solution of the associated PDE (\ref{HJB:discount:u:L})-(\ref{HJB:discount:u:D}). By Proposition \ref{thm:verification-lemma:discounted-policy}, we have that
\begin{align}
    V^u(z) = \mathbb{E}_z \left[ \int_0^{\infty} e^{-rt} \, c(Z^u(t), u(Z^u(t)) \, dt \right] +\mathbb{E}_z \left[ \int_0^T e^{-rt} \, \kappa \cdot dY^u(t) \right], \ z \in \mathbb{R}_+^d. 
    \label{eqn:verification:lemma:discounted:aux-ineq0}
\end{align}
On the other hand, because $V$ solves (\ref{HJB:discount:min:L})-(\ref{HJB:discount:min:D}) and 
\begin{align*}
    u(z) \cdot \nabla V(z) - c(z, u(z)) \leq \max_{\theta \in \Theta} \left\{ \theta \cdot \nabla V(z) - c(z, \theta \right\}, \  z \in \mathbb{R}_+^d,
\end{align*}
we conclude that
\begin{align}
    \mathcal{L} \, V(z) - u(z) \cdot \nabla V(z) + c(z,u(z)) \geq r \, V(z). \label{eqn:verification:lemma:discounted:aux-ineq1}
\end{align}
Now applying Ito's formula to $e^{-rt} \, V(Z^u(t))$ and using Equation (\ref{RBM:df}) yields
\begin{align*}
    e^{-rT} \, V(Z^u(T)) - V(z) = & \int_0^T \left( \mathcal{L} V(Z^u(t)) - u(Z^u(t)) \cdot \nabla V(Z^u(t)) -rV(Z^u(t)) \right)dt \\
     & + \int_0^T \mathcal{D} \, V(Z^u(t)) \cdot dY^u(t) + \int_0^T e^{-rt} \nabla V(Z^u(t)) \cdot dW(t).
\end{align*}
Combining this with Equations (\ref{RBM:Y1})-(\ref{RBM:Y2}), (\ref{HJB:discount:min:L})-(\ref{HJB:discount:min:D}) and (\ref{eqn:verification:lemma:discounted:aux-ineq1}) gives
\begin{align}
    e^{-rT} \, V(Z^u(t)) - V(z) \geq & - \int_0^T e^{-rt} \, c(Z^u(t), u(Z^u(t))) \, dt \\
    &- \int_0^T e^{-rt} \kappa \cdot dY^u(t) + \int_0^T e^{-rt} \nabla V(Z^u(t)) \cdot dW(t). \label{eqn:verification:lemma:discounted:aux-ineq2}
\end{align}
Because $\nabla V$ has polynomial growth and the action space $\Theta$ is bounded, we have that
\begin{align*}
    \mathbb{E}_z \left[ \int_0^T e^{-rt} \, \nabla V(Z^u(t)) \cdot dW(t) \right] = 0;
\end{align*}
see, for example, Theorem 3.2.1 of \citet{oksendal2013stochastic}. Using this and taking the expectation of both sides of Equation (\ref{eqn:verification:lemma:discounted:aux-ineq2}) yields
\begin{align*}
    V(z) \leq \mathbb{E}_z \left[ \int_0^T e^{-rt} \, c(Z^u(t), u(Z^u(t))) \,dt \right] 
    +\mathbb{E}_z \left[ \int_0^T e^{-rt} \, \kappa \cdot dY^u(t) \right]
    + e^{-rT} \, \mathbb{E} \left[ V(Z^u(T)) \right].
\end{align*}
Because $V$ has polynomial growth and $\Theta$ is bounded, the second term on the right-hand side vanishes as $T \rightarrow \infty$. Then, because $c$ has polynomial growth and $\Theta$ is bounded, passing to the limit yields
\begin{align}
    V(z) \leq \mathbb{E}_z \left[ \int_0^{\infty} e^{-rt} \, c(Z^u(t), u(Z^u(t))) \, dt \right] 
    +\mathbb{E}_z \left[ \int_0^{\infty} e^{-rt} \, \kappa \cdot dY^u(t) \right] = V^u(z), 
    \label{eqn:verification:lemma:discounted:aux-ineq3}
\end{align}
where the equality holds by Equation (\ref{eqn:verification:lemma:discounted:aux-ineq0}). 

Now, consider the optimal policy $u^*$, where $u^*(z) =\arg \max_{\theta \in \Theta} \left\{ \theta \cdot \nabla V(z) - c(z,\theta) \right\}$. For notational brevity, let $Z^* = Z^{u^{*}}$ denote the RBM under policy $u^*$. Note from Equation (\ref{HJB:discount:min:L}) that
\begin{align}
    \mathcal{L} V(z) - u^*(z) \cdot \nabla V(z) + c(z, u^*(z)) = r V(z), \  z \in \mathbb{R}_+^d. 
    \label{eqn:verification:lemma:discounted:aux-PDE}
\end{align}
Repeating the preceding steps with $u^*$ in place of $u$ and replacing the inequality with an equality, cf. Equations (\ref{eqn:verification:lemma:discounted:aux-ineq1}) and (\ref{eqn:verification:lemma:discounted:aux-PDE}), we conclude that 
\begin{align*}
    V(z) = \mathbb{E}_z \left[ \int_0^{\infty} e^{-rt} \, c(Z^*(t),u^*(Z^*(T))) \, dt \right] 
    +\mathbb{E}_z \left[ \int_0^{\infty} e^{-rt} \, \kappa \cdot dY^{u^{*}}(t) \right]
    = V^{u^*}(z).
\end{align*}
Combining this with Equation (\ref{eqn:verification:lemma:discounted:aux-ineq3}) yields (\ref{problem;discount:min}). 
\end{proof}

\subsection{Ergodic Control}

\label{appendix:verification:ergodic}

\begin{proposition}
\label{thm:verification-lemma:ergodic-policy} Let $u \in \mathcal{U}$ be an
admissible policy and $(\tilde{\xi},v^u)$ a $C^2$ solution of the associated
PDE (\ref{HJB:ergodic:L})-(\ref{HJB:ergodic:D}). Further assume that $v^u$ and its gradient have polynomial growth. Then 
\begin{align*}
\tilde{\xi} = \xi^u = \int_{\mathbb{R}_+^d} c(z,u(z)) \, \pi^u(dz) + \sum_{i=1}^{d} \kappa_i \nu_i^u(S_i).
\end{align*}
\end{proposition}

\begin{proof} 
Let $\pi^u$ denote the stationary distribution of RBM under policy $u$, and let $Z^u$ denote the RBM under policy $u$ that is initiated with $\pi^u$. That is,
\begin{align*}
    \mathbb{P}(Z^u(0) \in B) = \pi^u(B)  \ \text{ for } B \subset \mathbb{R}_+^d.
\end{align*}
Then applying Ito's formula to $v^u(Z^u(t))$ and using Equation (\ref{RBM:df}) yields
\begin{align*}
    v^u(Z^u(t)) - v^u(Z^u(0)) = & \int_0^T (\mathcal{L} v^u(Z^u(t)) - u(Z^u(t)) \cdot \nabla v^u(Z^u(t))) \, dt \\
            & + \int_0^T \mathcal{D}v^u(Z^u(t)) \cdot dY^u(t) +  \int_0^T \nabla v^u(Z^u(t)) \cdot dW(t).
\end{align*}
Then using Equations (\ref{RBM:Y1})-(\ref{RBM:Y2}) and (\ref{HJB:ergodic:L})-(\ref{HJB:ergodic:D}), we arrive at the following:
\begin{align}
    v^u(Z^u(T)) - v^u(Z^u(0)) = & \int_0^T \left[ \tilde{\xi} - c(Z^u(t), u(Z^u(t))) \right]dt \nonumber \\
    &- \kappa \cdot Y^u(T) + \int_0^T \nabla v^u(Z^u(t)) \cdot dW(t).
    \label{eqn:ver:lemma:ergodic:policy:auxiliary}
\end{align}
Note that the marginal distribution of $Z^u(t)$ is $\pi^u$ for all $t \geq 0$. Thus, we have that 
\begin{align*}
    \mathbb{E}_{\pi^u}[v^u(Z^u(T))] = \mathbb{E}_{\pi^u}[v^u(Z^u(0))].
\end{align*}
Moreover, using Equation (\ref{eqn:admissibility-cond-ergodic-case}) and the polynomial growth of $\nabla v^u$, we conclude that
\begin{eqnarray*}
    \mathbb{E}_{\pi^u} \left[ \int_0^T \left| \nabla v^u(Z^u(t)) \right|^2\, dt \right] = T \, \int_{\mathbb{R}_+^d} \left| \nabla v^u(z) \right|^2 \, \pi^u(dz) < \infty. 
\end{eqnarray*}
Consequently, we have that $\mathbb{E} \left[ \int_0^T \nabla v^u(Z^u(t)) \cdot dW(t)  \right] = 0 $; see, for example, Theorem 3.2.1 of \citet{oksendal2013stochastic}. Combining these and taking the expectation of both sides of (\ref{eqn:ver:lemma:ergodic:policy:auxiliary}) gives
\begin{align*}
    \tilde{\xi} &= \frac{1}{T} \int_0^T \mathbb{E}_{\pi^u} \left[ c(Z^u(t),u(Z^u(t))) \right] \, dt 
    + \frac{1}{T} \mathbb{E}_{\pi^u} \left[ \kappa \cdot Y^u(T) \right] \\
    &= \int_{\mathbb{R}_+^d} c(z,u(z)) \, \pi_u(dz) + \sum_{i=1}^d \kappa_i \nu_{i}^{u}(S_i) \\ 
    &= \xi^u. 
\end{align*}
\end{proof}

\begin{proposition}
\label{thm:verification-lemma:ergodic-optimal} Let $(v, \xi)$ be a $C^2$ solution of the HJB
equation (\ref{HJB:ergodic:min:L})-(\ref{HJB:ergodic:min:D}), and further assume that both $v$ and its gradient have polynomial growth. Then (\ref{eqn:ergodic:optimal-xi}) holds, and moreover, $\xi = \xi^{u^*}$ where the
optimal policy $u^{*}$ is defined by (\ref{eqn:defn:ergodic:optimal:policy}).
\end{proposition}

\begin{proof} 
First, consider an arbitrary policy $u$ and note that
\begin{eqnarray*}
    \xi^{u} = \int_{\mathbb{R}_+^d} c(z,u(z)) \, \pi^u(dz) + \sum_{i=1}^d \kappa_i \, \nu_i^u(S_i),
\end{eqnarray*}
where $\pi^u$ is the stationary distribution of RBM under policy $u$ and $\nu_i^u$ is the corresponding boundary measure on the boundary surface $S_i = \{ z \in \mathbb{R}_+^d : z_i =0 \}$. Let $Z^u$ denote the RBM under policy $u$ that is initiated with the stationary distribution $\pi^u$. That is,
\begin{align*}
    \mathbb{P}(Z^u(0) \in B) = \pi^u(B), \ \ B \subset \mathbb{R}_+^d.
\end{align*}
On the other hand, because ($v, \xi$) solves the HJB equation and
\begin{align*}
    u(z) \cdot \nabla v(z) - c(z,u(z)) \leq \max_{\theta \, \in \, \Theta} \left\{ \theta \cdot \nabla v(z) - c(z, \theta) \right\},
\end{align*}
we have that
\begin{align}
    \mathcal{L}v(z) - u(z) \cdot \nabla v(z) + c(z, u(z))) \geq \xi 
    \label{eqn:aux:ergodic:PDE-inequality}
\end{align}
Now, we apply Ito's formula to $v(Z^u(t))$ and use Equation (\ref{RBM:df}) to get
\begin{align*}
    v(Z^u(T)) - v(Z^u(0)) = & \int_0^T (\mathcal{L} v(Z^u(t))) - u(Z^u(t)) \cdot \nabla v(Z^u(t)))) \, dt \\
                    & + \int_0^T \nabla v(Z^u(t)) \cdot dY^u(t) + \int_0^T \nabla v(Z^u(t)) \cdot dW(t).
\end{align*}
Combining this with Equations (\ref{RBM:Y1})-(\ref{RBM:Y2}), (\ref{HJB:ergodic:min:D}) and (\ref{eqn:aux:ergodic:PDE-inequality}) gives
\begin{align}
    v(Z^u(T)) - v(Z^u(0)) \geq \int_0^T (\xi - c(Z^u(t), u(Z^u(t))) \, dt - \kappa \cdot Y^u(T) + \int_0^T \nabla v(Z^u(t)) \cdot dW(t). 
    \label{eqn:aux:ergodic:integral-inequality}
\end{align}
Note that the marginal distribution of $Z^u(t)$ is $\pi^u$ for all $t \geq 0$. Thus, we have that 
\begin{align*}
    \mathbb{E}_{\pi^u}[v(Z^u(T))] = \mathbb{E}_{\pi^u}[v(Z^u(0))].
\end{align*}
Moreover, using Equation (\ref{eqn:admissibility-cond-ergodic-case}) and the polynomial growth of $\nabla v$, we conclude that
\begin{eqnarray*}
    \mathbb{E}_{\pi^u} \left[ \int_0^T \left| \nabla v(Z^u(t)) \right|^2\, dt \right] = T \, \int_{\mathbb{R}_+^d} \left| \nabla v(z) \right|^2 \, \pi^u(dz) < \infty. 
\end{eqnarray*}
Consequently, we have that $\mathbb{E} \left[ \int_0^T \nabla v(Z^u(t)) \, dW(t)  \right] = 0 $; see for example Theorem 3.2.1 of \citet{oksendal2013stochastic}. Combining these and taking the expectation of both sides of (\ref{eqn:aux:ergodic:integral-inequality}) give
\begin{align}
    \xi & \leq \frac{1}{T} \int_0^T \mathbb{E}_{\pi^u} \left[ c(Z^u(t),u(Z^u(t)) \right] \, dt 
    + \frac{1}{T} \mathbb{E}_{\pi^u} \left[ \kappa \cdot Y^u(T) \right] \nonumber \\ 
   &  = \int_{\mathbb{R}_+^d} c(z,u(z)) \, \pi_u(dz) + \sum_{i=1}^d \kappa_i \, \nu_i^u(S_i) \nonumber \\
    &=\xi^u. \label{eqn:aux:validity:HJB:ergodic}
\end{align}
Now, consider policy $u^*$. For notational brevity, let $Z^*(t) = Z^{u^*}(t)$ denote the RBM under policy $u^*$ that is initiated with the stationary distribution $\pi^{u^*}$. In addition, note from (\ref{HJB:ergodic:min:L}) that 
\begin{eqnarray}
    \mathcal{L} v(z) - u^*(z) \cdot \nabla v(z) + c(z, u^*(z)) = \xi, \ \ z \in \mathbb{R}_+^d. 
    \label{eqn:ergodic:aux:PDE:optimal-policy}
\end{eqnarray}
Repeating the preceding steps with $u^*$ in place of $u$ and replacing the inequality with an equality, cf. Equations (\ref{eqn:aux:ergodic:PDE-inequality}) and (\ref{eqn:ergodic:aux:PDE:optimal-policy}), we conclude
\begin{align*}
    \xi = \int_{\mathbb{R}_+^d} c(z,u^*(z)) \, \pi_{u^{*}}(dz) + \sum_{i=1}^d \kappa_i \, \nu_i^{u^*}(S_i)  = \xi^{u^*}.
\end{align*}
Combining this with Equation (\ref{eqn:aux:validity:HJB:ergodic}) completes the proof. 
\end{proof}

\section{Derivation of the covariance matrix of the feed-forward examples}

\label{appendix:var} By the functional central limit theorem for the renewal
process \citep{billingsley2013convergence}, we have%
\begin{equation*}
\hat{E}^{n}(\cdot )\Rightarrow W_{E}\left( \cdot \right) ,
\end{equation*}%
where $W_{E}\left( \cdot \right) $ is an one-dimensional Brownian motion
with drift zero and variance $\lambda a^{2}=\mu _{0}a^{2}.$ Furthermore, we
have 
\begin{equation*}
\hat{S}_{k}^{n}(t)\Rightarrow W_{k}\left( \cdot \right) ,\text{ for }%
k=1,2,\ldots ,K,
\end{equation*}%
where $W_{k}\left( \cdot \right) $ is an one-dimensional Brownian motion
with drift zero and variance $\mu _{0}p_{k}s_{k}^{2}.$ Now, we turn to $\hat{%
S}_{0}^{n}(t)$ and $\hat{\Phi}^{n}(t).$ By \citet{harrison1988brownian}, we
have 
\begin{equation*}
\mathrm{Cov}\left( \left[ 
\begin{array}{c}
\hat{S}_{0}^{n}(t) \\ 
\hat{\Phi}^{n}(t)%
\end{array}%
\right] \right) =\mu _{0}\Omega ^{0}+\mu _{0}s_{0}^{2}R^{0}\left(
R^{0}\right) ^{\top },
\end{equation*}%
where $\Omega _{kl}^{0}=p_{k}(\mathbb{I}\{k=l\}-p_{l})$ for,$k,l=0,\ldots ,K$
and $R^{0}=[1,-p_{1},\ldots ,-p_{K}]^{\top }.$ Therefore, we have 
\begin{equation*}
\mathrm{Cov}\left( \left[ 
\begin{array}{c}
\hat{S}_{0}^{n}(t) \\ 
\hat{\Phi}^{n}(t)%
\end{array}%
\right] \right) =\mu _{0}\left[ 
\begin{array}{ccccc}
s_{0}^{2} & -p_{1}s_{0}^{2} & \cdots & \cdots & -p_{K}s_{0}^{2} \\ 
-p_{1}s_{0}^{2} & p_{1}(1-p_{1})+p_{1}^{2}s_{0}^{2} & p_{1}p_{2}\left(
s_{0}^{2}-1\right) & \cdots & p_{1}p_{K}\left( s_{0}^{2}-1\right) \\ 
\vdots & p_{1}p_{2}\left( s_{0}^{2}-1\right) & \ddots &  & \vdots \\ 
\vdots & \vdots &  & \ddots & p_{K-1}p_{K}\left( s_{0}^{2}-1\right) \\ 
-p_{K}s_{0}^{2} & p_{1}p_{K}\left( s_{0}^{2}-1\right) & \cdots & \cdots & 
p_{K}(1-p_{K})+p_{K}^{2}s_{0}^{2}%
\end{array}%
\right] .
\end{equation*}%
Therefore, the variance of $\chi$ is 
\begin{eqnarray*}
A &=&diag(\lambda a^{2},\mu _{1}s_{1}^{2},\ldots ,\mu _{K}s_{K}^{2})+ \\
&&\mu _{0}\left[ 
\begin{array}{ccccc}
s_{0}^{2} & -p_{1}s_{0}^{2} & \cdots & \cdots & -p_{K}s_{0}^{2} \\ 
-p_{1}s_{0}^{2} & p_{1}(1-p_{1})+p_{1}^{2}s_{0}^{2} & p_{1}p_{2}\left(
s_{0}^{2}-1\right) & \cdots & p_{1}p_{K}\left( s_{0}^{2}-1\right) \\ 
\vdots & p_{1}p_{2}\left( s_{0}^{2}-1\right) & \ddots &  & \vdots \\ 
\vdots & \vdots &  & \ddots & p_{K-1}p_{K}\left( s_{0}^{2}-1\right) \\ 
-p_{K}s_{0}^{2} & p_{1}p_{K}\left( s_{0}^{2}-1\right) & \cdots & \cdots & 
p_{K}(1-p_{K})+p_{K}^{2}s_{0}^{2}%
\end{array}%
\right] \\
&=&\mu _{0}\left[ 
\begin{array}{ccccc}
s_{0}^{2}+a^{2} & -p_{1}s_{0}^{2} & \cdots & \cdots & -p_{K}s_{0}^{2} \\ 
-p_{1}s_{0}^{2} & p_{1}(1-p_{1})+p_{1}^{2} s_{0}^{2}+p_{1} s_{1}^{2}& 
p_{1}p_{2}\left( s_{0}^{2}-1\right) & \cdots & p_{1}p_{K}\left(
s_{0}^{2}-1\right) \\ 
\vdots & p_{1}p_{2}\left( s_{0}^{2}-1\right) & \ddots &  & \vdots \\ 
\vdots & \vdots &  & \ddots & p_{K-1}p_{K}\left( s_{0}^{2}-1\right) \\ 
-p_{K}s_{0}^{2} & p_{1}p_{K}\left( s_{0}^{2}-1\right) & \cdots & \cdots & 
p_{K}(1-p_{K})+p_{K}^{2}s_{0}^{2}+p_Ks_{K}^2%
\end{array}%
\right] .
\end{eqnarray*}%
In particular, if the arrival and service processes are Poisson processes,
we have $a=1$ and $s_{k}=1$ for $k=0,1,2,\ldots ,K.$ Then, we have%
\begin{equation*}
A_{\mathrm{Poisson}}=\mu _{0}\left[ 
\begin{array}{ccccc}
2 & -p_{1} & \cdots & \cdots & -p_{K} \\ 
-p_{1} & 2p_{1} &  &  &  \\ 
\vdots &  & \ddots &  &  \\ 
\vdots &  &  & \ddots &  \\ 
-p_{K} &  &  &  & 2p_{K}%
\end{array}%
\right]
\end{equation*}%
Furthermore, if the service time for server zero is deterministic, i.e., $%
s_{0}=0$ , we have 
\begin{equation*}
A_{\mathrm{deterministic}}=\mu _{0}\left[ 
\begin{array}{ccccc}
a^{2} & 0 & \cdots & \cdots & 0 \\ 
0 & p_{1}(1-p_{1})+p_{1}s_{1}^{2} & -p_{1}p_{2} & \cdots & -p_{1}p_{K}
\\ 
\vdots & -p_{1}p_{2} & \ddots &  & \vdots \\ 
\vdots & \vdots &  & \ddots & -p_{K-1}p_{K} \\ 
0 & -p_{1}p_{K} & \cdots & \cdots & p_{K}(1-p_{K})+p_{K}s_{K}^{2}%
\end{array}%
\right].
\end{equation*}

\section{\!\!\! Analytical solution of one-dimensional test problems}

\label{appendix:decomposable:test-problem}

\subsection{Ergodic control formulation with linear cost of control}

We consider the one-dimensional control problem with the cost function%
\begin{equation*}
c(z,\theta )=hz+c\theta \text{ for }z\in \mathbb{R}_{+}\text{ and }\theta
\in \Theta =[0,b].
\end{equation*}
In the ergodic control case, the HJB equation (\ref{HJB:ergodic:min:L}) - (\ref%
{HJB:ergodic:min:D}) is 
\begin{eqnarray}
&&\frac{a}{2}v^{\prime \prime }(z)-\max_{\theta \in \lbrack 0,b]}\left\{
\theta \cdot v^{\prime }(z)-hz-c\theta \right\} =\xi ,\text{ and}
\label{HJB:ergodic:one-dimensional} \\
&&v^{\prime }(0)=0\text{ and }v^{\prime }(z)\text{ having polynomial growth
rate,} \label{HJB:ergodic:one-dimensional2}
\end{eqnarray}%
where the covariance matrix $A=a$ in this one-dimensional case. The HJB
equation (\ref{HJB:ergodic:one-dimensional}) - (\ref{HJB:ergodic:one-dimensional2}) is equivalent to 
\begin{equation*}
\frac{a}{2}v^{\prime \prime }(z)+hz-(v^{\prime }(z)-c)^{+}b=\xi,
\end{equation*}%
and the solution is 
\begin{equation*}
\left( v^{\ast }\right) ^{\prime }(z)=\left\{ 
\begin{array}{c}
\frac{2}{\sqrt{a}}\sqrt{ch+\frac{ah^{2}}{4b^{2}}}z-\frac{h}{a}z^{2} \\ 
\frac{h}{b}z+\frac{ha}{2b^{2}}-\frac{\sqrt{a}}{b}\sqrt{ch+\frac{ah^{2}}{%
4b^{2}}}+c\text{ }%
\end{array}%
\right. 
\begin{array}{c}
\text{if }z<z^{\ast } \\ 
\text{if }z\geq z^{\ast }%
\end{array}%
,\text{ with}
\end{equation*}%
\begin{equation*}
z^*=\frac{1}{h}\sqrt{a\left( ch+\frac{ah^{2}}{4b^{2}}\right) } -\frac{a}{2b}\text{ and } \xi ^{\ast }=\sqrt{a\left( ch+\frac{ah^{2}}{4b^{2}}\right) },
\end{equation*}%
and the optimal control is
\begin{equation*}
\theta ^{\ast }(z)=\left\{ 
\begin{array}{c}
0 \\ 
b%
\end{array}%
\right. 
\begin{array}{c}
\text{if }z<z^{\ast }, \\ 
\text{if }z\geq z^{\ast }.%
\end{array}%
\end{equation*}

\subsection{Discounted formulation with linear cost of control}

The cost function is still 
\begin{equation*}
c(z,\theta )=hz+c\theta \text{ for }z\in \mathbb{R}_{+}\text{ and }\theta
\in \Theta =[0,b],
\end{equation*}
and in the discounted control case, the HJB equation (\ref{HJB:discount:min:L}) -
(\ref{HJB:discount:min:D}) is 
\begin{eqnarray*}
\frac{a}{2}V^{\prime \prime }(z)+hz-(V^{\prime }(z)-c)^{+}b &=&rV(z), \\
V^{\prime }(0) &=&0.
\end{eqnarray*}
The solution is 
\begin{equation*}
V^{\ast }(z)=\left\{ 
\begin{array}{c}
V_{1}(z) \\ 
V_{2}(z)%
\end{array}%
\right. 
\begin{array}{c}
\text{if }z<z^{\ast } \\ 
\text{if }z\geq z^{\ast }%
\end{array}%
,
\end{equation*}%
and the optimal control is
\begin{equation*}
\theta ^{\ast }(z)=\left\{ 
\begin{array}{c}
0 \\ 
b%
\end{array}%
\right. 
\begin{array}{c}
\text{if }z<z^{\ast } \\ 
\text{if }z\geq z^{\ast }%
\end{array}%
,
\end{equation*}%
where 
\begin{equation*}
V_{1}(z)=\frac{h\sqrt{a}e^{-\frac{\sqrt{2}\sqrt{r}z}{\sqrt{a}}}}{\sqrt{2}%
r^{3/2}}+\frac{hz}{r}+C_{1}e^{\frac{\sqrt{2}\sqrt{r}z}{\sqrt{a}}}+C_{1}e^{-%
\frac{\sqrt{2}\sqrt{r}z}{\sqrt{a}}},\text{ and}
\end{equation*}%
\begin{equation*}
V_{2}(z)=\frac{-bh+brc+hrz}{r^{2}}+C _{2}e^{z\left( \frac{b}{a}-\frac{\sqrt{%
b^{2}+2ra}}{a}\right) },
\end{equation*}%
for some parameters $z^{\ast },C _{1},C _{2},$ to be determined later.

\textbf{Case 1:} $h\leq rc.$ Note that if $C _{1}=0,$ then we have 
\begin{equation*}
V_{1}^{\prime }(z)=\frac{h}{r}\left( 1-e^{-\frac{\sqrt{2}\sqrt{r}z}{\sqrt{a}}%
}\right) <\frac{h}{r}\leq c.
\end{equation*}%
Therefore, we have 
\begin{equation*}
V^{\ast }(z)=\frac{h\sqrt{a}e^{-\frac{\sqrt{2}\sqrt{r}z}{\sqrt{a}}}}{\sqrt{2}%
r^{3/2}}+\frac{hz}{r}
\end{equation*}%
for the case $h\leq rc$ and the optimal control is always to set $\theta
^{\ast }(z)=0.$

\textbf{Case 2:} $h>rc.$ We have 
\begin{eqnarray*}
z^{\ast } &=&\frac{a\log \left( \frac{\left( h-rc\right) a}{C _{2}\lambda
\left( \sqrt{b^{2}+2ra}-b\right) }\right) }{b-\sqrt{b^{2}+2ra}}, \\
\text{ }V_{2}^{\prime }(z^{\ast }) &=&c,\text{ and} \\
\text{ }V_{2}^{\prime \prime }(z^{\ast }) &=&\frac{\left( h-rc\right) \left( 
\sqrt{b^{2}+2\lambda a}-b\right) }{ra}.
\end{eqnarray*}%
At point $z^{\ast },$ we must have 
\begin{equation*}
V_{1}^{\prime }(z^{\ast })=V_{2}^{\prime }(z^{\ast })\text{ and }V_{1}^{\prime \prime
}(z^{\ast })=V_{2}^{\prime \prime }(z^{\ast }).
\end{equation*}
Then we can numerically solve for $C _{1}$ and $C_{2}$ using the following
equations:
\begin{eqnarray*}
V_{1}^{\prime }(z^{\ast }) &=&c, \\
V_{1}^{\prime \prime }(z^{\ast }) &=&\frac{\left( h-rc\right) \left( \sqrt{%
b^{2}+2ra}-b\right) }{ra}.
\end{eqnarray*}
Table \ref{tab:numerical_values} presents numerical values of $z^{\ast }$ for
different parameter combinations. 
\begin{table}[tbh]
\caption{The numerical values of $z^{\ast }$ for different parameter combinations (%
$a=c=1$).  }
\label{tab:numerical_values}\centering
\begin{tabular}{llcc}
\toprule &  & $r = 0.01$ & $r = 0.1$ \\ 
\midrule \multirow{2}[0]{*}{$b = 2$} & $h=2$ & 0.501671 & 0.517133 \\ 
& $h=1.9$ & 0.519136 & 0.535753 \\ 
\midrule \multirow{2}[0]{*}{$b = 10$} & $h=2$ & 0.660354 & 0.674135 \\ 
& $h=1.9$ & 0.678797 & 0.693707 \\ 
\bottomrule &  &  & 
\end{tabular}%
\end{table}

\subsection{Ergodic control formulation with quadratic cost of control}

We consider the cost function 
\begin{equation*}
c(\theta ,z)=\alpha (\theta -\underline{\theta })^{2}+hz.
\end{equation*}%
The HJB equation (\ref{HJB:ergodic:min:L}) - (\ref{HJB:ergodic:min:D}) then becomes%
\begin{eqnarray}
&&\frac{a}{2}v^{\prime \prime }(z)-\max_{\theta }\left\{ \theta \cdot
v^{\prime }(z)-hz-\alpha (\theta -\underline{\theta })^{2}\right\} =\xi ,%
\text{ and} \\
&&v^{\prime }(z)=0\text{ and }v^{\prime }(z)\text{ having polynomial growth
rate,}
\end{eqnarray}%
which is equivalent to 
\begin{eqnarray*}
hz+\frac{a}{2}v^{\prime \prime }(z)-\frac{1}{4\alpha }\left( v^{\prime
}(z)\right) ^{2}-\underline{\theta }v^{\prime }(z) &=&\xi , \\
v^{\prime }(0) &=&0.
\end{eqnarray*}%
Let $f(z)=v^{\prime }(z)$ with $f(0)=0.$ Then we have%
\begin{equation*}
\xi =hz+\frac{a}{2}f^{\prime }(z)-\frac{1}{4\alpha }\left( f(z)\right) ^{2}-%
\underline{\theta }f(z),f(0)=0,
\end{equation*}%
which is a Riccati equation. One can solve this equation numerically to find $\xi$ such that $f(\cdot)$ has polynomial growth. For example, if $\alpha =\underline{\theta }=a=1,$ and $h=2$, we have $\xi ^{\ast }=0.8017.$

\section{Implementation Details of Our Method}
\label{appendix:implementation:details} 

\textbf{Neural network architecture.} We used a three or four-layer, fully connected neural network with $20$ - $1000$ neurons in each layer; see Tables \ref{tab:hyper-thin} and \ref{tab:hyper-quad} for details.

\textbf{Common hyperparameters.} Batch size $B=256$; time horizon $T=0.1$,
discretization step-size $0.1/64$; see Tables \ref{tab:hyper-thin} and \ref{tab:hyper-quad} for details.

\textbf{Learning rate.} The learning rate starts from 0.0005, and decays to
0.0003 and 0.0001 with a rate detailed in Tables \ref{tab:hyper-thin} and \ref{tab:hyper-quad}.

\textbf{Optimizer.} We used the Adam optimizer \citep{kingma2014adam}.

\textbf{Reference policy.} The reference policy sets $\tilde{\theta}=1$.

\textbf{Activation function.} We use the 'elu' action function \citep{rasamoelina2020review}.

\textbf{Code.} Our code structure follows from that of \citet{han2018solving} and \citet{zhou2021actorcode}. We implement two major changes: First, we have separated the data generation and training processes to facilitate data reuse. Second, we have conducted the RBM simulation. We have also integrated all the features discussed in this section.

\begin{table}[!htbp]
{\small{ 
  \centering
  \caption{Hyperparameters used in the test problems with linear costs}
    \begin{tabular}{lllllllll}
    \toprule
    \multirow{2}[0]{*}{Hyperparameters} & \multicolumn{2}{l}{1-dimensional} & \multicolumn{2}{l}{2-dimensional} & \multicolumn{2}{l}{6-dimensional} & \multicolumn{2}{l}{30-dimensional} \\
    
          & b=2   & $\phantom{ddd}$b=10  & b=2   & $\phantom{ddd}$b=10  & b=2   & $\phantom{ddd}$b=10  & b=2   & $\phantom{ddd}$b=10 \\
          \midrule
    \#Iterations & \multicolumn{2}{l}{6000} & \multicolumn{2}{l}{6000} & \multicolumn{2}{l}{6000} & \multicolumn{2}{l}{6000} \\
     & & & & \\
    \#Epoches & 13    & $\phantom{ddd}$17    & 15    & $\phantom{ddd}$19    & 23    & $\phantom{ddd}$27    & 41   & $\phantom{ddd}$135 \\
     & & & & \\
    \multirow{3}[0]{*}{Learning rate  scheme} & \multicolumn{2}{l}{0.0005 (0,2000)} & \multicolumn{2}{l}{0.0005 (0,3000)} & \multicolumn{2}{l}{0.0005 (0,3000)} & \multicolumn{2}{l}{0.0005 (0,9500)} \\
          & \multicolumn{2}{l}{0.0003 (2000,4000)} & \multicolumn{2}{l}{0.0003 (3000,6000)} & \multicolumn{2}{l}{0.0003 (3000,6000)} & \multicolumn{2}{l}{0.0003 (9500,22000)} \\
          & \multicolumn{2}{l}{0.0001 (4000,$\infty$)} & \multicolumn{2}{l}{0.0001 (6000,$\infty$)} & \multicolumn{2}{l}{0.0001 (6000,$\infty$)} & \multicolumn{2}{l}{0.0001 (22000,$\infty$)} \\
           & & & & \\
    \#Hidden layers & \multicolumn{2}{l}{4} & \multicolumn{2}{l}{4} & \multicolumn{2}{l}{4} & \multicolumn{2}{l}{3} \\
    & & & & \\
    \#Neurons in each  layer & \multicolumn{2}{l}{50} & \multicolumn{2}{l}{50} & \multicolumn{2}{l}{50} & \multicolumn{2}{l}{300} \\
     & & & & \\
    $\tilde{c}_0$  & \multicolumn{2}{c}{\multirow{2}[0]{*}{\diagbox[width=9em]{\\}{}}}  & 0.4   &$\phantom{ddd}$ 7     & 0.4   &$\phantom{ddd}$ 7     & 0.4   & $\phantom{ddd}$7 \\
    $\tilde{c}_1$  &  \multicolumn{2}{c}{}  & \multicolumn{2}{l}{800} & \multicolumn{2}{l}{2400} & \multicolumn{2}{l}{4800} \\
    \bottomrule
    \end{tabular}%
  \label{tab:hyper-thin}%
}}
\end{table}%

\begin{table}[H]
{\small{
  \centering
  \caption{Hyperparameters used in the test problems with quadratic costs}
    \begin{tabular}{lllll}
    \toprule
    Hyperparameters & 1-dimensional & 2-dimensional & 6-dimensional & 100-dimensional \\
    \midrule
    \#Iterations & 6000  & 6000  & 6000  & 12000 \\
    & & & & \\
    \#Epoches & 12    & 14    & 22    & 110 \\
    & & & & \\
    \multirow{3}[0]{*}{Learning rate scheme} & 0.0005 (0,3000) & 0.0005 (0,3000) & 0.0005 (0,3000) & 0.0005 (0,9500) \\
          & 0.0003 (3000,6000) & 0.0003 (3000,6000) & 0.0003 (3000,6000) & 0.0003 (9500,22000) \\
          & 0.0001 (6000,$\infty$) & 0.0001 (6000,$\infty$) & 0.0001 (6000,$\infty$) & 0.0001 (22000,$\infty$) \\
          & & & & \\
    \#Hidden layers & 3    & 4    & 4    & 3 \\
    & & & & \\
    \#Neurons in each layer & 20     & 50     & 50     & 1000 \\
    \bottomrule
    \end{tabular}%
  \label{tab:hyper-quad}%
  }}
\end{table}%

\subsection{Decay loss in the test example with linear cost of control}
Recall in our main test example with linear cost of control, the cost function is
$$c(z,\theta) = h^\top z + c^\top \theta.$$
In the discounted cost formulation, substituting this cost function into the $F$ function defined in Equation (\ref{eq:Ffunction}) gives the following:
\begin{equation}
    F(\tilde{Z}(t),G_{w_2}(\tilde{Z}(t))) = \tilde{\theta} \cdot x +h^\top z -b\sum_{i=1}^d\max(G_{w_2}(\tilde{Z}(t))_i-c,0) .
    \label{eq:Ffunction2}
\end{equation}
Note that if $G_{w_2}(\tilde{Z}(t)) <c $, we have
$$\frac{\partial F(\tilde{Z}(t),G_{w_2}(\tilde{Z}(t)))  }{w_2} =0,
$$
which suggests that the algorithm may suffer from the gradient vanishing problem \citep{hochreiter1998vanishing}, well-known in the deep learning literature. To overcome this difficulty, we propose an alternative $F$ function
\begin{equation}
    \tilde{F}(\tilde{Z}(t),G_{w_2}(\tilde{Z}(t))) = \tilde{\theta} \cdot x +h^\top z -b\sum_{i=1}^d\max(G_{w_2}(\tilde{Z}(t))_i-c,0) - \tilde{b} \sum_{i=1}^d\min(G_{w_2}(\tilde{Z}(t))_i-c,0)
\end{equation}
where $\tilde{b}$ is a decaying function with respect to the training iteration. Specifically, we propose 
$$\tilde{b} =\left( \tilde{c}_0 - \frac{\textrm{iteration}}{\tilde{c}_1} \right)^+,$$
for some positive constants  $\tilde{c}_0$ and $\tilde{c}_1$.
The specific choices of $\tilde{c}_0$ and $\tilde{c}_1$ are shown in Table \ref{tab:hyper-thin}.

We proceed similarly in the ergodic cost case. 

\subsection{Variance loss function in discounted control}
Let us parametrize the value function as 
$ V_{w_1}(z) = \tilde{V}_{w_1}(z) +\xi.$ Note that $\partial\tilde{V}_{w_1}(z) / \partial z$=$\partial {V}_{w_1}(z) / \partial z$. Therefore, we can rewrite the loss function (\ref{dis:loss:l})
\begin{eqnarray}
\ell (w_{1},w_{2}) &=&\mathbb{E}\left[ \left( e^{-r T} \, (\tilde{V}_{w_{1}}(\tilde{Z}%
(T)) + \xi)-(\tilde{V}_{w_{1}}(\tilde{Z}(0)) + \xi)\right. \right.  \label{dis:loss:lE2} \\
&&\left. \left. -\int_{0}^{T}e^{-rt}G_{w_{2}}(\tilde{Z}(t)) \cdot \mathrm{d}%
W(t)+\int_{0}^{T}e^{-rt}F(\tilde{Z}(t),G_{w_{2}}(\tilde{Z}(t))) \, \mathrm{d}%
t\right) ^{2}\right] ,  \notag
\end{eqnarray}%
By optimizing $\xi$ first, we obtain the following variance loss function:
\begin{eqnarray*}
\tilde{\ell} (w_{1},w_{2}) &=&\mathrm{Var}\left[  e^{-r T} \, \tilde{V}_{w_{1}}(\tilde{Z}%
(T))-\tilde{V}_{w_{1}}(\tilde{Z}(0)) \right.   \\
&&\left. -\int_{0}^{T}e^{-rt}G_{w_{2}}(\tilde{Z}(t)) \cdot \mathrm{d}%
W(t)+\int_{0}^{T}e^{-rt}F(\tilde{Z}(t),G_{w_{2}}(\tilde{Z}(t))) \, \mathrm{d}%
t\ \right]. 
\end{eqnarray*}%
We observe that this trick could accelerate the training speed when $r$ is small. Because when $r>0$ is small, $\xi$ is of the order $O(1/r)$ and $\tilde{V}_{w_{1}}(\cdot),G_{w_{2}}(\cdot)$ are of the order $O(1)$.

\section{A Heuristic Approach to Identify an Effective Linear Boundary Policy in Asymmetric Cases}
\label{app:heuristic}
{ 
As a preliminary step, we first examine the tandem-queues network depicted in Figure \ref{fig:expsingle}, which can be considered a subnetwork of the queuing network shown in Figure \ref{fig:exp1}. 

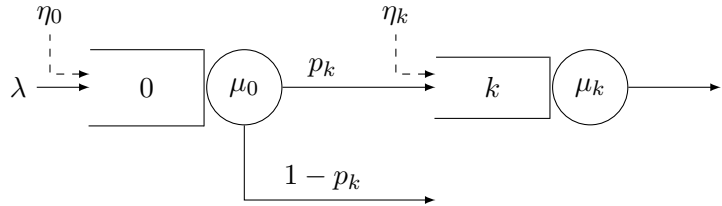
\begin{figure}[!ht]
	\centering
	\begin{tikzpicture}[start chain=going right,>=latex,node distance=1pt]
		\node[three sided,minimum width=1.5cm,minimum height = 1cm,on chain] (wa) {$0$};

		\node[draw,circle,on chain,minimum size=1cm] (se) {$\mu_0$};
		
		%
		
		\node[three sided,minimum width=1.5cm,minimum height = 0.8cm,on chain]  (wamiddle)[right =of 
		se,xshift=2cm]  {$k$};
		
		\node[draw,circle,on chain,minimum size=1cm] (semiddle) {$\mu_k$};


		\draw[->](se)  edge node[above , at end, xshift = -1.5cm] {$p_{k}$} (wamiddle.west);

		\draw[<-,dashed] (wamiddle.west) +(0pt,5pt) -|  +(-15pt,20pt)  node[above , at end = -0.95cm] {$\eta_k$};

		\draw[->] (semiddle.east) -- node[above] {}  +(35pt,0);

		\draw[<-,dashed] (wa.west) +(0pt,5pt) -|  +(-15pt,20pt)  node[above , at end = -0.95cm] {$\eta_0$};
		
		\draw[<-] (wa.west) -- +(-20pt,0) node[left] {$\lambda$};
		
		\node[]  (wa3)[right =of 
		se,xshift=2cm,yshift=-1.5cm]  
		{};

		\draw[->](se) |-  node[above, at end, xshift = -1.5cm] {$1-p_{k}$} (wa3.west);

	\end{tikzpicture}
	\caption{ A subnetwork of the feedforward queueing network with thin arrival streams.}
	\label{fig:expsingle}
\end{figure}

In this system, upon completing their service with server zero, jobs either move on to buffer $k$ with probability $p_k$ or exit the system with probability $1 - p_k$. Therefore, the reflection matrix associated with the $k^{th}$ subnetwork is given as follows:
\begin{equation}
	R=\left[ 
	\begin{array}{cc}
		1 &   0 \\ 
		-p_{k} & 1   \\ 
	\end{array}%
	\right] ,  \label{eqn:defn:reflection-matrix-subsystem}
\end{equation}%
Then, we search for four parameters  $\beta _{0,0}^{(k)},\beta _{0,k}^{(k)},\beta _{k,0}^{(k)},\beta _{k,k}^{(k)}$ for the $k$-th subnetwork  to determine the optimal linear boundary policies for each subnetwork ($k=1,2,\ldots,K$). These policies are represented as:
\begin{equation*}
	\theta^{(k)} _{0}(z)=b\mathbb{I}\left\{ \beta _{0,0}^{(k) }z_0+\beta _{0,k}^{(k) }z_k\geq 1\right\} \text{
		and } \, \theta^{(k)} _{k}(z)=b\mathbb{I}\left\{ \beta _{k,0}^{(k) }z_0+\beta _{k,k}^{(k) }z_k\geq 1 \right\}
	.
\end{equation*}%

Then, in the original  feedforward queueing network, we set the policies $\theta _{k}(z)=b\mathbb{I}\left\{ \beta _{k}^{\top }z\geq 1 \right\}$ for Server $k=1,2,\ldots,K$ as follows:
\begin{equation*}
\beta _{k,0}= \beta _{k,0}^{(k) }, \beta _{k,k}=\beta _{k,k}^{(k) }, \text{ and } \beta _{k,j} =0, \text{ for } j\neq0,k.
\end{equation*}%
The reasoning behind this heuristic policy is based on our observations from symmetric cases, indicating that the influence of the length of queue $j$ on the length of queue $k$ is small when $j$ is not equal to $0$ or $k$.

In order to complete our specification of the heuristic policy, we search the parameters for server zero in $\beta_0$. In the asymmetric case with $p=[0.3,0.3,0.2,0.1,0.1]$, there are four such parameters to tune.}

Lastly, for completeness, we provide the reflection matrix $R$ and the covariance matrix $A$ for our test example with asymmetric routing probabilities below.
$$R=\left[ 
\begin{array}{cccccc}
1 &  &  &  &  &\\ 
-0.3 & 1 &  &  &  & \\ 
-0.3  &  & 1  &  &  &\\ 
-0.2 &  &  & 1 &  &\\ 
-0.1  &  &  &  & 1 &\\
-0.1 &  &  & &  & 1%
\end{array}%
\right], \quad A=\left[ 
\begin{array}{cccccc}
1 & 0 & 0 &  0&0  &0\\ 
0 & 1 & -0.09 & -0.06 &-0.03  & -0.03 \\ 
0  &  -0.09 & 1  & -0.06 & -0.03 &-0.03\\ 
0 &    -0.06&-0.06 & 1 &-0.02  & -0.02\\ 
0  &    -0.03 &  -0.03 & -0.02 & 1 & -0.01\\
0 &   -0.03&  -0.03 & -0.02& -0.01 & 1 %
\end{array}%
\right].$$

\end{appendices}

\bibliographystyle{plainnat}
\bibliography{mybib}
\addcontentsline{toc}{section}{\refname}

\end{document}